\documentclass{article}
\usepackage[nonatbib,preprint]{neurips_2020}

\usepackage{tikz}
\usetikzlibrary{positioning}
\usetikzlibrary{external}
\usetikzlibrary{topaths,calc}
\usepackage[utf8]{inputenc}
\usepackage{amsmath}
\usepackage{amsthm}
\usepackage{times}
\usepackage{mathtools}
\usepackage{amssymb}
\usepackage{booktabs}
\usepackage{url}
\usepackage{hyperref}
\usepackage[ruled,vlined ]{algorithm2e}
\usepackage{bbm}
\hypersetup{colorlinks,          linkcolor=blue,
            citecolor=blue,
            urlcolor=magenta,
            linktocpage,
            plainpages=false}

\newcommand{\R}{\mathbb{R}}

\newcommand{\E}{\mathop{\mathbb{E}}}

\newcommand{\lprod}{\langle}
\newcommand{\rprod}{\rangle}

\renewcommand{\epsilon}{\ensuremath\varepsilon}

\renewcommand{\phi}{\ensuremath{\varphi}}

\theoremstyle{plain}
\newtheorem{theorem}{Theorem}[section]

\newtheorem{lemma}[theorem]{Lemma}

\newtheorem{problem}[theorem]{Problem}

\theoremstyle{definition}
\newtheorem{example}[theorem]{Example}

\newtheorem{definition}[theorem]{Definition}
\newtheorem{algo}[theorem]{Algorithm}
\theoremstyle{remark}

\usepackage{nicefrac}       
\usepackage{microtype}
\usepackage{subcaption}

\title{Estimating Rank-One Spikes from Heavy-Tailed Noise via Self-Avoiding Walks}

\author{%
  Jingqiu Ding\thanks{equal contribution}\\
  ETH Zurich\\
  \texttt{jding@ethz.ch}\\
 \And
 Samuel B. Hopkins\footnotemark[1]\\
 UC Berkeley\\
 \texttt{hopkins@berkeley.edu}\\
 \And
 David Steurer\footnotemark[1]\\
 ETH Zurich\\
 \texttt{dsteurer@inf.ethz.ch}\\
 }

\begin{document}

\maketitle

\begin{abstract}
  We study symmetric spiked matrix models with respect to a general class of noise distributions.
  Given a rank-1 deformation of a random noise matrix, whose entries are independently distributed with zero mean and unit variance, the goal is to estimate the rank-1 part.
  For the case of Gaussian noise, the top eigenvector of the given matrix is a widely-studied estimator known to achieve optimal statistical guarantees, e.g., in the sense of the celebrated BBP phase transition.
  However, this estimator can fail completely for heavy-tailed noise.

 In this work, we exhibit an estimator that works for heavy-tailed noise up to the BBP threshold that is optimal even for Gaussian noise.
 We give a non-asymptotic analysis of our estimator which relies only on the variance of each entry remaining constant as the size of the matrix grows: higher moments may grow arbitrarily fast or even fail to exist.
  Previously, it was only known how to achieve these guarantees if higher-order moments of the noises are bounded by a constant independent of the size of the matrix.

  Our estimator can be evaluated in polynomial time by counting self-avoiding walks via a color coding technique.
  Moreover, we extend our estimator to spiked tensor models and establish analogous results.
\end{abstract}

\section{Introduction}

Principal component analysis (PCA) and other spectral methods are ubiquitous in machine learning.
They are useful for dimensionality reduction, denoising, matrix completion, clustering, data visualization, and much more.
However, spectral methods can break down in the face of egregiously-noisy data: a few unusually large entries of an otherwise well-behaved matrix can have an outsized effect on its eigenvectors and eigenvalues.

In this paper, we revisit the \emph{single-spike recovery problem}, a simple and extensively-studied statistical model for the core task addressed by spectral methods, in the setting of \emph{heavy-tailed noise}, where the above shortcomings of PCA and eigenvector-based methods are readily apparent \cite{johnstone2001distribution}.
We develop and analyze algorithms for this problem whose provable guarantees \emph{improve over traditional eigenvector-based methods}.
Our main problem is:

\begin{problem}[Generalized spiked Wigner model, recovery]\label{defWigner}
  Given a realization of a symmetric random matrix of the form $Y = \lambda xx^\top + W$, where $x \in \R^n$ is an unknown fixed vector with $\|x\| = \sqrt{n}$, $\lambda > 0$, and the upper triangular off-diagonal entries of $W \in \R^{n \times n}$ are independently (but not necessarily identically) distributed with zero-mean and unit variance $\E W_{ij}^2 = 1$, estimate $x$.
\end{problem}

The main question about the spiked Wigner model is: how large should the signal-to-noise ratio $\lambda > 0$ be in order to achieve constant correlation with $x$?
The standard algorithmic approach to solve the spiked Wigner recovery problem is PCA, using the top eigenvector of the matrix $Y$ as an estimator for $x$.
This approach has been extensively studied (e.g. in~\cite{baik2005,pizzo2013}), usually under stronger assumptions on the distribution of the entries of $W$.

Assuming boundedness of $5$ moments, i.e. $\E |W_{ij}|^5 \leq O(1)$, a clear picture has emerged: the problem is information-theoretically impossible for $\lambda \sqrt{n}< 1$, and for $\lambda \sqrt{n}>1 $ the top eigenvector of $Y$ is an optimal estimator for $x$ -- this is the celebrated BBP phase transition~\cite{baik2005,pizzo2013}.
If we weaken the assumption to $\E |W_{ij}|^4 \leq O(1)$, it is well known that $\E \|W\| \leq 2 \sqrt n$, so PCA will estimate $x$ nontrivially when $\lambda\sqrt{n} > 2$.
However, many natural random matrices do not satisfy these conditions --
consider for instance random sparse matrices or matrices with heavy-tailed entries.

Our setting allows for much nastier noise distributions: we assume only that the entries of $W$ have unit variance -- $\E |W_{ij}|^{2.01}$ may grow arbitrarily fast with $n$, or even fail to exist.
Under such weak assumptions, the top eigenvector of $Y$ may be completely uncorrelated with the planted spike $x$, for $\lambda\sqrt{n} = O(1)$.
In this paper, we ask:

\begin{quote}
\emph{Main Question: For which $\lambda > 0$ is recovery possible in the spiked Wigner model via an efficient algorithm under heavy-tailed noise distributions?}
\end{quote}

A natural strategy to deal with heavy-tailed noise is to truncate unusually large entries before performing vanilla PCA. However, truncation-based algorithms can fail dramatically if the distributions of the noise entries are adversarially chosen, as our random matrix model allows.
We provide counterexamples to truncation-based algorithms in Section~\ref{truncation}.

\subsection{Our Contributions}
In this work, we develop and analyze computationally-efficient algorithms based on \emph{self-avoiding walks.}
PCA or eigenvector methods can be thought of as computing a power $Y^\ell$ of the input matrix, for $\ell \rightarrow \infty$.
The polynomial $Y^\ell$ in the entries of $Y$ can be expanded in terms of length-$\ell$ walks in the complete graph on $n$ vertices.
Our algorithms, by contrast, are based on a different degree-$\ell$ polynomial in the entries of $Y$, which can be expanded in terms of length-$\ell$ self-avoiding walks.
We describe the main ideas more thoroughly below, turning for now to our results.

\textbf{Spiked Matrices with Heavy-Tailed Noise:} The first result addresses the main question above, demonstrating that our self-avoiding walk algorithm addresses some of the shortcomings of PCA and eigenvector-based methods for the spiked Wigner recovery problem in the heavy-tailed setting.

\begin{theorem}
\label{thm:intro-matrix}
  For every $\delta > 0$, there is a polynomial-time algorithm such that for every $x \in \R^n$ with $\|x\|_2 = \sqrt{n}$ and $\|x\|_\infty \leq n^{1/2 - \delta}$ and every $n^{1/2}\lambda \geq 1 + \delta$, given $Y = \lambda xx^\top + W$ distributed as in the spiked Wigner model, the algorithm returns $\hat{x}$ such that $\E \langle \hat{x},x \rangle^2 \geq \delta^{O(1)} \cdot \|x\|_2^2 \cdot \E \|\hat{x}\|_2^2$.
\end{theorem}

To interpret the result, we note that even if the entries of $W$ are Gaussian, when $\lambda\sqrt{n} < 1$ no estimator $\hat{x}$ achieves nontrivial correlation with $x$ \cite{Perry_2018}, so the assumption $\lambda\sqrt{n} \geq 1 + \delta$ is the weakest one can hope for.
Furthermore, under this assumption, when $\delta$ is close to $0$, it is information-theoretically impossible to find $\hat{x}$ such that $\langle x,\hat{x} \rangle^2/(\|x\|^2 \|\hat{x}\|^2) \rightarrow 1$.
The guarantee we achieve, that $\hat{x}$ is nontrivially correlated to $x$, is the best one can hope for.
(For the regime $\lambda\sqrt{n} \rightarrow \infty$, our algorithm does achieve correlation going to $1$.
Improving the $\delta^{O(1)}$ term to be quantitatively optimal is an interesting open question.)

\textbf{Spiked Tensors with Heavy-Tailed Noise:} The self-avoiding walk approach to algorithm design is quite flexible, and in particular is not limited to spiked matrices.
We also study an analogous problem for \emph{spiked tensors}.
The single-spike tensor model is the analogue of the spiked Wigner model above, but for the task of recovering information from noisy multi-modal data, which has many applications across machine learning \cite{anandkumar2014tensor,richard2014statistical}.

\begin{theorem}
\label{thm:intro-tensor}\footnote{The theorem is stated for planted vectors sampled from independent zero mean prior distribution;  for fixed planted vector, similar guarantees can be obtained using nearly the same techniques as in the spiked matrix model}
  For every $c > 0,\delta<1$ there is a polynomial-time algorithm with the following guarantees.
  Let $x \in \R^n$ be a random vector with independent, mean-zero entries having $\E x_i^2 = 1$ and $\Gamma=\E x_i^4 \leq n^{o(1)}$.
  Let $\lambda > 0$.
  Let $Y = \lambda \cdot x^{\otimes 3} + W$, where $W \in \R^{n \times n \times n}$ has independent, mean-zero entries with $\E W_{ijk}^2 = 1$.
  Then if $\lambda \geq c n^{-3/4}$, the algorithm finds $\hat{x} \in \R^n$ such that $\E \langle x, \hat{x} \rangle \geq \delta \cdot (\E \|x\|_2^2)^{1/2} \cdot (\E \|\hat{x}\|_2^2)^{1/2}$.
\end{theorem}

Under the additional assumption that all entries in $W$ have bounded $12$-th moments, a slightly modified algorithm finds $\hat{x}$ such that $\langle x, \hat{x} \rangle \geq (1- o(1)) (\E \|x\|_2^2)^{1/2} \cdot (\E \|\hat{x}\|_2^2)^{1/2}$, as shown in  appendix \ref{equivalStrongWeak}.
(We have not made an effort to optimize the constant $12$; some improvement may be possible.)
The results are stated for order-$3$ tensors for simplicity; there is no difficulty in extending them to the higher order case. (See appendix~\ref{HigherOrder}.)

Prior work considers the spiked tensor model only in the case that $W$ has either Gaussian or discrete entries \cite{hopkins2015tensor,Kikuchy,biroli2019iron,hastings2019classical,BGG,bhattiprolu2016sumofsquares,10.1145/3055399.3055417}, whereas our results make much weaker assumptions, in particular allowing the entries of $W$ to be heavy-tailed.
The requirement that $\lambda \geq \Omega(n^{-3/4})$ is widely believed to be necessary for polynomial-time algorithms \cite{hopkins2017power}.
Sub-exponential time algorithms are known recover $x$ successfully for $\lambda \leq n^{-3/4 - \Omega(1)}$ in Gaussian and discrete settings \cite{bhattiprolu2016sumofsquares,Kikuchy,hastings2019classical,10.1145/3055399.3055417} -- we show that a sub-exponential time version of our algorithm achieves many of the same guarantees while still allowing for heavy-tailed noise.
Concretely, we extend Theorem~\ref{thm:intro-tensor} as follows:

\begin{theorem}\label{polyPCA}
  In the same setting as theorem \ref{thm:intro-tensor}, for any $c \geq n^{-1/8}\Gamma^{1/4}\text{polylog}(n)$ and $\delta<1$, there is an $n^{O(1/c^4)}$-time algorithm such that $\E \langle x, \hat{x} \rangle \geq \delta \cdot (\E \|x\|_2^2)^{1/2} \cdot (\E \|\hat{x}\|_2^2)^{1/2}$.
\end{theorem}
In particular, the tradeoff we obtain between running time and signal-to-noise ratio $\lambda$ matches lower bounds in the \emph{low-degree model} for the (easier) case of Gaussian noise~\cite{kunisky2019notes}, for $c \geq n^{-1/8 + o(1)}$.

\textbf{Numerical Experiments:}
We test our algorithms on synthetic data -- random matrices (and tensors) with hundreds of rows and columns -- empirically demonstrating the improvement over vanilla PCA.

\subsection{Our Techniques}
We now offer an overview of the self-avoiding walk technique we use to prove Theorems~\ref{thm:intro-matrix} and~\ref{thm:intro-tensor}.
For this exposition, we focus on the case of spiked matrices (Theorem~\ref{thm:intro-matrix}).

Our techniques are inspired by recent literature on \emph{sparse stochastic block models,} in particular the study of nonbacktracking random walks in sparse random graphs \cite{abbe2017community}.
We remark further below on the relationship with this literature, but note for now that a self-avoiding walk algorithm closely related to the one we present here appeared in \cite{8104074} in the context of the sparse stochastic block model with overlapping communities.
In the present work we give a refined analysis of this algorithm to obtain Theorem~\ref{thm:intro-matrix}, and extend the algorithm to spiked tensors to obtain Theorem~\ref{thm:intro-tensor}.

Recall that given a spiked random matrix $Y = \lambda xx^\top + W$, our goal is to estimate the vector $x$.
For simplicity of exposition, we suppose $x \in \{ \pm 1\}^n$.
To estimate $x$ up to sign, we will in fact aim to estimate each entry of the matrix $xx^\top$.
Our starting point is the observation that any sequence $i_0,i_2,\ldots,i_\ell \in [n]$ without repeated indices (i.e. a length-$\ell$ self-avoiding walk in the complete graph on $[n]$) gives an estimator of $x_{i_0} x_{i_\ell}$ as follows:
\begin{equation}\label{eq:intro-1}
\E_W \prod_{j <\ell} Y_{i_j, i_{j+1}} = \lambda^{\ell} x_{i_1} x_{i_2}^2 \ldots x_{i_{\ell-1}}^2 x_{i_\ell} = \lambda^\ell x_{i_1} x_{i_\ell} \, .
\end{equation}

To aggregate these estimators into a single estimator for $xx^\top$, we relate them to self-avoiding walks in the complete graph on $[n]$. We denote by $\textrm{SAW}_{\ell}(i,j)$ the set of length-$\ell$ self-avoiding walks between $i,j$ on the vertex set $[n]$. Then we associate the polynomial $\prod_{j < \ell} Y_{i_j, i_{j+1}}$ to $\alpha=(i_0,i_2,\ldots,i_{\ell})\in \textrm{SAW}_{\ell}(i,j)$, where $i_0=i, i_{\ell}=j$, and we denote this polynomial as $\chi_\alpha(Y)$.

We define the \emph{self-avoiding walk matrix}:
\begin{definition}[Self-avoiding walk matrix]\label{sawMatrix}
 Let $P(Y)\in\mathbb{R}^{n\times n}$ be given by
 \begin{equation*}
     P_{ij}(Y)=\sum_{\alpha\in \textrm{SAW}_{\ell}(i,j)} \chi_\alpha(Y)
 \end{equation*}
\end{definition}
Our estimator for $x_i x_j$ will simply be $\frac{P_{ij}(Y)}{\lambda^{\ell}\lvert\textrm{SAW}_{\ell}(i,j)\rvert}$.
By \eqref{eq:intro-1}, $\frac{P_{ij}(Y)}{\lambda^\ell\lvert\textrm{SAW}_{\ell}(i,j)\rvert}$ is an unbiased estimator for $x_i x_j$.
The crucial step is to bound the variance of $P_{ij}(Y)$.
Our key insight is: because we average only over self-avoiding walks, $P_{ij}(Y)$ is multilinear in the entries of $W$, so $\E P_{ij}^2(Y)$ can be controlled under only the assumption of unit variance for each entry of $W$.
Our technical analysis shows that $\E P^2_{ij}(Y)$ is small enough to provide a nontrivial estimator of $x_i x_j$ when (a) $\lambda\sqrt{n} \geq 1 + \delta$ and (b) $\ell \geq O_{\delta}(\log n)$, for any $\delta > 0$.

\textbf{Rounding algorithm:}
Once we have $P(Y)$ achieving constant correlation with $xx^\top$, the following theorem, proved in \cite{8104074}, gives a polynomial time algorithm for extracting an estimator $\hat{x}$ for $x$.
\begin{theorem}\label{correlationPreservation}

  Let $Y$ be a symmetric random matrix and $x$ a vector.
  Suppose we have a matrix-valued function $P(Y)$ such that
\begin{equation*}
    \frac{\E\lprod P(Y),xx^\top\rprod}{\left(\E\lVert P(Y)\rVert_F^2 \cdot \|x\|^4\right)^{1/2}}=\delta\,.
\end{equation*}
then with probability $\delta^{O(1)}$, a random unit vector $\hat{x}$ in the span of top-$\delta^{-O(1)}$ eigenvectors of $P(Y)$ achieves $\langle x,\hat{x} \rangle^2 \geq \delta^{O(1)} \|x\|^2$.
\end{theorem}

\textbf{Prior-free estimation for general $x$:}
A significant innovation of our work over prior work such as \cite{8104074} investigating estimators based on self-avoiding walks is that we avoid the assumption of a prior distribution on the planted vector $x$; instead we assume only a mild bound on the $\ell_\infty$ norm of $x$.
While in the setting of Gaussian $W$ one can always assume that $x$ is random by applying a random rotation to the input matrix $Y$ (which preserves $W$ if it is Gaussian), in our setting working with fixed $x$ presents technical challenges.

In the foregoing discussion we assumed $x$ to be $\pm 1$-valued -- to drop this assumption, we must forego \eqref{eq:intro-1} and give up on the hope that each self-avoiding walk from $i$ to $j$ is an unbiased estimator of $x_i x_j$.
Instead, we are able to use the weak $\ell_\infty$ bound to control the bias of an \emph{average} self-avoiding walk as an estimator for $x_i x_j$, and hence control the bias of the estimator $P_{ij}(Y)$.
Compared to \cite{8104074}, which studies the cases of random or $\pm 1$-valued $x$, our calculation of the variance of $P_{ij}(Y)$ is also significantly more intricate, again because we cannot rely on either randomness or $\pm 1$-ness of $x$.

\textbf{Polynomial time via color coding:} The techniques described already yield an algorithm for the spiked Wigner model running in quasipolynomial time $n^{O_\delta(\log n)}$, simply by evaluating all of the self-avoiding walk polynomials.
We use the \emph{color coding} technique of \cite{alon1995color} (previously used in the context of the stochastic block model by \cite{8104074}) to improve the running time to $n^{O_\delta(1)}$.
Briefly, color coding speeds up the computation of the self-avoiding walk estimators $P_{ij}(Y)$ with a clever combination of randomization and dynamic programming.

\textbf{Extension to spiked tensors:}
The tensor analogue of the PCA algorithm for spiked matrices is the \emph{tensor unfolding} method, where an $n \times n \times n$ input tensor $Y = \lambda x^{\otimes 3} + W$ is unfolded to an $n^2 \times n$ matrix, and then the top $n$-dimensional singular vector of this matrix is used to estimate $x$.
This strategy is successful in the case of Gaussian noise, for $\lambda \gg n^{-3/4}$.
To prove Theorem~\ref{thm:intro-tensor} we adapt the self-avoiding walk method above to handle this form of rectangular matrix.
To prove Theorem~\ref{polyPCA}, we combine the self-avoiding walk method with higher-order spectral methods previously used to obtain subexponential time algorithms for the spiked tensor model \cite{Raghavendra2018,bhattiprolu2016sumofsquares}.

\textbf{Relationship to PCA and Non-Backtracking Walks}
To provide some further context for our techniques, it is helpful to observe the following relationship to PCA.
Given a symmetric matrix $Y$, PCA will extract the top eigenvector of $Y$.
Often, this is implemented via the power method -- that is, PCA will (implicitly) compute the matrix $Y^\ell$ for $\ell \approx \log n$.
Notice that the entries of $Y^\ell$ can be expanded as
\[
(Y^\ell)_{ij} = \sum_{k_1,\ldots,k_{\ell-1} \in [n]} Y_{i,k_1} \cdot \prod_{a\leq \ell} Y_{k_a,k_{a+1}} \cdot Y_{k_{\ell-1},j}
\]
which is a sum over all length-$\ell$ walks from $i$ to $j$ in the complete graph.
Our estimator $P(Y)$ can be viewed as removing some problematic (high variance) terms from this sum, leaving only the self-avoiding walks.

This approach is inspired by recent developments in the study of sparse random graphs, where vertices of unusually high degree spoil the spectrum of the adjacency matrix (indeed, this is morally a special case the heavy-tailed noise setting we consider).
In particular, inspired by statistical physics, \emph{nonbacktracking} walks were developed as a technique to learn communities in the stochastic block model \cite{mossel2018proof,Abbe2018,Saade2015SpectralDI,decelle2011asymptotic,krzakala2013spectral,bordenave2015non}.
A $k$-nonbacktracking walk $i_1,\ldots,i_\ell$ does not repeat any indices $i \in [n]$ among any consecutive $k$ steps; as $k$ increases from $0$ to $\ell$ this interpolates between na\"ive PCA and our self-avoiding walk estimator.

The $k$-nonbacktracking algorithm for $k < \ell$ is also a natural approach in the setting we study.
(Our approach corresponds to $k = \ell$.)
Indeed, there are some advantages to choosing $k = O(1)$: the $O(1)$-nonbacktracking-based estimator can be computed much more efficiently than the self-avoiding walk-based estimator.
Furthermore, in numerical experiments we observe that even $1$-step non-backtracking gives performance comparable with fully self-avoiding walks.
However, rigorous analysis of the $O(1)$-non-backtracking walk estimator  in our distribution-independent setting appears to be a major technical challenge -- even establishing rigorous guarantees in the stochastic block model was a major breakthrough \cite{bordenave2015non,mossel2018proof}.
An advantage of our estimator is that it comes with a relatively simple and highly adaptable rigorous analysis.

\subsection{Organization}

In section~\ref{spikedMatrix}, we discuss algorithms for the spiked matrix model, proving Theorem~\ref{thm:intro-matrix} and providing counterexamples to na\"ive truncation-based algorithms.
In section~\ref{sec:matrixExperiments} we discuss results of numerical experiments for spiked random matrices.
In section~\ref{sec:tensor} we describe our algorithm for the spiked tensor model, deferring the analysis to supplementary material.

\section{Algorithms for general spiked matrix model}
\label{spikedMatrix}

\subsection{Guarantee of self-avoiding walk estimator in fixed vector case}

Here we prove Theorem \ref{thm:intro-matrix} by analyzing the  self-avoiding walk estimator.
(Some details are deferred to supplementary material.)
We focus for now on the following main lemma, putting together the proof of Theorem~\ref{thm:intro-matrix} at the end of this section.


\begin{lemma}\label{matrixCorrlation}
In spiked matrix model $Y=\lambda xx^{\top}+W$ with $\lVert x\rVert=\sqrt{n}$ and the upper triangular entries in symmetric matrix $W$ independently sampled with zero mean and unit variance, we assume $\lVert x\rVert_{\infty}^2=n^{1-\Omega(1)}$. Then if $\lambda n^{1/2}=1+\delta=1+\Omega(1)$, setting $\ell=O(\log_{1+\delta}n)$, we have:
 \begin{equation*}
     \frac{\E \lprod P(Y),xx^\top\rprod}{n\left(\E \lVert P(Y)\rVert_F^2\right)^{1/2}}=\delta^{O(1)} \, ,
 \end{equation*}
 where $P(Y)$ is the length-$\ell$ self-avoiding walk matrix (Definition~\ref{sawMatrix}).
\end{lemma}
For Lemma~\ref{matrixCorrlation}, we will repeatedly need the following technical bound, which we prove in Appendix~\ref{proofQuantity}.
\begin{lemma}\label{criticalQuantity}
Let $V\subseteq [n]$, $\lVert x\rVert=\sqrt{n}$ and $t_1,t_2 \in \mathbb{N}$. We define the quantity $S_{t_1,t_2,V}$ as the following:
\begin{equation*}
    S_{t_1,t_2,V}=\mathop{\mathbb{E}}_{(v_1,\ldots,v_{t_1+t_2})\subseteq [n]\setminus V} \left[\prod_{i=1}^{t_1}x_{v_i}^2\prod_{i=t_1+1}^{t_1+t_2}x_{v_i}^4\right]
\end{equation*}
where $(v_1,v_2,\ldots,v_{t_1+t_2})$ is uniformly sampled from all size-$(t_1+t_2)$ ordered subsets of $[n]\setminus V$ (without repeating elements).
Then assuming $\lvert V\rvert,t_1,t_2=O(\log n)$ and $\lVert x\rVert_{\infty}^2=n^{1-\Omega(1)}$, we have 
$S_{t_1,t_2,V}\leq (1+n^{-\Omega(1)})\lVert x\rVert_\infty^{2t_2}$. Further if $t_2=0$, we have $S_{t_1,t_2,V}\geq 1-n^{-\Omega(1)}$. 

\end{lemma}

From the case $t_2=0$, one can easily deduce the following bound on $\E\lprod P(Y),xx^\top\rprod$.
\begin{lemma}\label{nominator}
Under the same setting as lemma \ref{matrixCorrlation}, we have  $\E\lprod P(Y),xx^\top\rprod=(1\pm o(1))\lambda^{\ell}n^{\ell+1}$
\end{lemma}
\begin{proof}
  We have
\begin{align*}
    \mathbb{E}P_{ij}(Y) &=  \sum_{\alpha\in  \textrm{SAW}_{\ell}(i,j)} \prod_{t=1}^{\ell-1} \lambda x_{\alpha_t}x_{\alpha_{t+1}}\\ &= \lambda^{\ell} \frac{(n-2)!}{(n-(\ell-1))!} x_ix_j\E_{\alpha \in \textrm{SAW}_{\ell}(i,j)}\left[ \prod_{t=1}^{\ell-1}x_{\alpha_t}^2\right] \\ & = (1+n^{-\Omega(1)}) 
    \lambda^{\ell} x_ix_j n^{\ell-1}\E_{\alpha \in \textrm{SAW}_{\ell}(i,j)}\left[ \prod_{t=1}^{\ell-1} x_{\alpha_t}^2\right] \, ,
\end{align*}
where the expectation is taken uniformly over $\alpha\in \textrm{SAW}_{\ell}(i,j)$. For simplicity of notation, we denote $\E_{\alpha \in \textrm{SAW}_{\ell}(i,j)}\prod_{t=1}^{\ell-1} x_{\alpha_t}^2$ as $S_{ij}$. Then according to lemma \ref{criticalQuantity}, we have $S_{ij}=1\pm o(1)$.
Therefore we have $\lprod P(Y),xx^\top\rprod=(1\pm o(1)) \lambda^{\ell} n^{\ell+1}$.
\end{proof}
To prove Lemma~\ref{matrixCorrlation}, the remaining task is to bound the second moment $\E \|P(Y)\|_F^2$.
We can expand the second moment in terms of pairs of self-avoiding walks.
For $\alpha,\beta\in\textrm{SAW}_{\ell}(i,j)$ and corresponding polynomials $\chi_\alpha(Y),\chi_\beta(Y)$, there is a close relationship between $\E[\chi_\alpha(Y)\chi_\beta(Y)]$ and the number of shared vertices and edges of $\alpha,\beta$.
Specifically,
\begin{align*}
    \E[\chi_\alpha(Y)\chi_\beta(Y)] &=\E\left[\prod_{(u,v)\in \alpha\cap\beta} Y_{uv}^2\prod_{(u,v)\in \alpha\Delta\beta} Y_{uv}\right]\\
    &=
    \prod_{(u,v)\in \alpha\Delta\beta}\lambda x_{u}x_v\prod_{(u,v)\in \alpha\cap \beta}\left(1+\lambda^2 x_u^2x_v^2\right)\\
    &= \lambda^{2\ell-2k} \prod_{u\in\textrm{deg}(\alpha\Delta\beta,2)} x_u^2
    \prod_{u\in \textrm{deg}(\alpha\Delta\beta,4)}x_u^4
    \prod_{(u,v)\in \alpha\cap\beta} (1+\lambda^2x_u^2x_v^2)
\end{align*}
where $k$ is number of shared edges between $\alpha,\beta$ and $\textrm{deg}(\alpha\Delta\beta,j)$ is the set of vertices with degree $j$ in the graph $\alpha\Delta\beta$. The size of $\textrm{deg}(\alpha\Delta\beta,4)$ is equal to the number of shared vertices which are not incident to any shared edge.
Thus for the analysis of $\E P_{ij}^2(Y)=\sum_{\alpha,\beta\in \textrm{SAW}_{\ell}(i,j)}\E [\chi_\alpha(Y)\chi_\beta(Y)]$, we classify pairs $\alpha,\beta\in \textrm{SAW}_{\ell}(i,j)$ according to numbers of shared edges and vertices between $\alpha,\beta$.  The following graph-theoretic lemma is needed for bounding the number of such pairs in each class; we will prove it in appendix \ref{proofQuantity}.

\begin{lemma}\label{simpleRelation}
Let $\alpha =\alpha_0,\alpha_1,\ldots,\alpha_\ell$ and $\beta = \beta_0,\beta_1,\ldots,\beta_\ell$ be two length-$\ell$ self-avoiding walks in the complete graph on $[n]$, with $\alpha_0 = \beta_0 = i$ and $\alpha_\ell = \beta_\ell = j$. Let $k$ be the number of shared edges between $\alpha,\beta$, $r$ be the number of shared vertices between $\alpha,\beta$ excluding $i,j$, and $s$ be the number of shared vertices which are not $i,j$ and not incident to shared edges.
Further we denote the number of connected components in $\alpha\cap\beta$ not containing $i,j$ as $p$. Then for $\alpha \neq \beta$ we have the relation $p\leq r-s-k$, and for $\alpha = \beta$ we have $p = s= 0$ and $r = k-1$.
\end{lemma}

We note that for self-avoiding walks $\alpha,\beta$, the connected components of $\alpha\cap\beta$ are all self-avoiding walks. A simple corollary of lemma \ref{criticalQuantity} turns out to be helpful, which we prove in appendix \ref{proofQuantity}
\begin{lemma}\label{corollarySegmentExpectation}
Suppose we have $x\in \mathbb{R}^n$ with norm $\sqrt{n}$.
If $V \subseteq [n]$ has $|V| = O(\log n)$ and if we average over size-$h$ directed self-avoiding walks $\xi$ on vertices $[n] \setminus V$, then for $h=O(\log n)$ we have the bounds

\begin{align*}
    \E_{\xi\subseteq [n]\setminus V}  \left[x_{\xi_0}^2 x_{\xi_h}^2 \prod_{(u,v)\in \xi} (1+\lambda^2 x_u^2x_v^2)\right] \leq & (1+n^{-\Omega(1)})\lVert x\rVert_{\infty}^2 (1+\lambda^2 \lVert x\rVert_{\infty}^2)^{h}\\
 \E_{\xi\subseteq [n]\setminus V} \left[x_{\xi_h}^2 \prod_{(u,v)\in \xi} (1+\lambda^2 x_u^2x_v^2)\right]  \leq & (1+n^{-\Omega(1)})(1+\lambda^2 \lVert x\rVert_{\infty}^2)^{h}\\ 
 \E_{\xi\subseteq [n]\setminus V} \left[x_{\xi_0}^2 \prod_{(u,v)\in \xi} (1+\lambda^2 x_u^2x_v^2)\right] \leq &  (1+n^{-\Omega(1)}) (1+\lambda^2 \lVert x\rVert_{\infty}^2)^{h}\\
 \E_{\xi\subseteq [n]\setminus V} \left[\prod_{(u,v)\in \xi} (1+\lambda^2 x_u^2x_v^2)\right] \leq &  (1+n^{-\Omega(1)}) (1+\lambda^2 \lVert x\rVert_{\infty}^2)^{h}
\end{align*}
where $\xi_0$ is the label of starting vertex of $\xi$, and $\xi_{h}$ is the label of the end vertex of $\xi$.
\end{lemma}
These bounds hold since we can expand the product into a sum of monomials and apply lemma \ref{criticalQuantity} for each monomial.

Now for self-avoiding walk pairs $(\alpha,\beta)$ intersecting on a given number of edges and vertices, we bound the correlation of corresponding polynomials and hence the contribution to the variance of $P$.
For simple expressions, we take $\lambda=O(n^{-1/2})$.
\begin{definition}
On the complete graph $K_n$, for pairs of self-avoiding walks $(\alpha,\beta)$ and $(\gamma,\xi)$, we say that $(\alpha,\beta)$ is isomorphic to $(\gamma,\xi)$ if there is a permutation $\pi \, : \, [n] \rightarrow [n]$ fixing $i,j$ such that $\pi(\alpha) = \gamma$ and $\pi(\beta) = \xi$.
We partition all pairs of length-$\ell$ self-avoiding walks between vertices $i,j$ into isomorphism classes.
We denote the set of all isomorphism classes containing pairs length-$\ell$ self-avoiding walks between vertices $i,j$ sharing $r$ vertices and $k$ edges as $\textrm{shape}(k,r,i,j)$.
\end{definition}
We note that $r<k$ is only possible when $r+1=k=\ell$ (that is, the two paths are identical).
\begin{lemma}[Self-avoiding walk polynomial correlation]\label{sawCorrelation}
In the spiked Wigner model $Y=\lambda xx^\top +W$, where $x$ has norm $\sqrt{n}$ and $W$ is symmetric with entries independently sampled with zero mean and unit variance, for any isomorphism class $\mathcal{S}\in \textrm{shape}(k,r,i,j)$ , we have
\begin{equation*}
  \E_{(\alpha,\beta)\sim \mathcal{S}}\E_{W} \left[\chi_\alpha(Y)\chi_\beta(Y)\right]\leq
  \begin{cases*}
  (1+n^{-\Omega(1)}) \lambda^{2\ell-2k}\lVert x\rVert_{\infty}^{2(r-k)}\left(1+\lambda^2 \lVert x\rVert_{\infty}^2\right)^k & if $r\geq k$ \\
   (1+n^{-\Omega(1)}) \left(1+\lambda^2 \lVert x\rVert_{\infty}^2\right)^k    & if $r+1=k=\ell$
  \end{cases*}
\end{equation*}
where $(\alpha,\beta) \sim \mathcal{S}$ is taken uniformly over the isomorphism class $\mathcal{S}$ and $\chi_\alpha(Y)=\prod_{(u,v)\in\alpha} Y_{u,v}$. 
\end{lemma}
\begin{proof}
We first consider the case $r\geq k$, where $\alpha \neq \beta$. For each $\alpha,\beta$ intersecting on $k$ edges and $r$ vertices, we have the following bound:
\begin{equation*}
    \E[\chi_\alpha(Y)\chi_\beta(Y)] = \lambda^{2\ell-2k} \prod_{u\in\textrm{deg}(\alpha\Delta\beta,2)} x_u^2
    \prod_{u\in \textrm{deg}(\alpha\Delta\beta,4)}x_u^4
    \prod_{(u,v)\in \alpha\cap\beta} (1+\lambda^2x_u^2x_v^2)
\end{equation*}
where $\textrm{deg}(\alpha\Delta\beta,j)$ is the set of vertices with degree $j$ in the graph $\alpha\Delta\beta$. 

For any subgraph $G$ of $K_n$, we denote by $V(G)$ the set of vertices incident to edges in $G$.
We denote $\lvert\textrm{deg}(\alpha\Delta\beta,4)\rvert$  as $s$ and the number of shared vertices between $\alpha,\beta$ excluding $i,j$ as $r$. We denote the number of connected components in $\alpha\cap\beta$ not containing $i,j$ as $p$. Then we have the relation $p\leq r-s-k$ for $\alpha\neq \beta$ according to lemma \ref{simpleRelation}.
 We note that $\alpha\cap\beta$ can be decomposed into a set of disjoint self-avoiding walks, which we denote as $\textrm{SAW}(\alpha\cap\beta)$.

Now we take the expectation over $(\alpha,\beta)$ on isomorphism class $\mathcal{S}$. This is equivalent to taking uniform expectation over the labeling of the $2(\ell-1)-r$ vertices $\alpha,\beta$ which are not equal to $i$ or $j$.
Then we have
   \begin{align*}
    & \E_{(v_1,v_2,\ldots v_{2(\ell-1)-r})} \left[\prod_{u\in\textrm{deg}(\alpha\Delta\beta,2)} x_{u}^2 
    \prod_{u\in \textrm{deg}(\alpha\Delta\beta,4)}x_{u}^4  \prod_{\xi \in \textrm{SAW}(\alpha\cap\beta)}\left(\prod_{(u,v)\in \xi}(1+\lambda x_u^2 x_v^2)\right)\right]  \\ 
    \leq  & (1+n^{-\Omega(1)})\lVert x\rVert_{\infty}^{2s} \E_{(v_1,v_2,\ldots v_{2(\ell-1)-r})}\left[\prod_{\substack{u\in   V(\alpha\cap\beta)\\ u\in V(\alpha\Delta\beta)}}x_u^2 \prod_{\xi \in \textrm{SAW}(\alpha\cap\beta)}\left(\prod_{(u,v)\in \xi}(1+\lambda x_u^2 x_v^2)\right) \right]
    \\  \leq  & (1+n^{-\Omega(1)})
    \lVert x\rVert_{\infty}^{2p+2s}\left(1+\lambda^2 \lVert x\rVert_{\infty}^2\right)^k \leq (1+n^{-\Omega(1)}) \lVert x\rVert_{\infty}^{2(r-k)}\left(1+\lambda^2 \lVert x\rVert_{\infty}^2\right)^k
\end{align*} 
where we use lemma \ref{criticalQuantity} in the first inequality, lemma \ref{corollarySegmentExpectation} in the second inequality, and lemma \ref{simpleRelation} in the last inequality.This proves the first claim.

For any isomorphism class $\mathcal{S}\in \textrm{shape}(k,r,i,j)$ with $k=r+1=\ell$, and $(\alpha,\beta)\in \mathcal{S}$,  we have $\alpha=\beta$. In this case we have
\begin{equation*}
    \E_{W} \left[\chi_\alpha(Y)\chi_\beta(Y)\right]=\prod_{(u,v)\in \alpha} (1+\lambda^2 x_u^2x_v^2)=\prod_{i\in [\ell]} \left(1+\lambda^2 x_{v_{i-1}}^2x_{v_{i}}^2\right)
\end{equation*}
By lemma \ref{corollarySegmentExpectation}, taking expectation over the labeling of the $\ell-1$ vertices in $\alpha$ which are not equal to $i,j$, we have
\begin{equation*}
    \E_{v_1,v_2,\ldots,v_{\ell-1}} \prod_{i\in [\ell]} \left(1+\lambda^2 x_{v_{i-1}}^2x_{v_{i}}^2\right) \leq \left(1+n^{-\Omega(1)}\right) \left(1+\lambda \lVert x\rVert_{\infty}^2\right)^{\ell} 
\end{equation*}
This proves the second claim.
\end{proof}

Now we finish the proof of lemma \ref{matrixCorrlation}. 
\begin{proof}[Proof of Lemma \ref{matrixCorrlation}]

We bound the variance of the estimator $P_{ij}(Y)$. As stated above,
\begin{equation}
    \E P^2_{ij}(Y) =
    \sum_{\alpha,\beta\in \textrm{SAW}_{\ell}(i,j)} \lambda^{2\ell-2k} \prod_{u\in\textrm{deg}(\alpha\Delta\beta,2)} x_u^2
    \prod_{u\in \textrm{deg}(\alpha\Delta\beta,4)}x_u^4
    \prod_{(u,v)\in \alpha\cap\beta} (1+\lambda^2x_u^2x_v^2)\label{eq:variance},
\end{equation}
where $k$ is number of shared edges between $\alpha,\beta$ and $\textrm{deg}(\alpha\Delta\beta,j)$ is the set of vertices with degree $j$ in the graph $\alpha\Delta\beta$.


We note that for fixed $i,j,r,k$ there are at most  $n^{2(\ell-1)-r}\ell^{O(r-k)}$ pairs of $\alpha,\beta$. For fixed $k,r$, we apply lemma \ref{sawCorrelation}. For $k<r$ the contribution to summation \ref{eq:variance} is bounded by
\begin{align*}\label{eq:contribution}
  & n^{2(\ell-1)-r} \ell^{O(r-k)}\E_{\mathcal{S}\sim  \textrm{shape}(k,r,i,j)} \left[ \lambda^{2\ell-2k}\lVert x\rVert_{\infty}^{2(r-k)}\left(1+\lambda^2 \lVert x\rVert_{\infty}^2\right)^k\right]\\
 & \qquad =  n^{-2} \cdot n^{2\ell} \cdot \lambda^{2\ell} \cdot n^{-r} \ell^{O(r-k)}\lambda^{-2k}\lVert x\rVert_{\infty}^{2(r-k)}\left(1+\lambda^2 \lVert x\rVert_{\infty}^2\right)^k\\
\end{align*}


where $\mathcal{S}$ is sampled with some distribution over all shapes in $\textrm{shape}(k,r,i,j)$. 

For $k=r+1=\ell$, if we take $\ell=C\log_{\lambda^2 n} n$ with constant $C$ large enough, then the contribution to summation \ref{eq:variance} is bounded by
\begin{equation*}
    n^{\ell-1}  \left(1+\lambda^2 \lVert x\rVert_{\infty}^2\right)^\ell\leq n^{-\Omega(1)} n^{2(\ell-1)}\lambda^{2\ell}
\end{equation*}
Combining all possible $k,r$, we have summation \ref{eq:variance} bounded by 
\begin{equation*}
     n^{2(\ell-1)}\lambda^{2\ell}\left[n^{-\Omega(1)} +\sum_{k=0}^{\ell-1} \left(  n^{-k}\lambda^{-2k}\left(1+\lambda^2 \lVert x\rVert_{\infty}^2\right)^k\right)\sum_{r=k}^{\ell-1} \left(\ell^{O(r-k)}\lVert x\rVert_{\infty}^{2(r-k)}n^{  k-r}\right) \right]
\end{equation*}
Since $\lambda\lVert x\rVert_{\infty}=n^{-\Omega(1)}$, we have $n^{-1}\lambda^{-2}\left(1+\lambda^2 \lVert x\rVert_{\infty}^2\right)^k\leq \frac{1}{1-\delta/2}$. Thus $\sum_{k=0}^{\ell-1} \left(  n^{-k}\lambda^{-2k}\left(1+\lambda^2 \lVert x\rVert_{\infty}^2\right)^k\right)\leq \delta^{-O(1)}$. On the other hand,  since $\ell^{O(1)}\lVert x\rVert_{\infty}^2n^{-1}=n^{-\Omega(1)}$ by the assumption on $\lVert x\rVert_{\infty}$, we have $\sum_{r=k}^{\ell-1} \left(\ell^{O(r-k)}\lVert x\rVert_{\infty}^{2(r-k)}n^{  k-r}\right)\leq 1+n^{-\Omega(1)}$. Thus summation \ref{eq:variance} is bounded by $\delta^{-O(1)}\lambda^{2\ell}n^{2\ell-2}$.

Summing over $n^2$ pairs of $i,j$ we have 
\begin{equation*}
    \E \lVert P(Y)\rVert_F^2\leq \delta^{-O(1)}\left(n^{2\ell} \lambda^{2\ell}\right)
\end{equation*}


Combining with lemma \ref{nominator} we have 
$\frac{\E\lprod P(Y),xx^\top\rprod}{n\left(\E\lVert P(Y)\rVert_F^2\right)^{1/2}}=\delta^{O(1)}=\Omega(1)$ and the lemma follows.
\end{proof}

Finally, using color-coding method, the degree $O(\log n)$ polynomial $P(Y)$ can be well approximated in polynomial time, which we prove in appendix~\ref{sawEvaluation}
\begin{lemma}[Formally stated in appendix~\ref{sawEvaluation}]\label{sawEva}
For $\delta=\lambda n^{1/2}-1>0$ and  $\ell=O(\log_{1+\delta} n)$, $P(Y)$ can be accurately evaluated in $n^{\delta^{-O(1)}}$ time.
\end{lemma}

The evaluation algorithm \ref{algoSelfAvoidingWalkMain} is based on the idea of color-coding method\cite{alon1995color}. Similar algorithm has already appeared and analyzed in the literature \cite{8104074}.  

\begin{algorithm}
\KwData{Given $Y\in \mathbb{R}^{n\times n}$ s.t $Y=\lambda xx^\top+W$}
\KwResult{$P(Y)\in\mathbb{R}^{n\times n}$ where $P_{ij}(Y)$ is the sum of multilinear monomials corresponding to length   $\ell$ self-avoiding walk between $i,j$(up to accuracy $1+n^{-\Omega(1)}$)}
 $C\gets \exp(100\ell)$\;
 \For{$i\gets1$ \KwTo $C$}{
     Sample coloring $c_t:[n]\mapsto [\ell]$ uniformly at random\;
   Construct a $\mathbb{R}^{2^{\ell}n\times 2^{\ell} n}$ matrix $M$, with rows and columns indexed by $(v,S)$, where $v\in [n]$  and $S$ is a subset of $[\ell]$\;
    a matrix $H\in \mathbb{R}^{n\times  2^{\ell} n}$ with rows indexed by $[n]$ and columns indexed by $(v,S)$ where $v\in [n]$  and $S$ is a subset of $[\ell]$\;
    
    a matrix $N\in \mathbb{R}^{2^\ell n\times n}$, with rows  indexed by $(v,S)$ where $S$ is a subset of $[\ell]$ and columns indexed by $[n]$\;
    
    Record matrix $p_{c_i}= HM^{\ell-2}N$\;}
Return  $\sum _{i=1}^C p_{c_i}/C$
 \caption{Algorithm for evaluating self-avoiding walk matrix }\label{algoSelfAvoidingWalkMain}
 \end{algorithm}

We describe how to construct matrices $H,M,N$ used in the algorithm \ref{algoSelfAvoidingWalkMain} given coloring $c:[n]\mapsto [\ell]$. 
 For matrix $M$, the entry $M_{(v_1,S),(v_2,T)}=Y_{v_1,v_2}$ if $S\cup \{c(v_1)\}=T$ and $c(v_1)\not\in S$. Otherwise $M_{(v_1,S),(v_2,T)}=0$. For matrix $H$, the entry $H_{v_1,(v_2,S)}=Y_{v_1,v_2}$ if $S=\{c(v_1)\}$. Otherwise $H_{v_1,(v_2,S)}=0$. For matrix $N$, the entry $N_{(v_1,S),v_2}=Y_{v_1,v_2}$ if $c(v_1)\not\in  S$ and $S\cup \{c(v_1),c(v_2)\}=[\ell]$. Otherwise $N_{(v_1,S),v_2}=0$.
 
 The critical observation is that for coloring $c:[n]\mapsto [\ell]$ sampled uniformly at random, we have $\E_c p_{c}(Y)=P(Y)$. By averaging over lots of such random colorings, we have an unbiased estimator with low variance. The proof is deferred to appendix \ref{sawEvaluation}.


Combining theorem \ref{correlationPreservation}, lemma \ref{matrixCorrlation},\ref{sawEva}, we have theorem \ref{thm:intro-matrix}.

\subsection{Guarantee and failure of Truncation algorithm}\label{truncation}
 In this section we show that while truncating entries  at threshold $\tau(n)$ can help on many occasions, it can fail for some noise distributions  we consider.
 




The class of truncation algorithm we consider can be described as following: 
\begin{algo}\label{truncationAlgorithm}
Given matrix $Y\in \mathbb{R}^{n\times n}$, set truncation threshold $\tau=\tau(n)$.
We first obtain $Y^{\prime}$ by truncating the entries $Y_{ij}$ with magnitude larger than $\tau$ to $\textrm{sgn}(Y_{ij})\tau$.
  Then, we obtain $Y''$ by subtracting the average value of all entries in $Y'$.
  Finally we extract the top eigenvector of $Y^{\prime\prime}$.
 \end{algo}

First we show that for many long tail distributions, PCA algorithm can be saved by such truncation. 
(We defer the proof to appendix \ref{truncation-proof-section}.)

 
  

\begin{theorem}\label{truncationProof}
  Consider problem \ref{defWigner} such that the signal-to-noise ratio satisfies $\epsilon=n^{1/2}\lambda-1=\Omega(1)$, the upper triangular entries of $W$ are identically distributed, the entries of $X=\lambda xx^\top$ are bounded by $o(1)$, and the entries of $x$ sum to $0$.
  Then, for $\tau=\frac{100}{\textrm{min}(\epsilon^2,1)}$, the algorithm \ref{truncationAlgorithm} outputs unit norm estimator $\hat{x}\in \mathbb{R}^n$ s.t $\lprod x,\hat{x}\rprod^2 =\Omega(n)$:
\end{theorem}

However, as illustrated by the following examples, this truncation strategy can fail inherently when the noise entries are not identically distributed and their distributions are adversarily chosen (depending on the vector $x$). We show that there is no choice of truncation level $\tau$ for which Algorithm~\ref{truncationAlgorithm} outputs a vector whose correlation with the planted vector $x$ is nonvanishing
 for all choices of noise matrix $W$ whose entries are independently sampled with zero mean and unit variance.

 First truncating at $\tau=\Omega(\sqrt{n})$ fails the following example
\begin{example}\label{ egErdosRenyi}
For $d=\omega(1)$,
$W_{ij}$ equals to $-\sqrt{\frac{n-d}{d}}$ with probability $d/n$ and $\sqrt{\frac{d}{n-d}}$ with probability $1-\frac{d}{n}$. \footnote{There is trivial algorithm for  this specific noise distribution, but it  breaks down easily for other noise distribution included in the class we consider.}
\end{example}
 This is just normalized and centralized
 adjacency matrix of Erdos-Renyi random graph, the spectrum of which is well studied in the literature \cite{benaychgeorges2017spectral, Montanari:2016:SPS:2897518.2897548}. For superconstant $d$, the spectral norm is of order $\omega(\sqrt{n})$, much larger than the spectral norm of $\lambda xx^\top$.  
 Therefore, the leading eigenvector will not be correlated with hidden vector $x$ as we desire.
 
Then we only need to consider  $\tau=o(\sqrt{n})$. For simplicity, we analyze an alternative strategy where  entries $Y_{ij}>\tau$ are truncated to $0$. Similar results for truncation to $\tau \textrm{sgn}(Y_{ij})$ are in the appendix \ref{failureExample}. We consider the example below
\begin{example}\label{egMixed}
For $i+j$ even, we let $W_{ij}$ sampled as in example \ref{ egErdosRenyi}. For $i+j$ odd, we let $-W_{ij}$ distributed the same as above.  
\end{example}

For $d=o(n/\tau^2)$, only entries perturbed by noise $\pm \sqrt{\frac{d}{n-d}}$ are preserved. Then $Y^{\prime}_{ij}=\lambda x_ix_j+\sqrt{\frac{d}{n-d}}(-1)^{i+j}$ with probability $1-\frac{d}{n}$ and $0$ with probability $d/n$. Therefore the leading eigenvector of $Y^{
\prime}$ will be well correlated with $h$ rather than $x$.Since $Y^{\prime}$ has zero mean, $Y^{\prime\prime}-Y^{\prime}$ has small Frobenius norm, thus the leading eigenvector of $Y^{\prime\prime}$ is close to $Y^{\prime}$. 





\subsection{Experiments}
\label{sec:matrixExperiments}
For comparing the performance of algorithms proposed, we conduct experiments with several typical distributions of noise: (1) the noise is distributed as example \ref{ egErdosRenyi}. (2) the noise is distributed as example \ref{egMixed} (3) entry $W_{ij}$ is distributed as $N(0,1)$ when $i+j$ is even and as example \ref{ egErdosRenyi} when $i+j$ is odd. In each case planted vector $x$ is randomly sampled from  $N(0,\mathrm{Id}_n)$. In these examples, the smaller parameter $d$ corresponds to the more heavy tailed noise distribution.

In experiments with size $n=10^2 - 10^3$, self-avoiding walk estimator shows better performance than naive PCA and truncation PCA algorithm.
Furthermore, the non-backtracking algorithm achieves performance no worse than self-avoiding walk estimator under many settings.
The results are shown in figure~\ref{fig:Experiment}.  

\begin{figure}
    \centering
    \begin{subfigure}[b]{0.45\linewidth}
        \includegraphics[width=\linewidth]{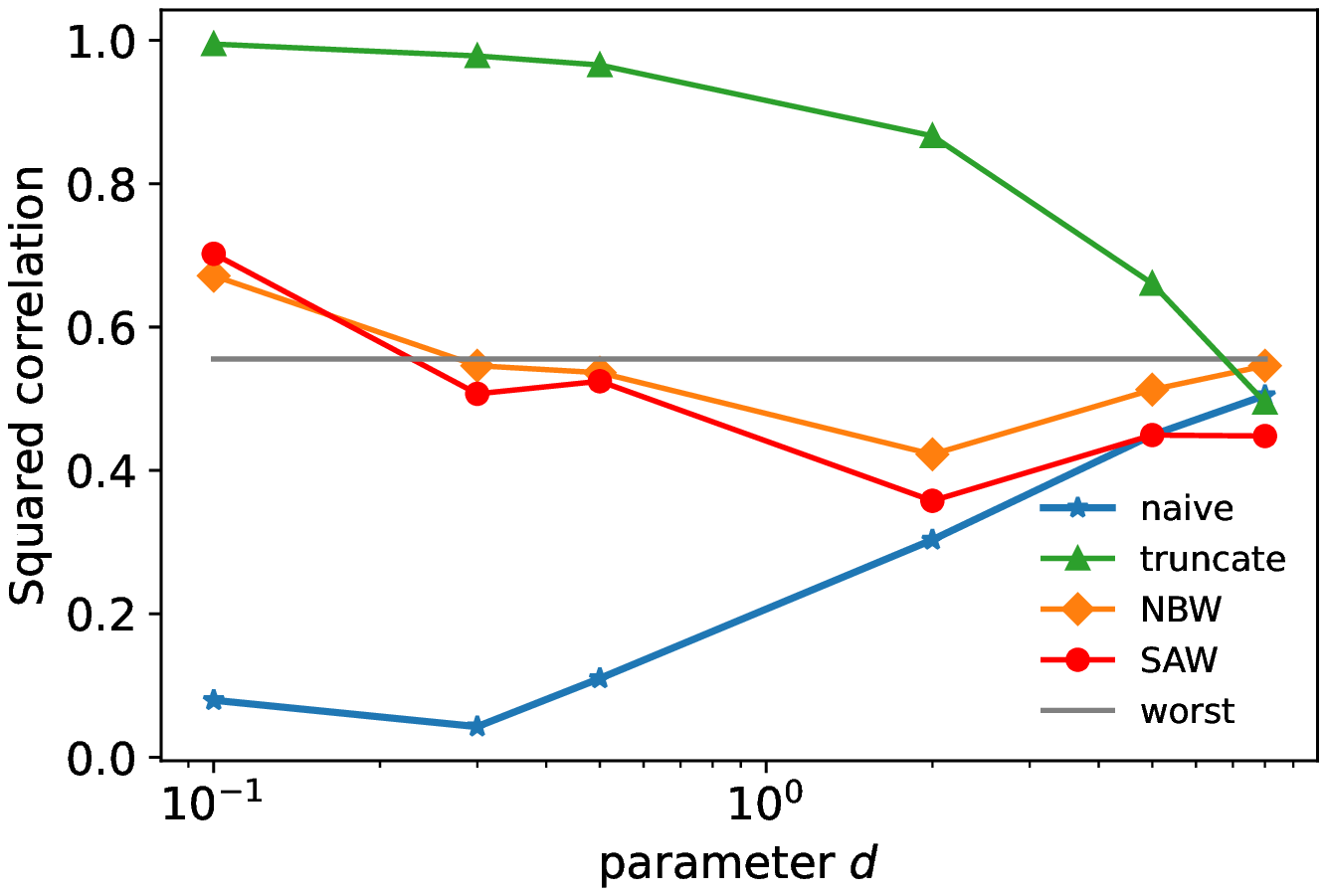}
        \label{fig:a}
        \caption{$n=200$, $\lambda^{\prime}=1.5$, $d\in [0.1,8]$}
    \end{subfigure}
          \begin{subfigure}[b]{0.45\linewidth}
        \includegraphics[width=\linewidth]{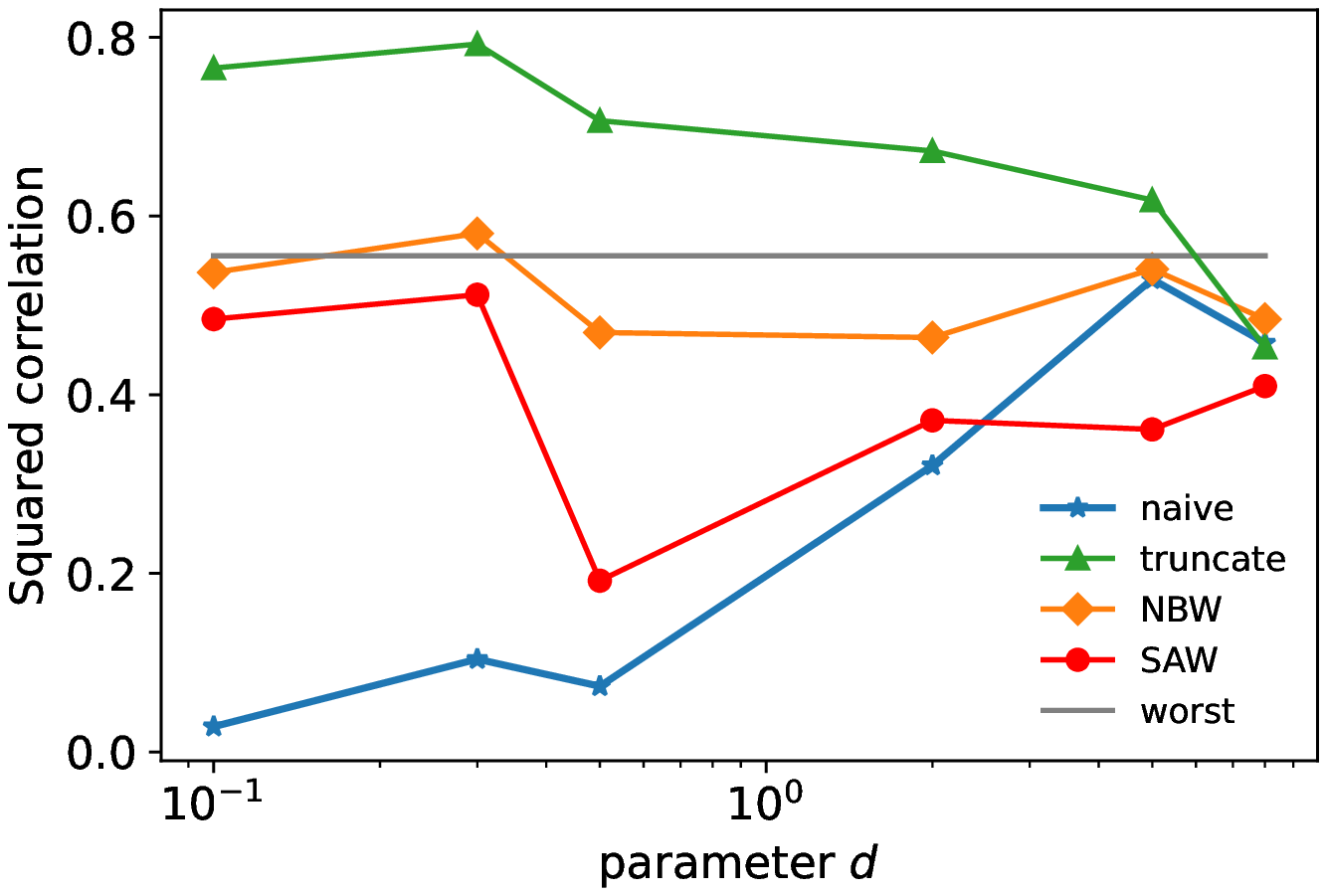}
        \label{fig:b}
        \caption{$n=200$, $\lambda^{\prime}=1.5$, $d\in [0.1,8]$}
    \end{subfigure}
    
    \begin{subfigure}[b]{0.45\linewidth}
        \includegraphics[width=\linewidth]{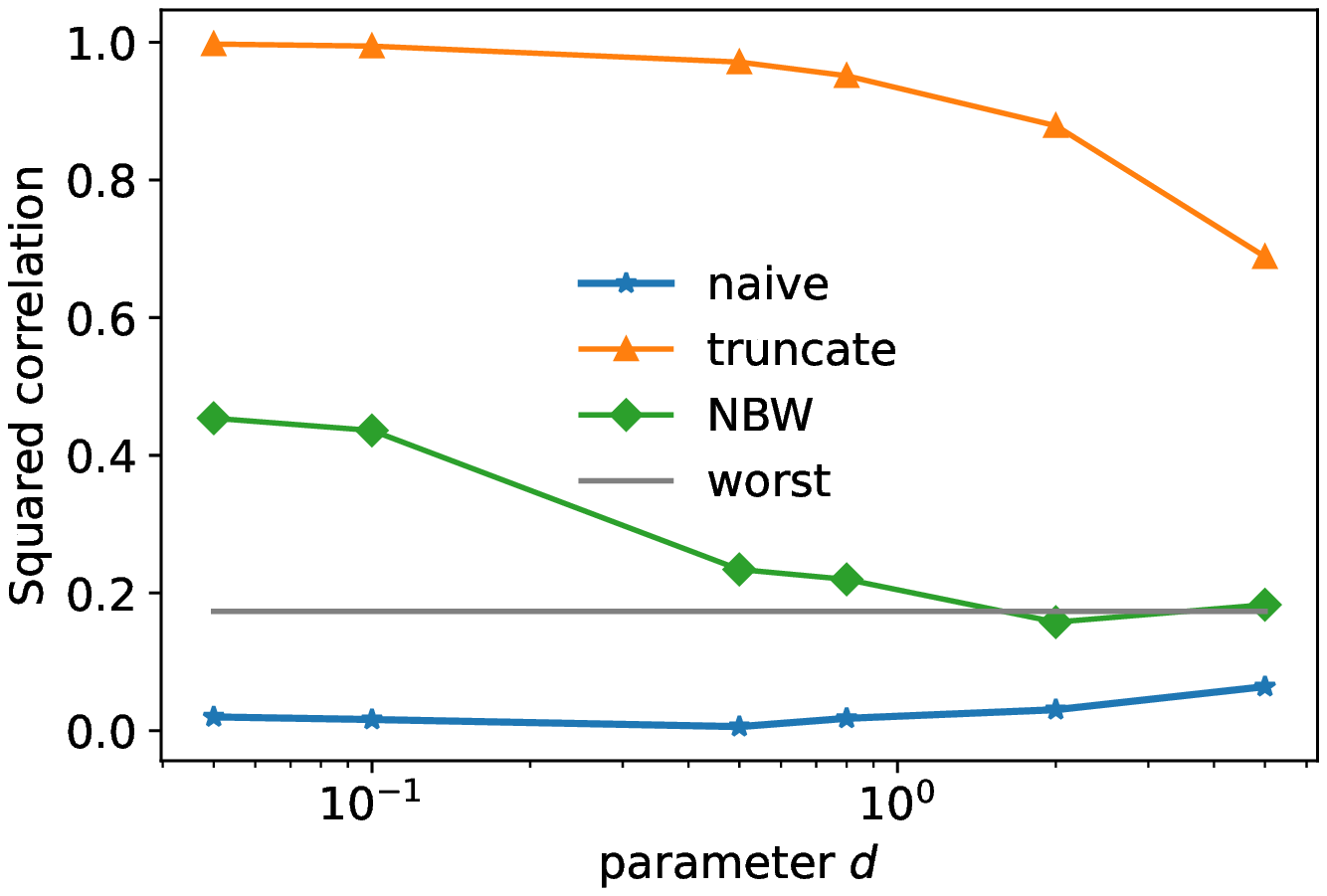}
        \label{fig:c}
        \caption{$n=400$, $\lambda^{\prime}=1.1$, $d\in [0.05,5]$}
    \end{subfigure}
    \begin{subfigure}[b]{0.45\linewidth}
        \includegraphics[width=\linewidth]{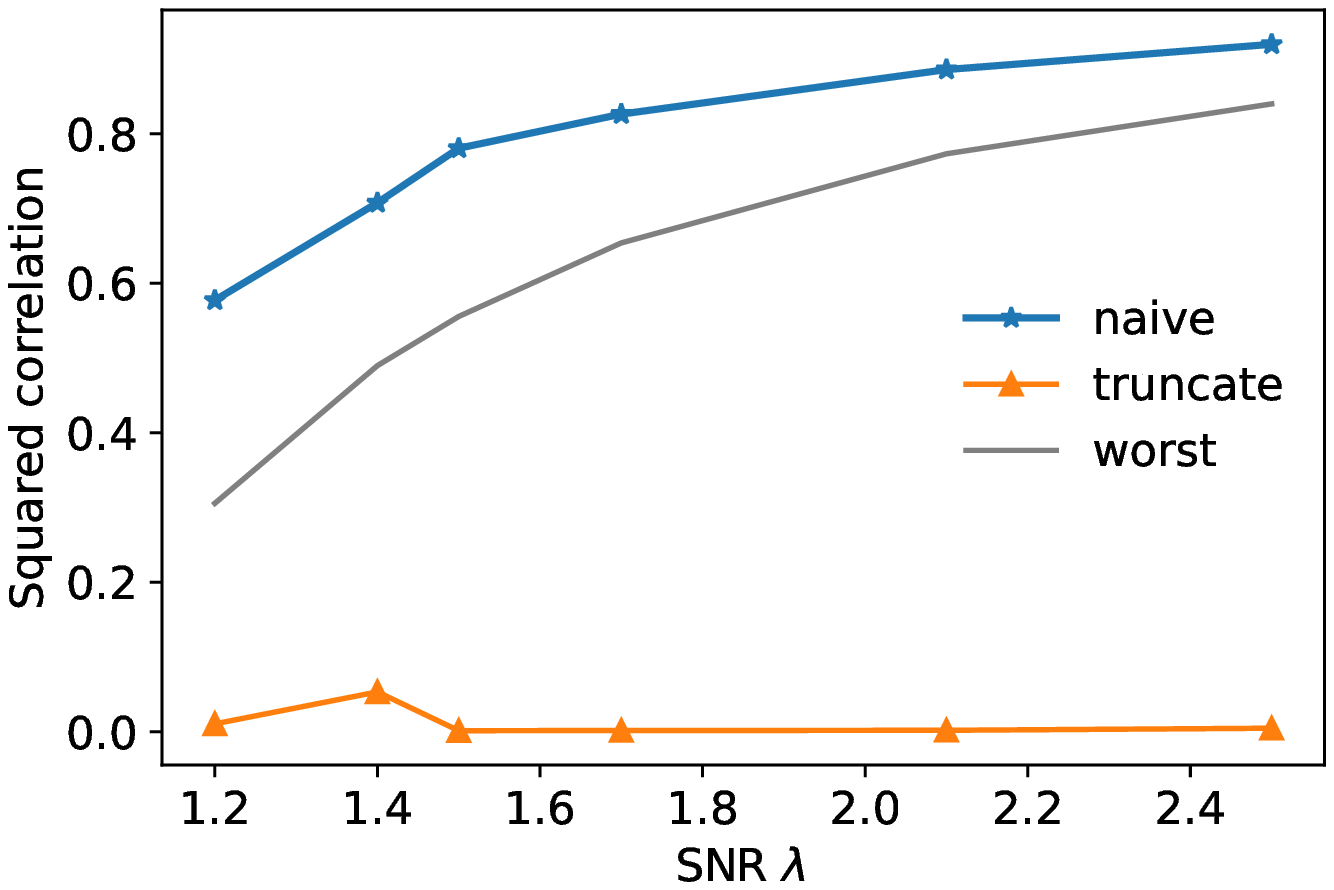}
        \label{fig:d}
        \caption{$n=2000$, $\lambda^{\prime}\in [1.2,1.6]$, $d=30$}
    \end{subfigure}
        
    \caption{
    (a)(b) The performance of  non-backtracking  walk estimator  with $\ell=10$ is no worse than self-avoiding walk estimator with $\ell=7$ under distribution (2),(3). They drastically beat naive PCA algorithm. (c) The performance of  non-backtracking walk estimator with length $\ell=17$ can be much better than PCA under distribution. (1) (d) Truncating at $\tau=5$ can fail drastically under distribution $(2)$. \newline
    Each data point is the result of averaging $20$ trials. For notation, $\lambda^{\prime}=\lambda  n^{1/2}$, the $y$ axis represents mean of squared correlation $\frac{\lprod\hat{x},x\rprod^2}{\lVert\hat{x}\rVert^2\lVert x\rVert^2}$.  The
 line ``worst'' represents the optimal guarantee in case of  Gaussian noise with same $\lambda$, while the line ``NBW'' represents the experiment results from non-backtracking algorithm.
 }\label{fig:Experiment}
\end{figure}

\section{Algorithms for general spiked tensor model}
\label{sec:tensor}

For proving theorem \ref{thm:intro-tensor}, we use the sum of multilinear polynomials corresponding to a variant of self-avoiding walk.  Here we only describe a simple special case of the algorithm, which provides estimation guarantee when $\lambda>n^{-3/4}$.

\begin{definition}[Polynomial time  estimator for spiked tensor recovery]\label{defEstimator}
Given tensor $Y\in\mathbb{R}^{n\times n\times n}$, we have estimator $P(Y)\in\mathbb{R}^n$ where each entry is degree $2\ell-1$ polynomial given by $P_i(Y)=\sum_{\alpha\in S_{\ell,i}} \chi_\alpha(Y)$, where $\chi_\alpha(Y)$ is multilinear polynomial basis $\chi_\alpha(Y)=\prod_{(i,j,k)\in\alpha} Y_{ijk}$ and $S_{\ell,i}$ is the set of directed hypergraph associated with vertex $i$ generated in the following way: 
\begin{itemize}
      \item we construct $2\ell$ levels of distinct vertices. Level $0$ is vertex $i$. For $0<t<\ell$, level $2t$ contains $1$ vertex and level $2t-1$ contains $2$ vertices. Level $2\ell-1$ contains $1$ vertex.
      \item We connect a hyperedge between adjacent levels $t-1,t$ for  $t\in [2\ell-2]$. Each hyperedge directs from level $t-1$ to level $t$.
      \item For vertex $u$ which lies in level $2\ell-2$ and vertices $v,v'$ which lie in level $2\ell-1$, we add the hyperedge $(u,v,v')$.
  \end{itemize}
\end{definition}

An illustration of a self-avoiding walk $\alpha\in S_{\ell,i}$  is given by figure~\ref{fig:sawTensor}.
\begin{figure}
    \centering

\begin{tikzpicture}[
roundnode/.style={circle, draw=green!60, fill=green!5, very thick, minimum size=7mm},
squarednode/.style={rectangle, draw=red!60, fill=red!5, very thick, minimum size=5mm},
scale=0.8,transform shape]
\node[roundnode]        (v1)          {$i$};

\node[roundnode]        (v2)       [below left=of v1] {$v_2$};
\node[roundnode]        (v3)       [below right=of v1] {$v_3$};
\node[roundnode]        (v4)       [below right=of v2] {$v_4$};
\node[roundnode]        (v5)       [below left=of v4] {$v_5$};
\node[roundnode]        (v6)       [below right=of v4] {$v_6$};
\node[roundnode]        (v7)       [below right=of v5] {$v_7$};
\node[roundnode]        (v8)       [below=of v7] {$v_{\ell_v-1}$};
\node[roundnode]        (v9)       [below left=of v8] {$v_{\ell_v}$};
\node[roundnode] (v10)
 [below right=of v8] {$v_{\ell_v}$};
    \begin{scope}[fill opacity=0.2]
    \filldraw[fill=yellow!70] ($(v1)+(0,0.5)$) 
        to[out=180,in=180] ($(v2) + (-0.4,-0.4)$) 
        to[out=0,in=180] ($(v3) + (0.4,-0.4)$)
        to[out=0,in=0] ($(v1)+(0,0.5)$);
        \filldraw[fill=blue!70] ($(v4)+(0,-0.5)$) 
        to[out=180,in=180] ($(v2) + (-0.4,0.4)$) 
        to[out=0,in=180] ($(v3) + (0.4,0.4)$)
        to[out=0,in=0] ($(v4)+(0,-0.5)$); 
        \filldraw[fill=yellow!70] ($(v4)+(0,0.5)$) 
        to[out=180,in=180] ($(v5) + (-0.4,-0.4)$) 
        to[out=0,in=180] ($(v6) + (0.4,-0.4)$)
        to[out=0,in=0] ($(v4)+(0,0.5)$);
                \filldraw[fill=blue!70] ($(v7)+(0,-0.5)$) 
        to[out=180,in=180] ($(v5) + (-0.4,0.4)$) 
        to[out=0,in=180] ($(v6) + (0.4,0.4)$)
        to[out=0,in=0] ($(v7)+(0,-0.5)$); 
   \filldraw[fill=yellow!70] ($(v8)+(0,0.7)$) 
        to[out=180,in=180] ($(v9) + (-0.4,-0.4)$) 
        to[out=0,in=180] ($(v10) + (0.4,-0.4)$)
        to[out=0,in=0] ($(v8)+(0,0.7)$);     
    \end{scope}






\draw[dotted](v7.south) -- (v8.north);







\end{tikzpicture}
    \caption{Illustration for a self-avoiding walk $\alpha\in \textrm{S}_{\ell,i}$ for tensor estimation. Each colored area corresponds to a hyperedge.}
    \label{fig:sawTensor}
\end{figure}
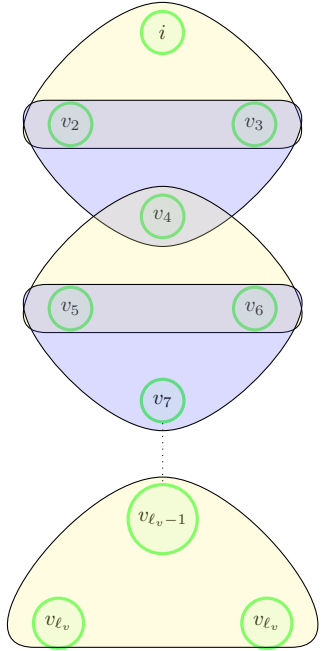

In appendix \ref{proofWeakRecovery} we show that by introducing width to levels, estimation under smaller SNR $\lambda$ is possible by exploiting more computational power. In appendix \ref{colorCodingTensor} we show that $P(Y)$ can be evaluated in $n^{O(v)}$ time using color coding method. These lead to the proofs of theorems \ref{thm:intro-tensor}, \ref{polyPCA}.



\section{Conclusion}
We provide an algorithm which nontrivially estimates rank-one spikes of Wigner matrices for signal-to-noise ratios $\lambda$ approaching the sharp threshold $\lambda \sqrt n \rightarrow 1$, even in the setting of heavy-tailed noise (having only $2$ finite moments) with unknown, adversarially-chosen distribution.
For future work, it would be intriguing to obtain strengthened guarantees along (at least) two axes.
First, \cite{Perry_2018} give an algorithm which recovers rank-one spikes for even smaller values of $\lambda$, when (a) the distribution of the entries of $W$ is known, and (b) a large constant number of moments of the entries of $W$ are $O(1)$.
Relaxing either of the assumptions (a) or (b) while keeping $\lambda \sqrt{n} \ll 1$ would be very interesting.

In a different direction, our experiments suggest that the non-backtracking walk estimator performs as well as the self-avoiding walk estimator which we are able to analyze rigorously.
Rigorously establishing similar guarantees for the non-backtracking walk estimator -- or finding counterexamples -- would be of great interest. 

\section{Acknowledgement}
S.B.H. is supported by a Miller Postdoctoral Fellowship. J.D and D.S are supported   by ERC consolidator grant.
\bibliographystyle{alpha}
\bibliography{refs}

\newcommand{\etalchar}[1]{$^{#1}$}
\begin{thebibliography}{BGL{\etalchar{+}}16a}

\bibitem[Abb17]{abbe2017community}
Emmanuel Abbe.
\newblock Community detection and stochastic block models: recent developments.
\newblock {\em The Journal of Machine Learning Research}, 18(1):6446--6531,
  2017.

\bibitem[AGH{\etalchar{+}}14]{anandkumar2014tensor}
Animashree Anandkumar, Rong Ge, Daniel Hsu, Sham~M Kakade, and Matus Telgarsky.
\newblock Tensor decompositions for learning latent variable models.
\newblock {\em Journal of Machine Learning Research}, 15:2773--2832, 2014.

\bibitem[AS18]{Abbe2018}
Emmanuel Abbe and Colin Sandon.
\newblock {Proof of the Achievability Conjectures for the General Stochastic
  Block Model}.
\newblock {\em Communications on Pure and Applied Mathematics},
  71(7):1334--1406, 2018.

\bibitem[AYZ95]{alon1995color}
Noga Alon, Raphael Yuster, and Uri Zwick.
\newblock Color-coding.
\newblock {\em Journal of the ACM (JACM)}, 42(4):844--856, 1995.

\bibitem[BBAP05]{baik2005}
Jinho Baik, Gérard Ben~Arous, and Sandrine Péché.
\newblock Phase transition of the largest eigenvalue for nonnull complex sample
  covariance matrices.
\newblock {\em Ann. Probab.}, 33(5):1643--1697, 09 2005.

\bibitem[BCRT19]{biroli2019iron}
Giulio Biroli, Chiara Cammarota, and Federico Ricci-Tersenghi.
\newblock How to iron out rough landscapes and get optimal performances:
  Replicated gradient descent and its application to tensor pca, 2019.

\bibitem[BGBK17]{benaychgeorges2017spectral}
Florent Benaych-Georges, Charles Bordenave, and Antti Knowles.
\newblock Spectral radii of sparse random matrices, 2017.

\bibitem[BGL{\etalchar{+}}16a]{BGG}
Vijay Bhattiprolu, Mrinalkanti Ghosh, Euiwoong Lee, Venkatesan Guruswami, and
  Madhur Tulsiani.
\newblock Multiplicative approximations for polynomial optimization over the
  unit sphere.
\newblock 11 2016.

\bibitem[BGL16b]{bhattiprolu2016sumofsquares}
Vijay Bhattiprolu, Venkatesan Guruswami, and Euiwoong Lee.
\newblock Sum-of-squares certificates for maxima of random tensors on the
  sphere, 2016.

\bibitem[BLM15]{bordenave2015non}
Charles Bordenave, Marc Lelarge, and Laurent Massouli{\'e}.
\newblock Non-backtracking spectrum of random graphs: community detection and
  non-regular ramanujan graphs.
\newblock In {\em 2015 IEEE 56th Annual Symposium on Foundations of Computer
  Science}, pages 1347--1357. IEEE, 2015.

\bibitem[DKMZ11]{decelle2011asymptotic}
Aurelien Decelle, Florent Krzakala, Cristopher Moore, and Lenka Zdeborov{\'a}.
\newblock Asymptotic analysis of the stochastic block model for modular
  networks and its algorithmic applications.
\newblock {\em Physical Review E}, 84(6):066106, 2011.

\bibitem[Has19]{hastings2019classical}
M.~B. Hastings.
\newblock Classical and quantum algorithms for tensor principal component
  analysis, 2019.

\bibitem[HKP{\etalchar{+}}17]{hopkins2017power}
Samuel~B Hopkins, Pravesh~K Kothari, Aaron Potechin, Prasad Raghavendra, Tselil
  Schramm, and David Steurer.
\newblock The power of sum-of-squares for detecting hidden structures.
\newblock In {\em 2017 IEEE 58th Annual Symposium on Foundations of Computer
  Science (FOCS)}, pages 720--731. IEEE, 2017.

\bibitem[HS17]{8104074}
S.~B. {Hopkins} and D.~{Steurer}.
\newblock Efficient bayesian estimation from few samples: Community detection
  and related problems.
\newblock In {\em 2017 IEEE 58th Annual Symposium on Foundations of Computer
  Science (FOCS)}, pages 379--390, 2017.

\bibitem[HSS15]{hopkins2015tensor}
Samuel~B Hopkins, Jonathan Shi, and David Steurer.
\newblock Tensor principal component analysis via sum-of-square proofs.
\newblock In {\em Conference on Learning Theory}, pages 956--1006, 2015.

\bibitem[Joh01]{johnstone2001distribution}
Iain~M Johnstone.
\newblock On the distribution of the largest eigenvalue in principal components
  analysis.
\newblock {\em Annals of statistics}, pages 295--327, 2001.

\bibitem[KMM{\etalchar{+}}13]{krzakala2013spectral}
Florent Krzakala, Cristopher Moore, Elchanan Mossel, Joe Neeman, Allan Sly,
  Lenka Zdeborov{\'a}, and Pan Zhang.
\newblock Spectral redemption in clustering sparse networks.
\newblock {\em Proceedings of the National Academy of Sciences},
  110(52):20935--20940, 2013.

\bibitem[KWB19]{kunisky2019notes}
Dmitriy Kunisky, Alexander~S. Wein, and Afonso~S. Bandeira.
\newblock Notes on computational hardness of hypothesis testing: Predictions
  using the low-degree likelihood ratio, 2019.

\bibitem[MNS18]{mossel2018proof}
Elchanan Mossel, Joe Neeman, and Allan Sly.
\newblock A proof of the block model threshold conjecture.
\newblock {\em Combinatorica}, 38(3):665--708, 2018.

\bibitem[MS16]{Montanari:2016:SPS:2897518.2897548}
Andrea Montanari and Subhabrata Sen.
\newblock Semidefinite programs on sparse random graphs and their application
  to community detection.
\newblock In {\em Proceedings of the Forty-eighth Annual ACM Symposium on
  Theory of Computing}, STOC '16, pages 814--827, New York, NY, USA, 2016. ACM.

\bibitem[PRS13]{pizzo2013}
Alessandro Pizzo, David Renfrew, and Alexander Soshnikov.
\newblock On finite rank deformations of wigner matrices.
\newblock {\em Ann. Inst. H. Poincaré Probab. Statist.}, 49(1):64--94, 02
  2013.

\bibitem[PWBM18]{Perry_2018}
Amelia Perry, Alexander~S. Wein, Afonso~S. Bandeira, and Ankur Moitra.
\newblock Optimality and sub-optimality of pca i: Spiked random matrix models.
\newblock {\em The Annals of Statistics}, 46(5):2416–2451, 2018.

\bibitem[RM14]{richard2014statistical}
Emile Richard and Andrea Montanari.
\newblock A statistical model for tensor pca.
\newblock In {\em Advances in Neural Information Processing Systems}, pages
  2897--2905, 2014.

\bibitem[RRS17]{10.1145/3055399.3055417}
Prasad Raghavendra, Satish Rao, and Tselil Schramm.
\newblock Strongly refuting random csps below the spectral threshold.
\newblock In {\em Proceedings of the 49th Annual ACM SIGACT Symposium on Theory
  of Computing}, STOC 2017, page 121–131. Association for Computing
  Machinery, 2017.

\bibitem[RSS18]{Raghavendra2018}
Prasad Raghavendra, Tselil Schramm, and David Steurer.
\newblock {High-dimensional estimation via sum-of-squares proofs}.
\newblock {\em International Congress of Mathematicians}, 2018.

\bibitem[SLKZ15]{Saade2015SpectralDI}
Alaa Saade, Marc Lelarge, Florent Krzakala, and Lenka Zdeborov{\'a}.
\newblock Spectral detection in the censored block model.
\newblock {\em 2015 IEEE International Symposium on Information Theory (ISIT)},
  pages 1184--1188, 2015.

\bibitem[WAM19]{Kikuchy}
Alexander~S. Wein, Ahmed~El Alaoui, and Cristopher Moore.
\newblock The kikuchi hierarchy and tensor {PCA}.
\newblock In {\em 60th {IEEE} Annual Symposium on Foundations of Computer
  Science, {FOCS} 2019}, pages 1446--1468, 2019.

\end{thebibliography}
\newpage
\appendix
\section{Spiked matrix model}
\subsection{Proof of Theorem \ref{truncationProof}}\label{truncation-proof-section}

For proof of theorem \ref{truncationProof}, we need a result available in previous literature stating about the universality of spiked matrix model
\begin{theorem}[Theorem 1.1  in \cite{pizzo2013}]\label{thm:universality}
In spiked matrix model $Y=\lambda xx^T+W$, $x\in\mathbb{R}^n$ has norm $\sqrt{n}$, $W\in \mathbb{R}^{n\times n}$ is a symmetric random matrix of i.i.d entries with zero mean and variance bounded by $1$. If the $5$-th moment of entries in $W$ is bounded by $O(1)$, then the following guarantee will hold w.h.p:
\begin{equation*}
    \lambda_{\textrm{max}}(Y)\geq (1-o(1))\left(\lambda n+\frac{1}{\lambda}\right)
\end{equation*}
\end{theorem}
We also need a simple observation about the deterministic relation between leading eigenvalue and leading eigenvector in spiked matrix model.
\begin{lemma}\label{ObservationA11}
    For matrix $M\in\mathbb{R}^{n\times n}$ and  matrix $N=\gamma xx^T+M$(where $\gamma>0$ and $x\in\mathbb{R}^n$ has $\sqrt{n}$), if the leading eigenvalue $\lambda_{\textrm{max}}(N)$ is larger than $\lambda_{\textrm{max}}(M)$ by $\Omega(n\gamma)$, then the unit norm leading eigenvector of $N$ denoted by $\xi$ will achieve constant correlation with $x$:
    \begin{equation*}
        \lprod \xi, x\rprod^2 \geq \Omega(n)
    \end{equation*}
\end{lemma}
\begin{proof}
    We have $\lambda_{\textrm{max}}(N)=\xi^\top (\gamma xx^\top+M)\xi\leq \lambda_{\textrm{max}}(M)+\gamma \lprod\xi,x\rprod^2$. Since $\lambda_{\textrm{max}}(N) -\lambda_{\textrm{max}}(M)=\Omega(n\gamma)$, we have $\lprod\xi,x\rprod^2=\Omega(n)$. 
\end{proof}

\begin{proof}[Proof of Theorem \ref{truncationProof}]
By definition we  have
\begin{equation*}
    Y^{\prime}_{ij}=Y\mathbbm{1}_{\lvert  Y_{ij}\rvert\leq \tau}+ \tau\textrm{sgn}(Y_{ij}) \mathbbm{1}_{\lvert Y_{ij}\rvert\geq \tau}
\end{equation*}

Given the assumption on the $\lvert\lambda x_ix_j\rvert=o(\tau)$, one can observe that this can be decomposed into
\begin{equation*}
    Y^{\prime}=\lambda xx^\top+T+M+\Delta
\end{equation*}
where we have
\begin{eqnarray*}
   T_{ij} &=& W_{ij}\mathbbm{1}_{\lvert W_{ij}\rvert}+\tau \textrm{sgn}(W_{ij}) \mathbbm{1}_{\lvert W_{ij}\rvert>\tau}\\
   M_{ij} &=& -\lambda x_ix_j\mathbbm{1}_{\lvert W_{ij}\rvert\geq \tau}\\
   \Delta_{ij} &=& (\tau\textrm{sgn}(Y_{ij})-Y_{ij})(\mathbbm{1}_{\lvert\lambda x_ix_j+W_{ij}\rvert\geq \tau}-\mathbbm{1}_{\lvert W_{ij}\rvert\geq\tau})
\end{eqnarray*}
Further we denote $Y^{\prime}-Y^{\prime\prime}$ as $H$. Then we have $Y^{\prime\prime}=\lambda xx^\top+(T-\E        
 T)+M+\Delta-(H-\E T)$.

Next we analyze the terms in the decomposition of $Y^{\prime\prime}$. Specifically, we want to show that with constant probability the largest eigenvalue of $Y^{\prime\prime}$ is larger than the  one of $Y^{\prime\prime}-\lambda xx^T$ by $\Omega(\lambda n)$. If this is proved then for the leading unit eigenvector of $Y^{\prime\prime}$ denoted by $\hat{x}$, we must have $\lprod\hat{x},x\rprod^2=\Omega(n)$ with constant probability by lemma \ref{ObservationA11}. 

First for matrix $T-\mathbb{E}T$, the variance of each entry is bounded by $1$. Further each entry is bounded by $2\tau$. According to  theorem \ref{thm:universality}, the largest eigenvalue of matrix $\lambda xx^\top+T-\mathbb{E} T$
 is given by $\lambda n+\frac{1}{\lambda}$ and the largest eigenvalue of matrix $T-\mathbb{E} T$ is given by $2\sqrt{n}$ with high probability. 
 
 For matrix $M$,we have $\mathbb{E} \mathbbm{1}_{\lvert W_{ij}\rvert\geq \tau}\leq  \frac{1}{\tau^2}$ because  the variance of entries in $W$ is bounded by $1$. Therefore the expectation $\E\lVert M\rVert_F$ is bounded by $\frac{\lambda n}{\tau}$. 
For non-zero entries $(i,j)$ in matrix $\Delta$, we must have $\lvert Y_{ij}-\tau\textrm{sgn}(Y_{ij})\rvert \leq \lvert \lambda x_ix_j\rvert$. Therefore these non zero entries $\Delta_{ij}$ are bounded by $\lvert\lambda x_ix_j\rvert$. Further each entry in $\Delta$ is non-zero with probability bounded by $\frac{1}{\tau^2}$. Therefore we have $\E \lVert \Delta\rVert_F$ bounded by $\frac{\lambda n}{\tau}$.
 
Finally we have $H=h\mathbbm{1}\mathbbm{1}^\top$ where $h=\frac{\sum _{i,j} Y^{\prime}_{i,j}}{n^2}$. Since $T_{ij}$ are i.i.d for $i\leq j$, we have
 \begin{equation*}
     \E\left[\left(\frac{\sum_{ij} T_{ij}}{n^2}-\E [T_{ij}]\right)^2\right] \leq \frac{4}{n^2}
 \end{equation*}
 By linearity of expectation $\E\lVert H-\mathbb{E} T\rVert\leq  O(1)\leq \frac{2\lambda n}{\tau}$. 
  
 In all we have $\E\lVert Y^{\prime\prime}-\lambda  xx^\top-T+\E T\rVert \leq \frac{4\lambda n}{\tau}$. By Markov inequality with probability $1/2$, we have $\lVert Y^{\prime\prime}-\lambda  xx^\top-T+\E T\rVert \leq \frac{8\lambda n}{\tau}$. As stated above, the largest eigenvalue of matrix $\lambda xx^\top+T-\mathbb{E} T$
 is given by $\lambda n+\frac{1}{\lambda}$ and the largest eigenvalue of matrix $T-\mathbb{E} T$ is given by $2\sqrt{n}$ with high probability.
 If we take $\tau$ large enough constant(e.g $\frac{100}{\textrm{min}(\epsilon^2,1)}$), then with probability $\Omega(1)$ the spectral norm of $Y^{\prime\prime}$ is larger than $\lambda n+\frac{1}{\lambda}-0.1\textrm{min}(\epsilon^2,1) \lambda n$ and spectral norm of  $Y^{\prime\prime}-\lambda xx^\top$ is smaller than  $2\sqrt{n}+0.1\textrm{min}(\epsilon^2,1) \lambda n$. 
 
 Therefore for $\lambda=1+\epsilon$ with $\epsilon=\Omega(1)$, the spectral norm of $Y^{\prime\prime}$ is larger than the spectral norm of $Y^{\prime\prime}-\lambda x x^\top$ by $\Omega(\lambda n)$ with  constant probability.
 As a result, with constant probability the leading eigenvector of $Y^{\prime\prime}$ must achieve $\Omega(1)$ correlation with hidden vector $x$ by lemma \ref{ObservationA11}.
\end{proof}

\subsection{Example of failure for truncation algorithm}\label{failureExample}
For example \ref{egMixed}, we only explain why truncating to $0$ could fail in main article. Next we show that truncating to $\tau\textrm{sgn}(Y_{ij})$ will fail as well for $d$ between $o(n/\tau^2)$ and $\omega(1)$. 


We still denote $h=\{\pm 1\}^n$ as the Rademacher vector with alternating sign and $x\in \{\pm 1\}^n$ orthogonal to all-$1$ vector. For $d=o(n/\tau^2)$ and $d=\omega(1)$, only entries perturbed by noise $\pm \sqrt{\frac{d}{n-d}}$ are not truncated. Then $Y^{\prime}_{ij}=\lambda x_ix_j+\sqrt{\frac{d}{n-d}}(-1)^{i+j}$ with probability $1-\frac{d}{n}$ and $-\tau$ with probability $d/n$. Therefore $Y^{\prime}$ can be decomposed into
\begin{equation*}
    Y^{\prime }=\lambda xx^\top+\sqrt{\frac{d}{n-d}}hh^\top+\Delta
\end{equation*}
where $\E \lVert\Delta\rVert=o(\sqrt{nd})$. Above computational threshold $\lambda\geq n^{-1/2}$, the spectral norm  of $\lambda xx^\top$ is smaller than $\sqrt{n}$. Further matrix $H=Y^{\prime\prime}-Y^{\prime}$ also has spectral norm $o(\sqrt{nd})$  by central limit theorem. 

For unit norm leading eigenvector $\xi$ of matrix $Y^{\prime\prime}$, we suppose that $\E\lprod\xi,x\rprod^2\geq \Omega(n)$ and prove by contradiction.
Because $\xi$ is leading eigenvector, we have $\E\lprod\xi\xi^\top,Y^{\prime\prime}\rprod\geq \E\lprod hh^\top, Y^{\prime\prime}\rprod\geq (1-o(1))\sqrt{nd}$. However,  we have $\lprod\xi,h\rprod^2\leq 1-\lprod\xi,x\rprod^2/n$. Therefore, we have $\lprod\xi\xi^\top,Y^{\prime\prime}\rprod
\leq (1-\Omega(1)) \sqrt{nd}+o(\sqrt{nd})$. This leads to contradiction.

\subsection{Proof of lemma \ref{criticalQuantity}, lemma  \ref{simpleRelation} and lemma \ref{corollarySegmentExpectation}}\label{proofQuantity}
We first prove lemma \ref{criticalQuantity}
\begin{proof}[Proof of Lemma \ref{criticalQuantity}]
First since the bound on infinity norm of $x$, we have 
\begin{equation*}
    S_{t_1,t_2,V}\leq \lVert x\rVert_\infty^{2t_2}\mathbb{E}_{(v_1,\ldots,v_{t_1+t_2})\subseteq [n]\setminus V} \left[\prod_{i=1}^{t_1+t_2}x_{v_i}^2\right]
\end{equation*}
Denote $t_1+t_2=t$,
now we take marginal on $x_{v_{t}}$. Then the marginal is given by
\begin{align*}
& \mathbb{E}_{(v_1,\ldots,v_{t_1+t_2})\subseteq [n]\setminus V} \left[\prod_{i=1}^{t_1+t_2}x_{v_i}^2\right]\\ =
 & \mathbb{E}_{(v_1,\ldots,v_{t-1})\subseteq [n]\setminus V} \left[\prod_{i=1}^{t-1}x_{v_i}^2 \left(\frac{n-\sum_{i=1}^{t-1}x_{v_i}^2-\sum_{i\in V} x_i^2}{n-\lvert V\rvert-(t-1)}\right)\right] \\ = & \left(1\pm O\left(\frac{(|V|+t)\lVert x\rVert_\infty^2}{n}\right) \right)  \mathbb{E}_{(v_1,\ldots,v_{t-1})\subseteq [n]\setminus V} \left[\prod_{i=1}^{t-1}x_{v_i}^2\right] 
\end{align*}
 By induction this is $(1\pm n^{-\Omega(1)})^t$. In all we have $S_{t_1,t_2,V}$ bounded by $(1+o(1))\lVert x\rVert_{\infty}^{2t_2}$. 
\end{proof}

Next we prove lemma \ref{simpleRelation}.

\begin{proof}[Proof of Lemma \ref{simpleRelation}]
   For self-avoiding walks $\alpha,\beta\in\textrm{SAW}(i,j)$, the connected components of $\alpha\cap\beta$ are all self-avoiding walks. We consider quantity $g=r-k$. Then each of the $p$ connected components in $\alpha\cap\beta$ not containing $i$ or $j$ contributes $1$ to $g$.  Since $\alpha\neq \beta$, other connected components in $\alpha\cap\beta$ can only contain one of $i,j$. Such connected components contribute $0$ to $g$. 
   Further each shared vertex not incident to any shared edge contribute $1$ to $g$ and the total number of such vertices is given by $s$. Therefore we have $r-k=p+s$
\end{proof}

Finally we prove lemma \ref{corollarySegmentExpectation}. 
\begin{proof}[Proof of Lemma \ref{corollarySegmentExpectation}]
   We prove the first bound.  We represent $\xi$ as an ordered set of vertices $(\xi_{0},\xi_{1},\ldots,\xi_{h})$ then we note that the product in the expectation can be expanded to the sum of monomials:
   \begin{equation*}
   \E_{\xi\subseteq [n]\setminus V}  \left[x_{\xi_{\textrm{0}}}^2 x_{\xi_{\textrm{h}}}^2 \prod_{i\in [h]} (1+\lambda^2 x_{\xi_i}^2x_{\xi_{i-1}}^2)\right]
    = \E_{\xi \subseteq [n]\ V}\left[\sum_{S \subseteq [h]} x_{\xi_{\textrm{0}}}^2x_{\xi_{\textrm{h}}}^2 \prod_{i\in S }\lambda^2 x_{\xi_i}^2x_{\xi_{i-1}}^2\right]   
    \end{equation*}
    Since for fixed set $S\subseteq [h]$, the number of variables $x_u$ with degree $4$ in the monomial is bounded by $\lvert S\rvert$+1,  by lemma \ref{criticalQuantity},we have 
    \begin{equation*}
   \E_{\xi_0,\xi_1,\ldots,\xi_h}\left[ x_{\xi_0}^2x_{\xi_h}^2  \prod_{i\in S }(\lambda^2 x_{\xi_i}^2x_{\xi_{i-1}}^2)\right]\leq (1+n^{-\Omega(1)})\lVert x\rVert_{\infty}^2 \left(\lambda \lVert x\rVert_{\infty}\right)^{2\lvert S\rvert} 
   \end{equation*}
   
   On the other hand we have
   \begin{equation*}
     (1+\lambda^2\lVert x\rVert_{\infty}^{2})^{h}=\sum_{S \subseteq [h]}
     \left(\lambda^2\lVert x\rVert_{\infty}^{2}\right)^{\lvert S\rvert}
   \end{equation*}
   Therefore we have 
   \begin{equation*}
    \E_{\xi\subseteq [n]\setminus V}  \left[x_{\xi_{0}}^2 x_{\xi_h}^2 \prod_{i\in [h]} (1+\lambda^2 x_{\xi_i}^2x_{\xi_{i-1}}^2)\right] \leq  (1+n^{-\Omega(1)})\lVert x\rVert_{\infty}^2 (1+\lambda^2 \lVert x\rVert_{\infty}^2)^{h}
    \end{equation*}
    The other three bounds can be proved in very similar ways.
\end{proof}

\subsection{Evaluation of self-avoiding walk estimator}\label{sawEvaluation}
In spiked matrix model $Y=\lambda xx^\top +W$, we denote $\delta=n^{1/2}\lambda-1$.
For evaluation of degree $\ell=O(\log_{1+\delta} n)$ self-avoiding walk polynomial:
\begin{equation*}
    P_{ij}(Y)=\sum_{\alpha\in \textrm{SAW}(i,j)} \prod_{(u,v)\in\alpha} Y_{uv}
\end{equation*}
we use color coding strategy pretty similar to the literature in \cite{8104074}. The algorithm and construction  of matrices has already been described in the main body. We restate the algorithm \ref{algoSelfAvoidingWalk} for readers' convenience. 
\begin{algorithm}
\KwData{Given $Y\in \mathbb{R}^{n\times n}$ s.t $Y=\lambda xx^\top+W$}
\KwResult{$P(Y)\in\mathbb{R}^{n\times n}$ where $P_{ij}(Y)$ is the sum of multilinear monomials corresponding to length   $\ell$ self-avoiding walk between $i,j$(up to accuracy $1+n^{-\Omega(1)}$)}
 \For{$i\gets1$ \KwTo $C$}{
     Sample coloring $c_t:[n]\mapsto [\ell]$ uniformly at random\;
   Construct a $\mathbb{R}^{2^{\ell}n\times 2^{\ell} n}$ matrix $M$, where rows and columns are indexed by $(v,S)$, where $v\in [n]$  and $S$ is a subset of $[\ell]$\;
    a matrix $H\in \mathbb{R}^{n\times  2^{\ell} n}$ where each row is indexed by $[n]$ and each column is indexed by $(v,S)$ where $v\in [n]$  and $S$ is a subset of $[\ell]$\;
    
    a matrix $N\in \mathbb{R}^{2^\ell n\times n}$, where each row is indexed by $(v,S)$ where $S$ is a subset of $[\ell]$ and each column is indexed by $[n]$\;
    
    Record $p_{c_i}= HM^{\ell-2}N$\;}
Return  $\sum _{i=1}^C p_{c_i}/C$
 \caption{Algorithm for evaluating self-avoiding walk matrix }\label{algoSelfAvoidingWalk}
 \end{algorithm}

 On complete graph $K_n$, for a specific coloring $c\in [n]\mapsto [\ell]$, we say that a length-$\ell$ self-avoiding walk $\alpha=(v_0,v_1,\ldots,v_\ell)$ is colorful if the colors of $v_0,v_1,\ldots,v_\ell$ are different.Then a critical observation is that $p_{c,i,j} (Y)=\sum_{\alpha\in \textrm{SAW}_\ell(i,j)}F_{c,\alpha}\chi_{\alpha}(Y)$,
 where $F_{c,\alpha}$ is the $0$-$1$ indicator of  random event that $\alpha$ is colorful. Taking uniform expectation on $c$ over all random colorings, we have the following relation:
 \begin{lemma}\label{varianceReduction}
    In the algorithm \ref{algoSelfAvoidingWalkMain}, for $C\geq \exp(100\ell)$, we have 
    \begin{equation*}   
    \E_{Y,c_1,\ldots,c_C} \left[\left(\frac{1}{C}\sum_{t}p_{c_t,i,j}(Y)-\E_c p_{c,i,j}(Y)\right)^2\right]\leq \exp(-O(\ell))\E_{Y}\left(\E_c P_{c,i,j}(Y) \right)^2
     \end{equation*} 
     where random colorings $c_1,c_2,\ldots,c_C\in [n]\mapsto [\ell]$ are independently sampled uniformly at random and the expectation of random coloring $c\in [n]\mapsto [\ell]$ on right hand side is taken uniformly at random.
 \end{lemma}
  \begin{proof}
  For a fixed path $\alpha$, the probability that $F_{c,\alpha}=1$ is bounded by $\frac{\ell!}{\ell ^{\ell}}\geq \exp(-O(\ell))$. Therefore we have 
  \begin{subequations}
     \begin{align}
         \E \left[p_{c,i,j}^2(Y)\right] & =\E\left[\sum_{\alpha,\beta\in \textrm{SAW}_{\ell}(i,j)} \chi_\alpha(Y)\chi_\beta(Y) F_{c,\alpha}F_{c,\beta}\right]\\
         & \leq  \E_Y\left[\sum_{\alpha,\beta\in \textrm{SAW}_{\ell}(i,j)} \left[ \chi_\alpha(Y)\chi_\beta(Y)\right]\right]\\
         & \leq \exp(O(\ell))\E_Y\left[\left(\sum_{\alpha\in \textrm{SAW}_{\ell}(i,j)} \E_c\left[ \chi_\alpha(Y)F_{c,\alpha}\right]\right)^2\right]\\
         & = \exp(O(\ell)) \E_Y\left[\left(\E_c [p_{c,i,j}(Y)]\right)^2\right]
     \end{align}
     \end{subequations}
        For step (3b) and (3c), we use the fact that $0\leq F_{c,\alpha}\leq 1$ and $\E [\chi_\alpha(Y)\chi_\beta(Y)]\geq 0$ for all $\alpha,\beta\in \textrm{SAW}_{\ell}(i,j)$. For step  (3c), we also use the fact that $\E F_{c,\alpha}\geq \exp(-O(\ell))$. Therefore
     \begin{equation*}
        \E_Y \E_c \left[\left(p_{c,i,j}(Y)-\E_c p_{c,i,j}(Y)\right)^2\right] \leq \exp(O(\ell)) \E_Y \left[\left(\E_c p_{c,i,j}(Y)\right)^2\right] 
     \end{equation*}
     By averaging $p_{c,i,j}(Y)$ for $C$ independent random colorings, the variance is reduced and we have
     \begin{equation*}
         \E_{Y,c_1,\ldots,c_C} \left[\left(\frac{1}{C}\sum_{t,i,j}p_{c_t,i,j}(Y)-\E_c p_{c,i,j}(Y)\right)^2\right]\leq \frac{1}{C}\exp(O(\ell))\E_{Y}\left(\E_c P_{c,i,j}(Y) \right)^2
     \end{equation*}
     Therefore let $C=\exp(100\ell)$, the lemma is proved.
  \end{proof}
  This lemma implies that the average of $p_c(Y)$ for $n^{\delta^{-O(1)}}$ independent random colorings $p(Y)$ gives accurate approximation of $P(Y)$. The following simple corollary implies that the this matrix $p(Y)$ achieves the same correlation with $xx^\top$ as $P(Y)$. 
  
  \begin{lemma}[Formal statement of Lemma \ref{sawEva}]\label{formal28}
  The algorithm \ref{algoSelfAvoidingWalk} runs in $n^{\delta^{-O(1)}}$ time when $\ell=O(\log_{1+\delta} n)$. For matrix returned by algorithm \ref{algoSelfAvoidingWalk}, we have
  $$
  \E_{c_1,\ldots,c_C}\left[\frac{1}{C}\sum_{t}p_{c_t,i,j}(Y)\right]=\frac{\ell !}{\ell^\ell}P_{ij}(Y)
  $$
       \begin{equation*}   \E_{Y,c_1,\ldots,c_C} \left[\left(\frac{1}{C}\sum_{t}p_{c_t,i,j}(Y)\right)^2\right]\leq \left(1+n^{-\Omega(1)}\right)\E_{Y}\left[\left(\frac{\ell !}{\ell^\ell}P_{ij}(Y) \right)^2\right]
     \end{equation*}
  \end{lemma}
  
  \begin{proof}[Proof of Lemma \ref{formal28}]
    First we note that for $\ell=O(\log_\delta n)$, algorithm \ref{algoSelfAvoidingWalk} runs in time $n^{\delta^{-O(1)}}$.
  
  For any random coloring $c$ and length-$\ell$ self-avoiding walk $\alpha$, the probability  that $F_{c,\alpha}=1$ is $\ell!/\ell^\ell$. Thus $\E F_{c,\alpha}=  \ell!/\ell^\ell$. Since $p_{c,i,j}=\sum_{\alpha\in\textrm{SAW}_{\ell}(i,j)}F_{c,\alpha}\chi_\alpha(Y)$, by linearity of expectation we get the first equality.
  
    By lemma \ref{varianceReduction}, we have 
    \begin{equation*}   \E_{Y,c_1,\ldots,c_C} \left[\left(\frac{1}{C}\sum_{t}p_{c_t,i,j}(Y)-\E_c p_{c,i,j}(Y)\right)^2\right]\leq \exp(-O(\ell))\E_{Y}\left(\E_c P_{c,i,j}(Y) \right)^2
     \end{equation*}
     Therefore 
        \begin{equation*}   \E_{Y,c_1,\ldots,c_C} \left[\left(\frac{1}{C}\sum_{t}p_{c_t,i,j}(Y)\right)^2\right]\leq \left(1+n^{-\Omega(1)}\right)\E_{Y}\left[\left(\E_c p_{c,i,j}(Y) \right)^2\right]
     \end{equation*}
     Further as stated above we have $\E_c p_{c,i,j}(Y)=\frac{\ell!}{\ell^{\ell}}P_{ij}(Y)$. Thus we get the inequality.
  \end{proof}
  
  Now the proof of theorem  \ref{thm:intro-matrix} is self-evident.
  \begin{proof}[Proof of Theorem \ref{thm:intro-matrix}]
    We denote $\frac{1}{C}\sum_{t=1}^C p_{c_t}(Y)$ as $p(Y)$. Then by lemma \ref{formal28} and lemma \ref{matrixCorrlation},  we have
     \begin{equation*}
       \frac{\E_{c_1,\ldots,c_C,Y} \lprod 
       \hat{p}(Y),xx^\top\rprod}{n\left(\E_{c_1,\ldots,c_C,Y} \lVert \hat{p}(Y) \rVert_F^2\right)^{1/2}}=\delta^{O(1)}=\Omega(1)
     \end{equation*}
     
          By the same rounding procedure as  in \cite{8104074}, we obtain theorem \ref{thm:intro-matrix} by extracting a random vector in the span of top $1/\delta^{O(1)}$ eigenvectors of $Y$.
  \end{proof}

\subsection{Algorithm for evaluating non-backtracking walk estimator}\label{NBWproof}
In experiments, we use estimator closely related to non-backtracking walk and color coding method. 

On complete graph $K_n$, for vertice labels $i,j\in [n]$, we define the set of length-$\ell$ $k$-step non-backtracking walks $(i,v_1,v_2,\ldots,j)$ as $\textrm{NBW}_{\ell}(i,j)$. For non-backtracking walk $\alpha$ and random coloring $c:[n]\mapsto [\ell]$, we denote $F_{c,\alpha}$ as $0$-$1$ indicator of the random event that each length $k$ chunk of walk $\alpha$ is colorful(i.e, not containing repeated colors). For a fixed path $\alpha$, the probability that $F_{c,\alpha}=1$ is bounded by $\left(1-\frac{k}{\ell}\right)^{\ell}\geq \exp(-O(k))$.

Then we use the following non-backtracking walk estimator $P(Y)\in\mathbb{R}^{n\times n}$:
\begin{equation}\label{eq:nbwPoly}
P_{ij}(Y)=\sum_{\alpha\in \textrm{NBW}_{\ell}(i,j)}\chi_\alpha(Y) \E_c F_{c,\alpha}
\end{equation}
where $\chi_\alpha(Y)= \prod_{(u,v)\in \alpha}Y_{uv}$ and the expectation on coloring $c$ is taken uniformly. The algorithm for approximating $P(Y)$ is given as following:
\begin{algorithm}
	\KwData{Given $Y\in \mathbb{R}^{n\times n}$ s.t $Y=\lambda xx^\top+W$}
	\KwResult{Approximation for $P(Y)\in\mathbb{R}^{n\times n}$ where  $P_{ij}(Y)=\sum_{\alpha \in \textrm{NBW}_{\ell}(i,j)} \hat{p}_{\alpha}\chi_\alpha(Y) $ where $\textrm{NBW}_{\ell}(i,j)$ is the set of length  $\ell$ non-backtracking walk between $i,j$ and $\hat{p}_{\alpha}=\E_c F_{c,\alpha}$}
	\For{$t\gets 1$ \KwTo $ C$}{
		Sample a random coloring $c_t:[n]\mapsto[\ell]$ \;
		Construct a $\mathbb{R}^{n \sum_{s=0}^k\ell^s\times n\sum_{s=0}^k\ell^s}$ matrix $M$, with rows and columns are indexed by $(v,S)$, where $v\in [n]$ and $S$ is an ordered subset of $[\ell]$ with size bounded by $k$ \;
		
		A matrix $H\in \mathbb{R}^{n\times  n\sum_{s=0}^k\ell^s}$ with rows indexed by $[n]$ and each column indexed by $(v,S)$ where $v\in [n]$ and $S$ is ordered subset of $[\ell]$ with size bounded by $k$\;
		
		A matrix $N\in \mathbb{R}^{n\sum_{s=0}^k\ell^s\times n}$, with columns indexed by $[n]$ and each row indexed by $(v,S)$ where $v\in [n]$ and $S$ is ordered subset of $[\ell]$ with size bounded by $k$\;
		Record $p_{c_t}=HM^{\ell-2}N$\;
	}
	Return $\sum_{t}p_{c_t}/C$ 
	\caption{Algorithm for evaluating color-coding non-backtracking walk matrix}\label{algoNonBacktracking}
\end{algorithm}
We now describe how to construct matrix $H,M,N$. For matrix $M$, corresponding to index $\left((v_1,S),(v_2,T)\right)$, the entry is given by $Y_{v_1,v_2}$ if 
\begin{itemize}
	\item the color of $v_1$ is not contained in $S$ and the color of $v_2$ is not contained in $T$
	\item ordered set $T$ is the concatenation of  color of $v_1$ and first $\ell-1$ elements of $S$.
\end{itemize}
Otherwise the entry is given by $0$.

For matrix $H$, corresponding to entry $(v_1,(v_2,S))$, the entry is given by $Y_{v_1,v_2}$ if $S$ contains single element: the color of $v_1$ and the color of $v_2$ is different with the color of $v_1$.

For matrix $N$, corresponding to entry $((v_1,S),v_2)$, the entry is given by $Y_{v_1,v_2}$ if 
\begin{itemize}
	\item $S$ contains $k$ colors and the color of $v_1$ is different from the last color of $S$
	\item the color of $v_2$ is different from the color of $v_1$ and first $k-1$ elements of $S$
\end{itemize}
and given by $0$ otherwise.

Conditioning on some assumptions, we can sample  a random vector in top-$\delta^{-O(1)}$ span of such matrix in quasilinear time. 


\begin{theorem}[Evaluation of color-coding Non-backtracking walk estimator]
	
	In spiked matrix model $Y=\lambda xx^T+W$, vector $x\in\mathbb{R}^n$ has norm $\sqrt{n}$ and entries in $W\in\mathbb{R}^{n\times n}$ are independently sampled with zero mean and unit variance. Considering $\ell=O(\log_{1+\delta} n)$ with $\delta=\Omega(1)$, we assume that 
	the distribution satisfies the following:
	$$\sum_{\alpha,\beta\in \textrm{NBW}_{\ell}(i,j)}\E\left[\chi_\alpha(Y)\chi_\beta(Y) F_{c,\alpha}F_{c,\beta}\right]=\exp(O(k))\left(\sum_{\alpha,\beta\in \textrm{NBW}_{\ell}(i,j)}\E\left[\chi_\alpha(Y)\chi_\beta(Y) \E_c F_{c,\alpha}\E_c F_{c,\beta}\right]\right)$$
	Then for $P(Y)$ defined as equation \ref{algoNonBacktracking}, a unit norm random vector in
	the span of top $\delta^{-O(1)}$ eigenvectors of $P(Y)$ can be sampled in time  $O(n^2\log^{k+1} (n) \delta^{-O(1)} \exp(O(k)))$ if $\ell=O(\log n)$. 
	
	If we assume that 
	$k$-step non-backtracking walk matrix achieves correlation $\delta$ with $xx^\top$, then this random vector $\xi$ achieves constant correlation with $x$: $\lprod \xi,x\rprod^2=\delta^{O(1)}$n. 
\end{theorem}
\textbf{Remark}: The assumption will be satisfied if for all $\alpha,\beta\in\textrm{NBW}_{\ell}(i,j)$, $\E [\chi_\alpha(Y)\chi_\beta(Y)]\geq 0$.
\begin{proof}
	First given single random coloring $c(v):[n]\to [\ell]$, the algorithm evaluates  matrix polynomial $p_{c}(Y)\in\mathbb{R}^{n\times n}$ with each entry given by
	\begin{equation*}
	p_{c,i,j}(Y)=\sum_{\alpha\in \textrm{NBW}(i,j)} \chi_\alpha(Y) F_{c,\alpha}
	\end{equation*}
	where $F_{c,\alpha}$ is the indicator of random event that  each length $k$ walk as a chunk of $\alpha$
	is colorful.  First we note that $ p_{c,i,j}(Y) $ is unbiased estimator for $P_{ij}(Y)$.
	
	Therefore we have 
	\begin{align*}
	\E \left[p_{c,i,j}^2(Y)\right] & =\E\left[\sum_{\alpha,\beta\in \textrm{NBW}(i,j)} \chi_\alpha(Y)\chi_\beta(Y) F_{c,\alpha}F_{c,\beta}\right]\\
	& \leq \exp(O(k))\E_Y\left[\left(\sum_{\alpha\in \textrm{NBW}(i,j)} \E_c\left[ \chi_\alpha(Y)F_{c,\alpha}\right]\right)^2\right]
	\end{align*}
	where we use the assumption. Therefore
	\begin{equation*}
	\E_{Y,c}\left[\left(p_{c,i,j}(Y)-\E_c p_{c,i,j}(Y)\right)^2\right] \leq \exp(O(k)) \E_Y \left[\left(\E_c p_{c,i,j}(Y)\right)^2\right] 
	\end{equation*}
	By averaging $p_{c,i,j}(Y)$ for $C=\omega(\exp(O(k)))$ random colorings, we have
	\begin{equation*}
	\E_{Y,c_1,\ldots,c_C} \left[\left(\frac{1}{C}\sum_{t,i,j}p_{c_t,i,j}(Y)-\E_c p_{c,i,j}(Y)\right)^2\right]\leq \frac{1}{C}\exp(O(k))\E_{Y}\left(\E_c P_{c,i,j}(Y) \right)^2
	\end{equation*}
	Therefore let $p_{ij}(Y)=\frac{1}{C}\sum_{t,i,j}p_{c_t,i,j}(Y)$, we have $$\E p_{ij}^2(Y)=(1+o(1))( \E_Y\left[(\E_c  p_{c,i,j}(Y))^2 \right])=(1+o(1))\left((\E P_{ij}(Y))^2\right)$$
	As a result, we can use $p_{ij}(Y)$ as substitute for $P_{ij}(Y)$. 
	
	For extracting the span of top $\delta^{-O(1)}$ eigenvectors, we apply power method. Since $p(Y)$ can be represented as a sum of chain product of matrices, we can iteratively apply matrix-vector product rather than obtaining $p(Y)$ explicitly. Since for matrix $H,M,N$, there are at most $n^2\log^{k+1}(n)$ non-zero elements.The resulting complexity is thus given by $O(n^2\log^{k+1} (n) \delta^{-1} \exp(O(k)))$
\end{proof}

\section{Order-3 spiked tensor model}

\subsection{Strong detection algorithm for spiked tensor model}\label{appenDetection}

For spiked tensor model, we also consider strong detection problem, which is closely related to weak recovery problem. Specifically given tensor $Y$ sampled from general spiked tensor model, we want to detect whether it's sampled with $\lambda=0$ or large $\lambda$ with high probability. 
\begin{definition}[Strong detection]
 Given tensor $Y$ sampled from planted distribution $\mathbb{P}$ or null distribution $\mathbb{Q}$ with equal probability, we need to find a function of entries in $Y$: $f(Y)\in\{0,1\}$ such that
\begin{equation*}
    \frac{1}{2}\mathbb{P}[f(Y)=1]+\frac{1}{2}\mathbb{Q}[f(Y)=0]=1-o(1)
\end{equation*}
\end{definition}

It's not explicitly stated  in previous literature how to obtain strong detection algorithm via low degree method. The following self-clear fact provides a systematic way for doing so. 
\begin{theorem}[Low degree polynomial thresholding algorithm]\label{thm:detectionMetaTheorem}
 Given $Y$ sampled from $\mathbb{P}$ and $\mathbb{Q}$ with equal probability, For  polynomial $P(Y)=\sum_{\alpha\in S_\ell}\hat{\mu}_\alpha\chi_\alpha(Y)$ where $\{\chi_\alpha(Y):\alpha \in S_\ell\}$ is a set of  polynomial basis orthonormal under measure $\mathbb{Q}$ and $\hat{\mu}_\alpha$ is the Fourier coefficient of likelihood ratio between planted and null distribution $\mu(Y)=\frac{\mathbb{P}(Y)}{\mathbb{Q}(Y)}$ corresponding to basis $\chi_\alpha(Y)$: $\hat{\mu}_\alpha=\E_{\mathbb{P}} \chi_\alpha(Y)$, if we have diverged low degree likelihood ratio $\sum_{\alpha\in S_\ell} \hat{\mu}_\alpha^2 =\omega(1)$ and concentration property $\E_{\mathbb{P}} P^2(Y)=(1+o(1)) \left(\E_{\mathbb{P}} P(Y)\right)^2$, then this implies strong detection algorithm by thresholding polynomial $P(Y)$. 
\end{theorem}
\begin{proof}
    Since $\sum_{\alpha\in S_\ell} \hat{\mu}_\alpha^2 =\omega(1)$, we have $\E_{\mathbb{P}} P(Y)=\E_{\mathbb{Q}} P^2(Y)=\omega(1)$. By Chebyshev's inequality, for $Y\sim \mathbb{P}$ w.h.p we have $P(Y)=(1\pm o(1)) \E_{\mathbb{P}} P(Y)=\omega(\sqrt{\E_{\mathbb{Q}} P^2(Y)})$ while for $Y\sim \mathbb{Q}$ w.h.p we have $P(Y)=O(\sqrt{\E_{\mathbb{Q}} P^2(Y)})$
\end{proof}

Our guarantee for strong detection in spiked tensor model can be stated as following:
\begin{theorem}
\label{thm:detection-tensor}
  Let $x \in \R^n$ be a random vector with independent, mean-zero entries having $\E x_i^2 = 1$ and $\E x_i^4 \leq n^{o(1)}$. Let $\lambda > 0$.
  Let $Y = \lambda \cdot x^{\otimes 3} + W$, where $W \in \R^{n \times n \times n}$ has independent, mean-zero entries with $\E W_{ijk}^2 = 1$.
  Then for $c\geq n^{-1/8+o(1)}$ and $\lambda \geq c n^{-3/4}$, there is $n^{O(1/c^4)}$ time algorithm achieving strong detection.
\end{theorem}

For strong detection algorithm, the thresholding polynomial we use is given by the following. 

\begin{definition}[Thresholding polynomial for strong detection]\label{3DetectionPolynomial}
On directed complete 3-uniform  hypergraph with $n$ vertices, we define $S_{\ell,v}$ as  the set of all copies of  $2$-regular hypergraphs generated in the following way:
\begin{itemize}
    \item we construct $2\ell$ levels of distinct vertices labeled by $0,1,\ldots,2\ell-1$. For levels $t\in [0,2\ell-1]$, it contains $v$ vertices if $t$ is even and $2v$ vertices if $t$ is odd.
    \item Then we construct a perfect matching between levels $t,t-1$ for $t\in [2\ell-1]$ and  between levels $0,2\ell-1$. For each hyperedge, $1$ 
    vertex comes from even level while $2$ vertices come from odd level. The hyperedges are directed from level $0$ to  level $2\ell-1$ and from level $t$ to level $t+1$ for $t\in [0,2\ell-2]$.(An example of such   construction is illustrated in figure \ref{fig:tensorLayerDetection}).
\end{itemize}
\begin{figure}
    \centering
    \includegraphics[width=0.8\textwidth]{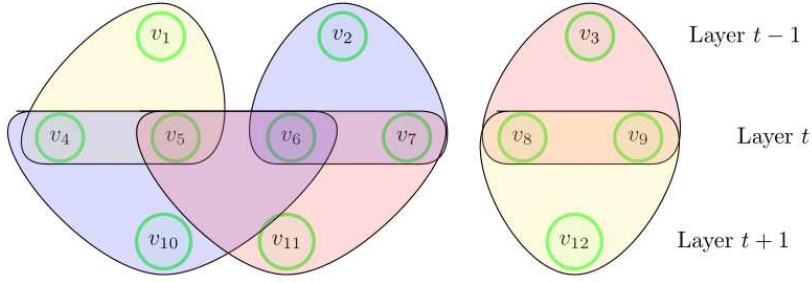}
    \caption{An example of possible directed hyperedge connection  between adjacent layers $t-1,t,t+1$ for an hypergraph $\alpha\in S_{\ell,3}$.
    (1) When $t\in [2\ell-2]$, the direction of hyperedges are given by $(v_1,v_4,v_5),(v_2,v_6,v_7),(v_3,v_8,v_9),(v_4,v_6,v_{10}), (v_5,v_7,v_{11}), (v_8,v_9,v_{12})$. (Note that each hyperedge directs from the layer $t-1$ to the layer $t$ or from the layer $t$ to the layer $t+1$.)
    (2) When $t=0$,the layers are given by $2\ell-2,0,1$ by periodic indexing. The directions of hyperedges are given by $(v_4,v_5,v_1),(v_6,v_7,v_2),(v_8,v_9,v_3),(v_4,v_6,v_{10}), (v_5,v_7,v_{11}), (v_8,v_9,v_{12})$.}
    \label{fig:tensorLayerDetection}
\end{figure}

    Given tensor $Y\in\mathbb{R}^{n\times n\times n}$, the degree $2\ell v$ polynomial $P(Y)$ is given by $P(Y)=\sum_{\alpha\in S_{\ell,v}} \chi_\alpha(Y)$, where $\chi_\alpha(Y)$ is the corresponding 
    multilinear polynomial basis $\chi_\alpha(Y)=\prod_{(i,j,k)\in\alpha} Y_{ijk}$. 
\end{definition}

For simplicity of formulation, we use periodic index below(i.e for  level $t=-1$ we mean level $2\ell-1$).

For proving strong detection guarantee, we first need two hypergraph properties.
\begin{lemma}\label{hypergraph3setSize}
On directed complete hypergraph with $n$ vertices, the number of hypergraphs contained in the set $S_{\ell,v}$ defined in \ref{3DetectionPolynomial} is given by
    $$\lvert S_{\ell,v}\rvert=(1-o(1)) \left({n\choose v}{n\choose 2v}((2v)!)^2\right)^{\ell}$$
\end{lemma}
\begin{proof}
For fixed vertices, between level $t$ and level $t+1$ there are $\frac{(2v)!v!}{v!}$ ways of connecting hyperedges, therefore we have $\lvert S_{\ell,v}\rvert=(1-o(1)) \left({n\choose v}{n\choose 2v}\left(\frac{v!(2v)!}{v!}\right)^2\right)^{\ell}$.
\end{proof}

\begin{lemma}\label{hypergraphSimpleRelation}
    For a pair of hypergraphs $\alpha,\beta\in S_{\ell,v}$, we denote the number of shared vertices as $r$, the number of shared hyperedges as $k$, and the number of shared vertices with degree $0$ or $1$ in $\alpha\cap\beta$ as $s$. Then we have relation
    $2r-s\geq 3k$. 
    
    Further if $2r=3k$, then 
    \begin{itemize}
        \item 
     either $\alpha,\beta$ are disjoint $k=r=0$
      \item or for all levels $t\in [0, 2\ell-1]$, $k_t\geq 1$ and are equal to the same value. 
    \end{itemize}
\end{lemma}
\begin{proof}
    In $3$-uniform sub-hypergraph $\alpha\cap\beta$, there are $k$ hyperedges, $r-s$ vertices with degree $2$ and at most $s$ vertices with degree $1$. By degree constraint, we have $3k\leq 2(r-s)+s=2r-s$. 
    
    When $2r=3k$, each shared vertice between $\alpha,\beta$ has degree $2$ in $\alpha\cap\beta$.  However if there exists $t\in [0,2\ell-1]$ such that $k_{t-1}\neq k_{t}$, then there are shared vertices at level $t-1$ with degree $1$ or $0$ in subgraph $\alpha\cap\beta$. Thus either $\alpha,\beta$ are disjoint($k_t=0$), or for all levels $t$,$k_t\geq 1$ and are all equal.
\end{proof}

We consider the set of hypergraph pairs $\alpha,\beta\in S_{\ell,v}$  such that in $\alpha$
\begin{itemize}
    \item at level $t$   there are $r_t$  vertices shared with $\beta$
    \item between level $t+1$ and level $t$, there are $k_t$  hyperedges  shared with $\beta$.
\end{itemize}
We denote such set of hypergraph pairs  as $S_{\ell,v,k,r}$, where $k=\sum k_t,r=\sum r_t$.
Then $r$ is just the number of shared vertices between $\alpha,\beta$ as $r$ and $k$ is just the number of shared hyperedges between $\alpha,\beta$. Although we abuse the notations(since the set $S_{\ell,v,k,r}$ is related to $k_t,r_t$),
 by the following lemma we can bound the size of such set only using $k,r,\ell,v$. 
\begin{lemma}\label{tensor3PolynomialCorrelation}
    On directed complete hypergraph with $n$ vertices,  for any set $S_{\ell,v,k,r}$ with $\ell=O(\log n)$, $v=o(n)$, $k,r\leq \ell v$, the number of hypergraph pairs contained in the set $S_{\ell,v,k,r}$ is bounded by
    \begin{equation*} \lvert S_{\ell,v,k,r}\rvert\leq \lvert S_{\ell,v}\rvert^2 n^{-r}v^{2r-4k}\ell^{2r-3k} v^{k/2}\exp (O(r))  
    \end{equation*}
\end{lemma}
\begin{proof}
We generate pairs of hypergraph $\alpha,\beta\in S_{\ell,v}$ in the following way: we first choose $\alpha\cap \beta$ and shared vertices as a subgraph of hypergraph in
$S_{\ell,v}$, and then choose remaining graph respectively for $\alpha,\beta$. In hypergraph $\alpha$, suppose there are $r_t$ shared vertices in level  $t$ and $k_t$ shared hyperedges between level $t+1$ and  level $t$.

 We define parity function $\delta(t)=2$ if $t$ is odd and $1$ if $t$ is even. Then we have relation $r_t \geq \delta(t)\textrm{max}(k_{t-1},k_t)$. For choice of $\alpha\cap\beta$,there are
 \begin{align*}
     N_{\alpha\cap\beta}&=\prod_{t=0}^{\ell}{n\choose r_{2t}}{n\choose r_{2t+1}} {r_{2t}\choose k_{2t}} {r_{2t}\choose k_{2t-1}} {r_{2t+1}\choose 2k_{2t}}{r_{2t+1}\choose 2k_{2t+1}} (2k_{2t})!
     (2k_{2t+1})!\\
     &\leq n^{r} v^{k/2} \exp (O(r))
 \end{align*}
 such subgraphs. On the other hand, the number of choices for the remaining hypergraph of $\alpha$ is bounded by
 \begin{align*}
     N_{\alpha-\beta} &=\prod_{t=0}^{\ell-1} {n\choose  v-r_{2t}}{n\choose 2v-r_{2t+1}}(2(v-k_{2t}))! (2(v-k_{2t+1}))!\\
      &= \lvert S_{\ell,v}\rvert n^{-r}v^{r-2k}\exp(O(r))
 \end{align*}
 
 Then we consider the number of $\beta$. Denote the  number of degree-$1$ vertices in $\alpha\cap\beta$ as $s_1$ and the number of shared vertices not contained in $\alpha\cap\beta$ as $s_0$. Let $s=s_0+s_1$, then there are at most $\ell^s$ ways of embedding $\alpha\cap \beta$ in $\beta$. Choosing the remaining hypergraph of $\beta$ is also bounded by $\lvert S_{\ell,v}\rvert n^{-r}v^{r-2k}\exp(-O(k))$. Therefore, with respect to fixed number of vertices and hyperedges $r_t,k_t$ in each level of $\alpha\cap\beta$, the total number of such hypergraph pairs $S_{r,k,\ell,v}$ is bounded by
 \begin{equation*}
     \sum_{s=0}^{2r-3k}\left[\lvert S_{\ell,v}\rvert^2 n^{-r}v^{2r-4k}\ell^s \prod_{t=0}^{2\ell-1} \frac{r_{t}!}{k_t !}\exp (O(r))\right]\leq   \lvert S_{\ell,v}\rvert^2 n^{-r}v^{2r-4k}\ell^{2r-3k} v^{k/2}\exp (O(r)) 
 \end{equation*}
\end{proof}

Next we prove the strong detection through lemma \ref{diverge_density} and \ref{concentration_detection}.

\begin{lemma}\label{diverge_density}
When $\gamma=0.001 n^{3/2}v^{1/2}\lambda^2=1+\Omega(1)$, $\ell=\Omega(\log_{\gamma}n)$ and $n=\omega(v^2\text{poly}(\ell\Gamma))$,
    the projection of likelihood ratio $\mu(Y)$ with respect to $S_{\ell,v}$ diverges, i.e $\sum _{\alpha\in S_{\ell,v}} \hat{\mu}_\alpha^2=\omega(1)$. 
\end{lemma}\label{Tensor3LowDegreeRatio}
\begin{proof}
    Because for any hypergraph $\alpha\in S_{\ell,v}$, we have $\lvert \alpha\rvert=2\ell v$ and $\hat{\mu}_\alpha=\E \chi_\alpha(Y)=\lambda^{2\ell v}$. By lemma \ref{hypergraph3setSize}, 
 we have \[\sum_{\alpha\in S_{\ell,v}} \hat{\mu}_\alpha^2\geq (1-o(1)) \left({n\choose v}{n\choose 2v}((2v)!)^2\right)^{\ell}\lambda^{4\ell v}\geq  n^{3\ell v}v^{\ell v} \lambda^{4\ell v}\].    
    
    For $\gamma=1+\Omega(1)$ and $\ell=\Omega(\log_{\gamma} n)$, we have $\sum_{\alpha\in S_{\ell,v}} \hat{\mu}_\alpha^2=n^{\Omega(1)}$. 
\end{proof}

For proving $\left(\E P(Y)\right)^2=(1-o(1))\E P^2(Y)$, we first prove several preliminary lemmas. First we bound the expectation  of $P(Y)$
\begin{lemma}\label{tensor3PolynomialExpectation}
    In spiked tensor model $Y=\lambda x^{\otimes 3}+W$ where entries in $x\in\mathbb{R}^n$ are sampled independently with zero mean and unit variance, entries in $W\in\mathbb{R}^{n\times n\times n}$ are sampled independently with zero mean and unit variance. Then for $P(Y)$ defined in \ref{3DetectionPolynomial}, we have $\E P(Y)=(1-o(1))\lambda^{2\ell v}({n\choose v}{n\choose 2v}((2v)!)^2)^{\ell}$. 
\end{lemma}
\begin{proof}
 First for each $\alpha\in S_{\ell,v}$, we have $\E[\chi_\alpha(Y)]=\lambda^{2\ell v}$. 
 By lemma \ref{hypergraph3setSize}, we have $\lvert S_{\ell,v}\rvert=(1-o(1)) \left({n\choose v}{n\choose 2v}((2v)!)^2\right)^{\ell}$. Since $\sum_{\alpha\in S_{\ell,v}} \E[\chi_\alpha(Y)]=\lambda^{2\ell v}\lvert S_{\ell,v}\rvert$, we get the lemma. 
\end{proof}

Finally we need a lemma for bounding the summation over all possible $k_t,r_t$
\begin{lemma}\label{summationBounding1}
For $t\in \{0,1,\ldots,2\ell-1\}$, we define $r_t,k_t\in \{0,1,\ldots,2v\}$ satisfying that $r_t\geq \delta(t)\max (k_{t-1},k_t)$, where parity function $\delta(t)=2$ if $t$ is odd and $1$ if $t$ is even. We denote $k=\sum_{t=0}^{2\ell-1} k_t$ and $r=\sum_{t=0}^{2\ell-1} r_t$. We take scalars $\eta=\omega(\ell^2)$ and constant $\psi>1$. Then for $2r\geq 3k+1$,  we have:
\begin{equation*}
   \sum_{k_t\geq 0}\sum_{\substack{r_{t}\geq \delta(t)\max (k_{t-1},k_t) \\ 2r\geq 3k+1}} \eta^{-r+3k/2}\psi^{k}\leq o(1)
\end{equation*}
\end{lemma}
\begin{proof}
We note that given $k_t$ for $t\in [\ell]$, we have  $\ell^{r-3k/2}$ choices for $r_t$. Further we denote $k_\Delta=\sum_{t}\lvert k_{t+1}-k_t\rvert$. Then given $k_\Delta$,  all $k_t$ can take at most $k_\Delta$ different values. As a result, fixing these  $k_\Delta$ different values, there are $\ell^{k_\Delta}$ choices for $k_t$ for $t\in [\ell]$. Further we have $r-3k/2\geq k_\Delta/2$. Therefore the summation is bounded by
     \begin{equation*}
          \sum_{k_\Delta\geq  1}\left[ \left(\frac{\eta}{\ell^2}\right)^{- k_\Delta/2}\prod_{t=1}^{k_\Delta}\left(\sum_{k_t\geq 0}\psi^{k_t}\right)\right]=o(1)
     \end{equation*}
\end{proof}

\begin{lemma}\label{concentration_detection}
Denote $\gamma=0.001 n^{3/2}v^{1/2}\lambda^2$ and take $\ell=\Omega(\log_{\gamma}n)$ in the above estimator $P(Y)$, if we have $\gamma=1+\Omega(1)$ and $n=\omega(v^2\text{poly}(\ell\Gamma))$, then $\E P^2(Y)=(1+o(1)) (\E P(Y))^2$. 
\end{lemma}
\begin{proof}

We need to show that 
\begin{equation*}
    \left(\sum_{\alpha\in S_{\ell,v}} \E[\chi_\alpha(Y)]\right)^2= (1-o(1))\left(\sum_{\alpha,\beta\in S_{\ell,v}} \E[\chi_\alpha(Y) \chi_\beta(Y)]\right)
\end{equation*}

 For left hand side, we already have
 lemma \ref{tensor3PolynomialExpectation}. Thus we only need to bound $\left(\sum_{\alpha,\beta\in S_{\ell,v}} \E[\chi_\alpha(Y) \chi_\beta(Y)]\right)$. First by direct computation we have 
 \begin{equation*}
     \E[\chi_\alpha(Y) \chi_\beta(Y)]=(1+n^{-\Omega(1)})^{k}\lambda^{4\ell v-2k}\E\left[\prod_{i\in \alpha\Delta\beta} x_i^{\text{deg}(i,\alpha\Delta\beta)}\right]\leq \lambda^{4\ell v-2k}\Gamma^{2r-3k}
 \end{equation*}
 where $\text{deg}(i,\alpha)$ is the degree of vertex $i$ in hypergraph $\alpha$, $r$ is the number of shared vertices and $k$ is the number of shared hyperedges, $\Gamma=\E[x_i^4]=n^{o(1)}$ according to assumptions. 
 Using lemma \ref{tensor3PolynomialCorrelation}, lemma \ref{tensor3PolynomialExpectation}, for any set $S_{r,k,\ell,v}$, we have
 \begin{equation*}
     \frac{\sum_{\substack{\alpha,\beta\in S_{r,k,\ell,v}}}\E \chi_\alpha(Y)\chi_\beta(Y) }{\left(\E P(Y)\right)^2} \leq n^{-r}v^{2r-4k}\ell^{2r-3k}\lambda^{-2k}\Gamma^{2r-3k} v^{\frac{k}{2}} \exp(cr)
 \end{equation*}

 where $c$ is large enough constant.  Summing up for different $r_t,k_t$ and combining the fact that if $r=\frac{3}{2}k$ then $k\geq 2\ell,k_t\geq 1$, then we have
 \begin{align*}
    & \frac{\sum_{\alpha,\beta\in S_{\ell,v}}\E\chi_\alpha(Y)\chi_\beta(Y)}{(\E P(Y))^2}
    \\
      = & \sum_{r_t,k_t}c^{3k/2}n^{-3k/2} v^{-k/2}\lambda^{-2k}  \left(\frac{n}{cv^2\Gamma^2}\right)^{-r+3k/2} \ell^{2r-3k}\\ =& \sum_{k_t\geq 0}\sum_{\substack{r_{t}\geq \delta(t)\max (k_{t-1},k_t) \\ 2r\geq 3k+1}} \left(\frac{n}{cv^2\Gamma^2}\right)^{-r+3k/2}\left(c^{3/2}n^{-3/2} \lambda^{-2}
      v^{-1/2}\right)^{k}+ \\
      & \sum_{k_t\geq 1} \left(c^{3/2}n^{-3/2} \lambda^{-2} v^{-1/2}\right)^{\sum_t k_t}+1
      \end{align*}
      where parity function $\delta(t)=2$ if $t$ is odd and $1$ if $t$ is even.
     The term $1$ comes from the case $r=k=0$. When $\gamma=0.001 n^{3/2} \lambda^{2} v^{1/2}>1$ and $\ell=C\log_{\gamma} n$ with constant $C$ large enough, the second term is bounded by $n^{-\Omega(1)}$. Since $n=\omega(cv^2\Gamma^2 \textrm{poly}(\ell))$, by lemma \ref{summationBounding1} the first term is bounded by $o(1)$. Thus we get the theorem.
     
\end{proof}

\begin{proof}[Proof of theorem \ref{thm:detection-tensor}]
Combining lemma \ref{Tensor3LowDegreeRatio}, \ref{concentration_detection}, and theorem \ref{thm:detectionMetaTheorem}, thresholding $P(Y)$ will lead to strong detection algorithm. 
\end{proof}

\subsection{Proof of weak recovery in spiked tensor model}\label{proofWeakRecovery}
We define the notion of weak recovery and strong recovery in spiked tensor model.
\begin{definition}
  In spiked tensor model $Y=\lambda x^{\otimes p}+W$ for random vector $x\in\mathbb{R}^n$ and random tensor $W\in (\mathbb{R}^n)^{\otimes p}$. We define that estimator $\hat{x}(Y)\in \mathbb{R}^n$ achieves weak recovery if $\E \lprod\hat{x}(Y),x\rprod\geq \Omega\left( \left(\E\lVert\hat{x}(Y)\rVert^2\right)^{1/2} \left(\E\lVert x\rVert^2\right)^{1/2}\right)$. Further we define that $\hat{x}(Y)\in \mathbb{R}^n$ achieves strong recovery if $\E \lprod\hat{x}(Y),x\rprod\geq (
 1-o(1))\left(\E\lVert\hat{x}(Y)\rVert^2\right)^{1/2} \left(\E\lVert x\rVert^2\right)^{1/2}$
\end{definition}

The estimator $P(Y)\in\mathbb{R}^n$ we take is defined as following. 
\begin{definition}[Polynomial estimator for weak recovery]\label{defTensorNetworkEstimator}
  On directed complete $3$-uniform hypergraph with $n$ vertices and for $i\in [n]$, we define $S_{\ell,v,i}$ as the set of all copies of hypergraphs generated in the following way:
\begin{itemize}
      \item we construct $2\ell$ levels of distinct vertex. Level $0$ contains vertice $i$ and $(v-1)/2$ vertex in addition. For $0<t<\ell$, level $2t$ contains $v$ vertex and level $2t-1$ contains $2v$ vertex. Level $2\ell-1$ contains $v$ vertex.
      \item We construct a perfect matching between level $t,t+1$ for $t\in [2\ell-2]$. For each hyperedge, $1$ vertice comes from even level while $2$ vertex come from odd level. Each hyperedge directs from level $t$ to level $t+1$.(They are connected in the same way as $S_{\ell,v}$ of the  strong detection case, which is demonstrated in figure \ref{fig:tensorLayerDetection}.)
      \item Level $0$ and $1$ are bipartitely connected s.t each vertice in level $0$ excluding $i$ has degree $2$ while vertice $i$ and vertex in level $1$ has degree $1$. Level $2\ell-2$ and level $2\ell-1$ are bipartitely connected s.t vertex in level $2\ell-1$ have degree $2$ while vertex in level $2\ell-2$ have degree $1$.
  \end{itemize}

  Then given tensor $Y\in\mathbb{R}^{n\times n\times n}$, we have estimator $P(Y)\in\mathbb{R}^n$ where each entry is degree $(2\ell-1) v$ polynomial of entries in $Y$. For $i\in [n]$, the $i$-th entry is given by $P_i(Y)=\sum_{\alpha\in S_{\ell,v,i}} \chi_\alpha(Y)$, where $\chi_\alpha(Y)$ is multilinear polynomial basis $\chi_\alpha(Y)=\prod_{(i,j,k)\in\alpha} Y_{ijk}$ and $S_{\ell,v,i}$ is the set of  hypergraphs defined above.
\end{definition}

We prove that estimator $P(Y)$ defined in \ref{defTensorNetworkEstimator} achieves constant correlation with the second moment $xx^\top$. The proof is very similar to the proof of strong detection algorithm. 
\begin{lemma}\label{expectationRecovery3Tensor}
 In spiked tensor model, $Y=\lambda x^{\otimes 3}+W$, where the entries in  $x\in\mathbb{R}^n$, $W\in\mathbb{R}^{n\times n}$ are independently sampled with zero mean and unit variance, we consider estimator $P(Y)\in\mathbb{R}^n$ defined above with $\gamma=0.001 n^{3/2}v^{1/2}\lambda^{2}=1+\Omega(1)$ and $\ell=O(\log_{\gamma} n)$. Then we have
 \begin{equation*}
    \E[P_i(Y)x_i]=(1-o(1)) \lambda^{2(\ell-1) v}\left({n\choose v}{n\choose 2v}\right)^{\ell-1}{n\choose (v-1)/2}{n\choose v}\frac{((2v)!)^{2\ell-1}}{2^{(3v-1)/2}} 
 \end{equation*}
\end{lemma}
\begin{proof}
Since $\sum_{\alpha\in S_{\ell,v,i}} \E[\chi_\alpha(Y) x_i]=\lambda^{(2\ell-1) v}\lvert S_{\ell,v,i}\rvert$, we only need to bound  the size of $S_{\ell,v,i}$. Applying combinatorial  arguments to the generating process of $S_{\ell,v,i}$, we have \[\lvert S_{\ell,v,i}\rvert= (1-o(1))\left({n\choose v}{n\choose 2v}\right)^{\ell-1}{n\choose (v-1)/2}{n\choose v}\frac{((2v)!)^{2\ell-1}}{2^{(3v-1)/2}}\].  
\end{proof}

\begin{lemma}\label{Tensor3GraphPropertyRecovery}
On the directed complete $3$-uniform hypergraph,  
 for $i\in [n]$ and a set of simple hypergraph $S_{\ell,v,i}$, we consider any hypergraph  $\alpha,\beta\in S_{\ell,v,i}$. Between $\alpha,\beta$, we denote the number of shared vertices(excluding vertice $i$) in level $t$ of $\alpha$ as $r_t$, the number of shared hyperedges between level $t$ and level $t+1$ of $\alpha$ as $k_t$. Further we denote $r=\sum_t r_t$ as the total number of shared vertices excluding $i$ and $k=\sum k_t$ as the total number of shared hyperedges. 
 
 Then one of the following relations must hold:
 \begin{itemize}
     \item $2r>3k$  
     \item $2r=3k$ and $\alpha\cap\beta$ is a hyper-path starting from vertice $i$ or empty.
     \item  $3k \geq 2r\geq 3k-1$  and $k_t\geq 1$ for all $t$
 \end{itemize}
\end{lemma}
\begin{proof}
 Suppose we have $2r\leq 3k-1$. Then by degree constraint, in hypergraph $\alpha\cap\beta$, excluding vertice $i$, all other vertices have degree $2$. This is only possible if for all levels $t$, vertices contained in $\alpha\cap\beta$ are connected to two hyperedges, implying that $k_t,k_{t-1}\geq 1$. Further we have $2r=3k-1$ in the case.
 
 Suppose we have $2r=3k\neq 0$ and there is $t\in [2\ell]$ such that $k_t=0$. Then in the $\alpha\cap\beta$, exactly one vertice(excluding vertice $i$) has degree $1$ and all the other vertices have degree $2$. Thus there is only one level $t^{\prime}$ such that $k_t=0$. Thus all shared vertices in level $t$ have degree at most $1$. This implies that there is exactly one shared vertice in level $t$. This is only possible if $\alpha\cap\beta$ is a hyperpath starting from vertice $i$.  
 \end{proof}

\begin{lemma}
  For any $\alpha,\beta\in S_{\ell,v,i}$ sharing $k$ hyperedges, we have
 \begin{align*}
     \E[\chi_\alpha(Y) \chi_\beta(Y)]& =(1+n^{-\Omega(1)})\lambda^{2(2\ell-1)v-2k}\E\left[\prod_{j\in \alpha\Delta\beta} x_j^{\text{deg}(j,\alpha\Delta\beta)}\right]\\ &\leq \lambda^{-2k}\Gamma^{O(2r-3k)} \E[\chi_\alpha(Y)x_i]\E[\chi_\beta(Y)x_i]
 \end{align*}
 where $\text{deg}(j,\alpha\Delta\beta)$ represents the degree  of vertex $j$ in hypergraph $\alpha\Delta\beta$.
\end{lemma}

\begin{proof}
    This follows from the same computation as \ref{tensor3PolynomialCorrelation}.
\end{proof}

We consider the set of hypergraph pairs $\alpha,\beta\in S_{\ell,v,i}$  such that in $\alpha$
\begin{itemize}
    \item at level $t$   there are $r_t$  vertices(excluding vertice $i$) shared with $\beta$
    \item between level $t+1$ and level $t$, there are $k_t$  hyperedges  shared with $\beta$.
\end{itemize}
We denote such set of hypergraph pairs  as $S_{\ell,v,i,k,r}$, where $k=\sum k_t,r=\sum r_t$.
Then $r$ is just the number of shared vertices between $\alpha,\beta$ as $r$ and $k$ is just the number of shared hyperedges between $\alpha,\beta$. Although we abuse the notations(since the set $S_{\ell,v,k,r}$ is related to $k_t,r_t$),
 by the following lemma we can bound the size of such set only using $k,r,\ell,v$. 
\begin{lemma}\label{Tensor3RecoveryPairCount}
    On directed complete hypergraph with $n$ vertices,  for any set $S_{\ell,v,k,r}$ with $\ell=O(\log n)$, $v=o(n)$, $k,r\leq \ell v$, the number of hypergraph pairs contained in the set $S_{\ell,v,i,k,r}$ is bounded by
    \begin{equation*} \lvert S_{\ell,v,i,k,r}\rvert\leq \lvert S_{\ell,v,i}\rvert^2 n^{-r}v^{2r-7k/2}\ell^{2r-3k} \exp (O(r))  
    \end{equation*}
\end{lemma}
\begin{proof}
We first choose $\alpha\cap \beta$ and shared vertices as subgraph of hypergraph $\alpha\in S_{\ell,v,i}$, and then completing the remaining hypergraphs $\alpha\setminus \beta$ and $\beta\setminus\alpha$.  If in the $\alpha$ there are $r_t$ shared vertices(excluding $i$) in level $t$, and $k_t$ shared hyperedges between level $t$ and level $t+1$ for $t=0,1,\ldots,2\ell-1$, then the number of choices for shared vertices and $\alpha\cap\beta$ is bounded by 
\begin{align*}
    N_{\alpha\cap\beta} \leq &{n\choose r_{0}}{n\choose r_{2\ell-1}}{2r_{0}+1\choose k_0}{2r_{2\ell-1}\choose 2k_{2\ell-2}}\prod_{t=1}^{\ell-1}\left[ {n\choose r_{2t-1}}{n\choose r_{2t}} {r_{2t-1}\choose k_{2t-1}}{r_{2t-1}\choose k_{2t-2}} \right.\\ & {r_{2t}\choose 2k_{2t}}  \left.
    {r_{2t}\choose 2k_{2t-1}}\right] \prod_{t=0}^{2\ell-2} (2k_t)!
\end{align*}
This is upper bounded by $\prod_{t=0}^{2\ell-2} (2k_t)!\prod_{t=0}^{2\ell-1} {n\choose r_t}\exp (O(r))\leq n^{r} v^{k/2} \exp(O(r))$. 
Next we choose the remaining hypergraph $\alpha \setminus \beta$ and  $\beta\setminus \alpha$ respectively. For $\alpha \setminus \beta$, we have 
\begin{align*}
    N_{\alpha\setminus\beta} &={n\choose \frac{v-1}{2}-r_0}{n\choose  v-r_{2\ell-1}}\prod_{t=1}^{\ell-1} {n\choose 2v-r_{2t-1}}{n\choose v-r_{2t}} \prod_{t=0}^{2\ell-2} (2(v-k_t))! \\ 
    & \leq \lvert S_{\ell,v,i}\rvert n^{-r}v^rv^{-2k}\exp(O(r))
\end{align*}
Suppose there are $s_1$ degree $1$ vertices in $\alpha\cap \beta$ and $s_0$ vertices shared between $\alpha,\beta$ but not contained in $\alpha\cap\beta$, denoting $s=s_0+s_1$, then there are $\ell^{s}$ ways of placing $\alpha\cap\beta$ and shared vertices in hypergraph $\beta$ and the count of remaining hypergraph is also bounded by $\lvert S_{\ell,v,i}\rvert n^{-r}v^rv^{-2k}\exp(O(r))$. Multiplying together we will get the claim.
\end{proof}

\begin{lemma}[Recovery for general spiked model]\label{recEstimator}
In spiked tensor model $Y=\lambda x^{\otimes 3}+W$ with the same setting as theorem \ref{thm:intro-tensor}, 
     taking $\gamma=0.001 n^{3/2}v^{1/2}\lambda^{2}=1+\Omega(1)$ and $\ell=O(\log_{\gamma} n)$ in the estimator above, then if $n=\omega(v^2 \Gamma^2\text{polylog}(n))$, we have $\frac{\E\lprod P(Y),x\rprod}{\left(\E\lVert P(Y)\rVert^2\E\lVert x\rVert^2\right)^{1/2}}=\Omega(1)$. 
\end{lemma}
\begin{proof}
We need to show the estimator $P(Y)\in\mathbb{R}^n$ above achieves constant correlation with the hidden vector $x$. Equivalently, we want to show that for each $i\in[n]$
\begin{equation*}
    \left(\sum_{\alpha\in S_{\ell,v,i}} \E[\chi_\alpha(Y) x_i]\right)^2= \Omega\left(\sum_{\alpha,\beta\in S_{\ell,v,i}} \E[\chi_\alpha(Y) \chi_\beta(Y)]\right)
\end{equation*}

For left hand side, we can simply apply \ref{expectationRecovery3Tensor}. For the right hand side, we have
 \begin{equation*}
     \E[\chi_\alpha(Y) \chi_\beta(Y)] \leq \lambda^{-2k}\Gamma^{O(2r-3k)} \E[\chi_\alpha(Y)x_i]\E[\chi_\beta(Y)x_i]
 \end{equation*}
 By lemma, \ref{Tensor3RecoveryPairCount} the contribution to $\frac{\E [P^2_i(Y)]}{\left(\E [P_i(Y)x_i]\right)^2}$ with respect to specific $r_t,k_t$ is bounded by
\begin{equation*}
   (n^{-r}v^rv^{-2k})^2n^rv^{k/2}\ell^s\Gamma^{O(2r-3k)}\lambda^{-2k} \exp(O(r))\leq \left(\frac{n}{cv^2\ell^2\Gamma^{O(1)}}\right)^{-r+3k/2} \left(cn^{-3/2}v^{-1/2}\lambda^{-2}\right)^k  
\end{equation*}
where $c$ is constant. For $nv^{-2}=\omega(\text{poly}(\Gamma\ell))$, using  argument very similar to lemma \ref{summationBounding1}, the dominating term is given by $r\leq \frac{3}{2}k$. For $2r=3k+1$, by lemma \ref{Tensor3GraphPropertyRecovery} we must have $k\geq \ell$, therefore for $cn^{-3/2}v^{-1/2}\lambda^{-2}<1$ and $\ell=C\log n$ with constant $C$ large enough, the contribution is $n^{-\Omega(1)}$. For $2r=3k$, 
by lemma \ref{Tensor3GraphPropertyRecovery}
 either $k\geq 2\ell$ or $\alpha\cap\beta$ exists as a hyperpath starting from 
vertex $i$. The first case can be treated in the same way as  $2r=3k-1$. For the second case, the contribution is bounded by $\sum_{k=0}^{2\ell-1} \left(cn^{-3/2}v^{-1/2}\lambda^{-2}\right)^{k} \leq \frac{1}{1-cn^{-3/2}v^{-1/2}\lambda^{-2}}$

Therefore in all, we have $\frac{\left(\E P_i(Y)x_i\right)^2}{\E P_i^2(Y)}\geq 1-cn^{-3/2}v^{-1/2}\lambda^{-2}$. This is $\Omega(1)$ when we have relation $cn^{-3/2}v^{-1/2}\lambda^{-2}=1-\Omega(1)$ and $n=\omega(v^2 \text{poly}(\Gamma\ell))$. Therefore, such polynomial estimator $P(Y)\in \mathbb{R}^n$ achieving weak recovery under the given condition. 
\end{proof}

\subsection{Color coding method for polynomial evaluation in order-3 spiked tensor model}\label{colorCodingTensor}
\subsubsection{Strong detection polynomial}
For constant $v$, although the thresholding polynomial and polynomial estimator has degree $O(\log n)$, these polynomials can actually be evaluated in polynomial time via color coding method as a generalization of result in  \cite{8104074}. In the same way color coding method also improves the running time of sub-exponential time algorithms.

We first describe the evaluation algorithm for scalar polynomial, as shown in algorithm \ref{algoCircleTensor3}.

\begin{algorithm}
\KwData{Given $Y\in \mathbb{R}^{n\times n\times n}$ s.t $Y=\lambda x^{\otimes 3}+W$}
\KwResult{$P(Y)\in\mathbb{R}$ which is the sum of multilinear monomials corresponding to hypergraphs  in $S_{\ell,v}$(up to accuracy $1+n^{-\Omega(1)}$)}
 \For{$i\gets1$ \KwTo $C$}{
     Sample coloring $c_i:[n]\mapsto [\ell]$ uniformly at random\;
   Construct a matrix $M\in \mathbb{R}^{(2^{3\ell v}-1)n^v\times (2^{3\ell v}-1)n^{2v}}$, rows and columns of $M$ are indexed by $(V_1, S)$ and $(V_2,T)$ where $V_1\in [n]^{v}$ and $V_2\in [n]^{2v}$ are set of vertices while $S,T\subsetneq [3\ell v]$ are subset of colors.\;
    Matrices $Q,N\in \mathbb{R}^{\mathbb{R}^{(2^{3\ell v}-1)n^{2v}\times (2^{3\ell v}-1)n^{v}}}$, the rows and columns of which are indexed by $(V_2, S)$ and $(V_1,T)$ where $V_1\in [n]^{v}$ and $V_2\in [n]^{2v}$ are set of vertices while $S,T\subseteq [3\ell v]$ are non-empty subset of colors.
    
    Record $p_{c_i}= \frac{(3\ell v)^{3\ell_v}}{(3\ell v)!}\text{trace}((MN)^{\ell-1}MQ)$\;}
    Return  $\frac{1}{C}\sum_{i=1}^C p_{c_i}$
 \caption{Algorithm for evaluating the thresholding polynomial }\label{algoCircleTensor3}
 \end{algorithm}
 
  Next we describe the construction of matrices $M,N,Q$.
  
  We have $M_{(V_1,S),(V_2,T)}= 0$ if $S\cup \{c(v):v\in V_1\}\neq T$ or $\{c(v):v\in V_1\}$ and  $S$ are not disjoint. Otherwise $M_{(V_1,S),(V_2,T)}$ is given by $\sum_{\gamma \in S_{V_1,V_2}}\chi_{\gamma}(Y)$ where $S_{V_1,V_2}$ is the set of perfect matching induced by $V_1$ and $V_2$(each hyperedge in $S_{V_1,V_2}$ direct from $1$ vertice from $V_1$ to $2$ vertice from $V_2$). 
  
  In the same way,  We have $N_{(V_2,S),(V_1,T)}= 0$ if $S\cup \{c(v):v\in V_2\}\neq T$ or $\{c(v):v\in V_2\}$ and  $S$ are not disjoint. Otherwise $N_{(V_2,S),(V_1,T)}$ is given by $\sum_{\gamma \in S_{V_1,V_2}}\chi_{\gamma}(Y)$ where $S_{V_1,V_2}$ is the set of perfect matching induced by $V_1$ and $V_2$(each hyperedge in $S_{V_1,V_2}$ direct from $2$ vertice in $V_2$ to $1$ vertice in $V_1$). 
  
  For matrix $Q$, the indexing and non-zero  entry locations are the same as $N$. However the non-zero elements are given by $\sum_{\gamma\in S_{V_2,V_1}}\chi_\gamma(Y)$ where $S_{V_2,V_1}$ is the set of perfect matching induced by $V_1$ and $V_2$(each hyperedge in $S_{V_2,V_1}$ direct from $1$ vertice in $V_1$ to $2$ vertice in $V_2$).
 
\begin{lemma}[Evaluation of thresholding polynomial]\label{evaDetection}
       There exists a $n^{4v}\exp(O(\ell v))$-time algorithm that given a coloring $c$:$[n]\to [3\ell v]$(where $3\ell v$ is the number of vertices in hypergraph $\alpha\in S_{\ell,v}$) and a tensor $Y\in \mathbb{R}^{n\times n\times n}$ evaluates  degree $2\ell v$ polynomial in polynomial time
       \begin{equation}
           p_c(Y)=\mathop{\sum}_{\alpha\in S_{\ell,v}} \chi_\alpha(Y)F_{c,\alpha}
       \end{equation}
       \begin{equation}
    F_{c,\alpha}=\frac{(3\ell v)^{3\ell v}}{(3\ell v)!} \cdot \mathbf{1}_{c(\alpha)=[3\ell v]}
\end{equation}

    when thresholding polynomial $P(Y)$ defined in \ref{3DetectionPolynomial} satisfies $(\E P(Y))^2=(1-o(1))\E P^2(Y)$, we can take $\exp(O(\ell v))$ random colorings and give an accurate estimation of the thresholding polynomial by averaging $p_c(Y)$
    \end{lemma}
    \begin{proof}
    A critical observation is that for each given random coloring $c$, the algorithm above evaluates $p_c(Y)$. 
   we prove that averaging random coloring for $p_c(Y)$ will give  accurate estimate for $P(Y)$. This  follows from the same reasoning as in matrix case. First we note that $\E_c p_c(Y)=P(Y)$. Next, for single coloring we have
    \begin{equation*}
        \E p^2_c(Y)= \sum_{\alpha,\beta\in  S_{\ell,v}}\E [F_{c,\alpha} F_{c,\beta} \chi_\alpha(Y)\chi_\beta(Y)] \leq \exp(O(\ell v)) \E P^2(Y)\leq \exp(O(\ell v)) (\E p_c(Y))^2
    \end{equation*}
    where we use the result that $\E P^2(Y)=(1+o(1)) (\E P(Y))^2$. Therefore, by averaging $L=\exp(O(\ell v))$ random colorings, the variance can be reduced such that $\frac{\sum_{t=1}^L p_{c_t}(Y)}{L}=(1\pm o(1))P(Y)$ w.h.p.  
    \end{proof}
    
  \subsubsection{Evaluation of estimator for weak recovery}  
    Next we discuss about the evaluation of polynomial estimator for weak recovery. Except for matrices $M,N$ as defined above, we need to construct two additional matrices $A\in\mathbb{R}^{n^{(v+1)/2}\times (2^{\ell_v}-1)n^{2v}}$ and $B\in\mathbb{R}^{(2^{\ell_v}-1)n^{v}\times n^{v}}$, where $\ell_v =\frac{3(2\ell-1) v+1}{2}$ is the number of vertices in each hypergraph contained in the set $S_{\ell,v,i}$.
\begin{algorithm}
\KwData{Given $Y\in \mathbb{R}^{n\times n\times n}$ s.t $Y=\lambda x^{\otimes 3}+W$}
\KwResult{$P(Y)\in\mathbb{R}^n$, the $i$-th entry of which is the sum of multilinear monomials corresponding to hypergraphs  in $S_{\ell,v,i}$(up to accuracy $1+n^{-\Omega(1)}$)}
 \For{$i\gets1$ \KwTo $C$}{
     Sample coloring $c_i:[n]\mapsto [\ell]$ uniformly at random\;
   Construct matrices $M,N$ as in the algorithm of strong detection\;
    Construct matrix $A\in \mathbb{R}^{n^{(v+1)/2}\times (2^{\ell_v}-1)n^{2v}}$, where the rows are indexed by $(i,V_1)$($i\in [n]$ and $V_1\in [n]^{(v-1)/2}$), and columns are indexed by $(V_2,S)$($V_2\in [n]^{2v}$ and $S\subseteq [\ell_v]$)\;
    Construct matrix $B\in\mathbb{R}^{(2^{\ell_v}-1)n^{v}\times n^{v}}$, where the rows are indexed by $(V_1,S)$, and columns are indexed by $V_2$\;
    Record $p_{c_i}= A(NM)^{\ell-1}NB\mathbf{1}$\;}
    Return  $\frac{1}{C}\sum_{i=1}^C p_{c_i}$
 \caption{Algorithm for evaluating estimation polynomial vector}\label{algoEstimationTensor3}
 \end{algorithm}
   Then we describe how to construct matrices $A,B$. For $i\in [n]$,set of vertices $V_1$ in level $0$,set of vertices $V_2$ in level $1$ and set of colors $T$, denoting $S_{i,V_1,V_2}$ as the set of all possible connections between level $0$ and $1$, entry $A_{(i,V_1),(V_2,T)}$ is given by $\sum_{\alpha\in S_{i,V_1,V_2}}\chi_\alpha(Y)$ if $T=\{c(v):v\in V_1\cup v\},v\not\in V_1$ and $0$ otherwise. In the same way, denoting $\mathcal{L}_{V_1,V_2}$ as all possible connections between level $2\ell-2$ and $2\ell-1$, we have entry $B_{(V_1,S),V_2}$ given by   $\sum_{\alpha\in \mathcal{L}_{V_1,V_2}}\chi_\alpha(Y)$ if $S\cup \{c(v):v\in V_1\cup V_2\}=[\ell_v], S\cap \{c(v):v\in V_1\cup V_2\}=\emptyset$ and zero otherwise.

    \begin{lemma}[Evaluation of polynomial estimator]\label{evaRec}
        Denote the number of vertices in any hypergraph contained in $S_{\ell,v,i}$ as $\ell_v$, then $\ell_v =\frac{3(2\ell-1) v+1}{2}$. Then
              there exists a $n^{3v}\exp(O(\ell_v))$-time algorithm that given a coloring $c$:$[n]\to [\ell_v]$ and a tensor $Y\in \mathbb{R}^{n\times n\times n}$ evaluates vector $p_c(Y)\in \mathbb{R}^{n}$ with each entry $p_{c,i}(Y)$ a polynomial of entries in $Y$
       \begin{equation*}
           p_{c,i}(Y)=\mathop{\sum}_{\alpha\in S_{\ell,v,i}} \chi_\alpha(Y)F_{c,\alpha}
       \end{equation*}
       \begin{equation*}
    F_{c,\alpha}=\frac{\ell_v^{\ell_v}}{\ell_v!} \cdot \mathbf{1}_{c(\alpha)=[\ell_v]}
    \end{equation*}
    
    For polynomial estimator defined in \ref{defEstimator}, if $0.001\lambda^2 n^{3/2}v^{1/2}>1$, we can take $\exp(O(\ell v))$ random colorings and give an accurate estimation of the estimation polynomial $P_i(Y)$ by averaging $p_{c,i}(Y)$. When $\lambda^2 n^{3/2}=\omega(1)$, we have $n^{3+o(1)}$ time algorithm algorithm for evaluation. 
    \end{lemma}
    \begin{proof}
     The critical observation is that $p_{c,i}(Y)$ can be obtained from vector $\xi=A(NM)^{\ell-1}NB\mathbf{1}$($\mathbf{1}\in \mathbb{R}^{n^v}$ is all-$1$ vector), by summing up all rows in $\xi$ indexed by $(i,\cdot)$. 
    
    By the same argument as in the strong detection algorithm, we can obtain accurate estimate of $P(Y)$ by averaging $\exp(O(\ell))$ random colorings when weak recovery is achieved. Therefore, the estimator can be evaluated in time $n^{O(v)}$ when $0.001\lambda^2 n^{3/2}v^{1/2}>1$. 
    
    Moreover, when $\lambda^2 n^{3/2}=\omega(1)$, it's enough to take $v=1$ and $\ell=o(\log n)$. Thus $\xi=A(NM)^{\ell-1}NB\mathbf{1}$ can be evaluated in $n^{3}\exp(O(\ell))=n^{3+o(1)}$ time by recursively executing matrix-vector multiplication. Therefore the polynomial estimator achieving strong recovery can be evaluated in time $n^{3+o(1)}$.
    \end{proof}
    
    \begin{proof}[Proof of Theorem \ref{thm:intro-tensor}]
When $cn^{-3/2}v^{-1/2}\lambda^{-2}=1-\Omega(1)$ and $n=\omega(v^2 \text{poly}(\Gamma\ell))$, by lemma \ref{recEstimator} $\frac{P(Y)}{\lVert P(Y)\rVert}$ achieves $\Omega(1)$ correlation with $\frac{x}{\lVert x\rVert}$. Further by the color-coding method, according to lemma $P(Y)$ can be evaluated in time $n^{O(v)}$. These proves the claim in theorem \ref{thm:intro-tensor}.
\end{proof}

\subsection{Equivalence between strong  and weak recovery}\label{equivalStrongWeak}
Under some mild conditions, combining concentration argument and 'all or nothing' amplification, we can  actually obtain strong recovery algorithm when $n^{-3/2}v^{-1/2}\lambda^{-2}=\Theta(1)$. 

For this we need an assumption on the tensor injective norm. The injective norm of an order-$p$ tensor $W\in (\mathbb{R}^n)^{\otimes p}$ is defined as
\begin{equation*}
    \lVert W\rVert_{\textrm{inj}}=\max _{\left\|u^{(1)}\right\|=\cdots=\left\|u^{(p)}\right\|=1} \sum_{i_{1}, \ldots, i_{p}} W_{i_{1}, \ldots, i_{p}} u_{i_{1}}^{(1)} \cdots u_{i_{p}}^{(p)}
\end{equation*}

\begin{theorem}[Strong recovery]\label{StrongRecoveryTensor}
 In general spiked tensor model $Y=\lambda x^{\otimes 3}+W$,  we take estimation vector $y\in\mathbb{R}^n$ by setting $\gamma= 0.001 n^{3/2}v^{1/2}\lambda^2=1+\Omega(1)$ and $\ell=O(\log_\gamma n)$ in estimator \ref{defEstimator}. If injective norm of $W$ is  $o(n^{3/4})$ w.h.p, then for constant $v$ if $n=\omega(v^2\text{poly}(\ell\Gamma))$, we have $\hat{x}\in\mathbb{R}^n$ s.t $\hat{x}_i=\sum_{j_1,j_2} Y_{i,j_1,j_2} y_{j_1}y_{j_2}$ achieves strong recovery, i.e we  have $\lprod \frac{\hat{x}}{\lVert \hat{x}\rVert}, \frac{x}{\lVert x\rVert}\rprod=1-o(1)$ with high probability
\end{theorem}

We use assumption that the injective norm of $W$ is  $o(n^{3/4})$ w.h.p. Now we interpret this assumption. For Gaussian tensor the injective norm of $W$ is $o(n^{3/4})$ with high probability. For general tensor, this assumption is weaker than finite bounded moments.
\begin{lemma}
 For a tensor $W\in (\mathbb{R}^n)^{\otimes 3}$, the injective norm is $O(\sqrt{n})$ with high probability if the absolute value of entries in $W$ are all bounded by $B=o(n^{1/4})$ with high probability.
\end{lemma}
\textbf{Remark}: If entries $W_{ijk}$ have bounded $12$-th moment, then by Markov inequality and union bound, the entries in $W$ are all bounded by $B=o(n^{1/4})$ with high probability. 
\begin{proof}
   For fixed unit vectors $x,y,z$ and $T=\sum_{ijk} W_{ijk}x_iy_jz_k$, by Hoeffding bound we have
\begin{equation*}
    \textrm{Pr}\left[T\geq tB^{\prime}\right]\leq \exp(-ct^2)
\end{equation*}
where $c$ is constant and $B^{\prime}$ is the maximum absolute value of entries in $W$. We denote the event $B^{\prime}\leq B$ as $\mathcal{A}$. By assumption $\mathcal{A}$ happens with high probability. The $\epsilon$-net of unit sphere $\mathcal{S}_{\epsilon,n}$ has size at most $\left(O(1)/\epsilon\right)^n$(see e.g Tao's random matrix book lemma 2.3.4). Thus the size of set $\mathcal{B}=\{(x,y,z)| x,y,z\in \mathcal{S}_{\epsilon,n}\}$  is bounded by $\left(O(1)/\epsilon\right)^{3n}$. Taking union bound on this set we have
\begin{equation*}
    \textrm{Pr}[\max_{(x,y,z)\in \mathcal{B}} \lprod W,x\otimes y\otimes z\rprod\geq Bt]\leq \exp(c_2 n) \exp (-ct^2)+o(1)
\end{equation*}
where $c_2$ is constant. Taking $t=C\sqrt{n}$ with constant $C$ large enough, it follows that $\max_{(x,y,z)\in \mathcal{B}} \lprod W,x\otimes y\otimes z\rprod=o(n^{3/4})$ with high probability. 

Finally we have 
$$\lVert W\rVert_{\textrm{inj}}=\max_{\lVert x\rVert=1,\lVert y\rVert=1,\lVert z\rVert=1}\lprod W,x\otimes y\otimes z\rprod=\lprod W,x^{*}\otimes y^{*}\otimes z^{*}\rprod $$
, where $x^{*}, y^{*}, z^{*}$ is the maximizer. By definition of $\epsilon$-net, we can find $(\tilde{x},\tilde{y},\tilde{z})\in \mathcal{B}$ such that $\lVert \tilde{x}-x^{*}\rVert,\lVert \tilde{y}-y^{*}\rVert,\lVert \tilde{z}-z^{*}\rVert\leq  \epsilon$. Thus $\lprod W,x^{*}\otimes y^{*}\otimes z^{*}-\tilde{x}\otimes \tilde{y}\otimes \tilde{z}\rprod\leq 4\epsilon \lprod W,x^{*}\otimes y^{*}\otimes z^{*}\rprod$. For small $\epsilon$, we have $\lVert W\rVert_{\textrm{inj}}\leq 2\max_{(x,y,z)\in \mathcal{B}} \lprod W,x\otimes y\otimes z\rprod =o(n^{3/4})$ with high probability. This proves the claim. 
\end{proof}

This proof of theorem naturally follows from the weak recovery result above and the following two lemmas.
\begin{lemma}[All or nothing phenomenon]\label{allNothing}
 In general spiked tensor model $Y=\lambda x^{\otimes 3}+W$, if the injective norm of tensor $W$ is $o(n^{3/4})$, then if we have unit norm estimator $y\in\mathbb{R}^n$ satisfying $\lprod y, x/\lVert x\rVert\rprod= \Omega(1)$ w.h.p, then let $\hat{x}\in\mathbb{R}^n$ with $\hat{x}_i=\sum_{j_1,j_2} Y_{i,j_1,j_2} y_{j_1}y_{j_2}$, we have w.h.p $\left\lprod \frac{\hat{x}}{\lVert\hat{x}\rVert},\frac{x}{\lVert x\rVert}\right\rprod^2=1-o(1)$
\end{lemma}
This lemma follows the same proof as appendix D in  \cite{Kikuchy}
\begin{lemma}[Concentration property]\label{concentration_lemma}
 For the above estimator $P(Y)\in\mathbb{R}^n$ in definition \ref{defEstimator},when we take $\ell,v$ as described in theorem \ref{recEstimator},we have
 \begin{equation*}
     \E \lprod P(Y),x\rprod^2=(1+o(1)) \left(\E \lprod P(Y),x\rprod\right)^2
 \end{equation*}
\end{lemma}

We now prove lemma \ref{concentration_lemma}. The proof is very similar to the proof of strong detection in appendix   \ref{appenDetection}.
\begin{proof}[Proof of Lemma \ref{concentration_lemma}]
    We denote $S_{\ell,v}=\cup_{i\in [n]} S_{\ell,v,i}$. For $\alpha\in S_{\ell,v}$, we denote $x_\alpha=x_i$ if $\alpha\in S_{\ell,v,i}$. Then equivalently we want to show.
    \begin{equation*}
        \sum_{\alpha,\beta\in S_{\ell,v}}\E[\chi_\alpha(Y)\chi_\beta(Y)x_\alpha x_\beta]\leq (1-o(1))\sum_{\alpha,\beta\in S_{\ell,v}} \E[\chi_\alpha(Y)x_\alpha][\chi_\beta(Y)x_\beta] 
    \end{equation*}
     Since $\left(\sum_{\alpha\in S_{\ell,v}} \E[\chi_\alpha(Y) x_\alpha]\right)=\lambda^{(2\ell-1) v}\lvert S_{\ell,v}\rvert$, we only need to bound the size of $S_{\ell,v}$. Applying combinatorial  arguments to the generating process of $S_{\ell,v}$, we have \[\lvert S_{\ell,v}\rvert= (1-o(1))\left({n\choose v}{n\choose 2v}\right)^{\ell-1}{n\choose (v+1)/2}{n\choose v}\frac{((2v)!)^{2\ell-1}}{2^{(3v-1)/2}}
   \].  

Now we bound the right hand side. First for the case that $\alpha,\beta$ are disjoint($r=0$), we have $\E[\chi_\alpha(Y)\chi_\beta(Y)]=\E[\chi_\alpha(Y)]\E[\chi_\beta(Y)]$. Thus we have 
\begin{equation*}
        \sum_{\alpha,\beta\in S_{\ell,v}}\E[\chi_\alpha(Y)\chi_\beta(Y)x_\alpha x_\beta]\leq (1-o(1))\sum_{\alpha,\beta\in S_{\ell,v}} \E[\chi_\alpha(Y)x_\alpha][\chi_\beta(Y)x_\beta] 
\end{equation*}

For each pair of $\alpha,\beta\in S_{\ell,v}$ sharing $k$ hyperedges and $r$ vertices, we have
 \begin{align*}
     \E[\chi_\alpha(Y) \chi_\beta(Y)]& =(1+n^{-\Omega(1)})\lambda^{2(2\ell-1)v-2k}\E\left[\prod_{j\in \alpha\Delta\beta} x_j^{\text{deg}(j,\alpha\Delta\beta)}\right]\\ &\leq \lambda^{-2k}\Gamma^{O(2r-3k)} \E[\chi_\alpha(Y)]\E[\chi_\beta(Y)]
 \end{align*}
 where $\text{deg}(j,\alpha\Delta\beta)$ represents the degree of vertex $j$ in hypergraph $\alpha\Delta\beta$.

Next we bound the number of hypergraph pairs $(\alpha,\beta)$ sharing specified vertices and hyperedges. For this, we first choose $\alpha\cap \beta$ and shared vertices as a subgraph of a hypergraph $\alpha$ contained in $S_{\ell,v}$. We consider the following case: 
\begin{itemize}
    \item in level $t\in [2\ell-1]$ of $\alpha$ there are $r_t$  shared vertices
    \item between the levels $t$ and $t+1$ of $\alpha$ there are $k_t$ shared hyperedges.
\end{itemize}
Then the number of choices for these shared  hyperedges and vertices is bounded by $N_{\alpha\cap\beta}$
\begin{equation*}
    {n\choose r_{0}}{n\choose r_{2\ell-1}}{2r_{0}\choose k_0}{2r_{2\ell-1}\choose 2k_{2\ell-2}}\prod_{t=1}^{\ell-1} {n\choose r_{2t-1}}{n\choose r_{2t}} {r_{2t-1}\choose k_{2t-1}}{r_{2t-1}\choose k_{2t-2}}{r_{2t}\choose 2k_{2t}}{r_{2t}\choose 2k_{2t-1}} \prod_{t=0}^{2\ell-2} (2k_t)!
\end{equation*}
By Strling's approximation, this is upper bounded by $\prod_{t=0}^{2\ell-2} (2k_t)!\prod_{t=0}^{2\ell-1} {n\choose r_t}\exp (O(r))$. 
Next we choose the remaining hypergraph $\alpha \setminus \beta$ and  $\beta\setminus \alpha$ respectively. For $\alpha \setminus \beta$, we have 
\begin{align*}
    N_{\alpha\setminus\beta} &={n\choose \frac{v+1}{2}-r_0}{n\choose  v-r_{2\ell-1}}\prod_{t=1}^{\ell-1} {n\choose 2v-2r_{2t-1}}{n\choose v-r_{2t}} \prod_{t=0}^{2\ell-2} (2(v-k_t))! \\ 
    & \leq \lvert S_{\ell,v,i}\rvert n^{-r}v^rv^{-2k}\exp(O(r))
\end{align*}
For bounding the number of choices for $\beta$, suppose $\alpha,\beta$ share $s$ vertices with degree $0$ or $1$ in the subgraph $\alpha\cap\beta$. Then there are $\ell^{s}$ ways of embedding $\alpha\cap\beta$ and shared vertices in hypergraph $\beta$. Further the number of choices for the remaining hypergraph is also bounded by $\lvert S_{\ell,v}\rvert n^{-r}v^rv^{-2k}\exp(O(r))$.  Therefore the contribution to $\frac{\E [\lprod P(Y),x\rprod^2]}{\left(\E [\lprod P(Y),x\rprod]\right)^2}$ with respect to specific $r_t,k_t$ is bounded by
\begin{equation*}
   (n^{-r}v^rv^{-2k})^2n^rv^{k/2}\ell^s\Gamma^{O(2r-3k)}\lambda^{-2k} \exp(O(r))\leq \left(\frac{n}{cv^2\ell^2\Gamma^{O(1)}}\right)^{-r+3k/2} \left(cn^{-3/2}v^{-1/2}\lambda^{-2}\right)^k  
\end{equation*}
where $c$ is constant. As in the proof of \ref{thm:detection-tensor}, when $nv^{-2}=\omega(\text{poly}(\Gamma\ell))$ the dominating term is given by $r=\frac{3}{2}k$. In this case, we must have $k\geq \ell$.  For $cn^{-3/2}v^{-1/2}\lambda^{-2}<1$ and $\ell=C\log n$ with constant $C$ large enough, the contribution is $n^{-\Omega(1)}$.

Therefore in all, we have $\frac{\left(\E \lprod P(Y),x\rprod\right)^2}{\E \lprod P(Y),x\rprod^2}=1-o(1)$ when we have relation $cn^{-3/2}v^{-1/2}\lambda^{-2}=1-\Omega(1)$ and $n=\omega(v^2 \text{poly}(\Gamma\ell))$. 
\end{proof}

\begin{proof}[Proof of Theorem \ref{StrongRecoveryTensor}]
 As a result of lemma \ref{concentration_lemma}, by Chebyshev's inequality we have $\lprod P(Y),x\rprod=(1\pm o(1))\E \lprod P(Y),x\rprod$ w.h.p. Combined with correlation in expectation and Markov inequality, we have $\frac{\lprod P(Y),x\rprod}{\lVert x\rVert\lVert P(Y)\rVert}\geq \epsilon\cdot\Omega(1)$ with probability $1-\epsilon$. We set $\epsilon=o(1)$ and take estimator $z\in\mathbb{R}^n$ with $z_k=\sum_{ijk}Y_{ijk}P_i(Y)P_{j}(Y)$. Then according to  the theorem \ref{allNothing}, we get the strong recovery guarantee.
\end{proof}

\section{Higher order general spiked tensor model}\label{HigherOrder}
For clarity, we discuss the algorithms  for order-3 spiked tensor model and spiked Wigner model above. Such claims can be generalized to higher order tensor without difficulties. 

\begin{theorem}
\label{thm:detection-tensor-Higher-Order}
  Let $x \in \R^n$ be a random vector with independent, mean-zero entries having $\E x_i^2 = 1$ and $\Gamma=\E x_i^4 \leq n^{o(1)}$. Let $\lambda > 0$ and $Y = \lambda \cdot x^{\otimes p} + W$, where $W \in (\R^{n})^{\otimes p}$ has independent, mean-zero entries with $\E W_{i_1,i_2,\ldots,i_p}^2 = 1$.
  Then for $v= o(\frac{n^{1/2}}{\text{poly}(\Gamma\log n)})$ and $\lambda \geq c_p n^{-p/4}v^{-(p-2)/4}$, there is $n^{O(pv)}$ time algorithm achieving strong detection.
\end{theorem}

\begin{theorem}\label{thm:recovery-tensor-Higher-Order}
  Let $x \in \R^n$ be a random vector with independent, mean-zero entries having $\E x_i^2 = 1$ and $\Gamma=\E x_i^4 \leq n^{o(1)}$. Let $\lambda > 0$ and $Y = \lambda \cdot x^{\otimes p} + W$, where $W \in (\R^{n})^{\otimes p}$ has independent, mean-zero entries with $\E W_{i_1,i_2,\ldots,i_p}^2 = 1$.
  Then for $v= o(\frac{n^{1/2}}{\text{poly}(\Gamma\log n)})$ and $\lambda \geq c_p n^{-p/4}v^{-(p-2)/4}$, there is $n^{O(pv)}$ time algorithm giving unit norm estimator  $\hat{x}$ s.t $\lprod \hat{x},\frac{x}{\lVert x\rVert}\rprod^2 \geq \Omega(1)$.
\end{theorem}
 Specifically this leads to polynomial time algorithm when $\lambda=\Omega(n^{-p/4})$. 

When the order-$p$ is odd, the analysis is very similar to the one for the case $p=3$. Therefore we mainly talk about case where order $p$ is even and prove the guarantee of the theorem. 

\subsection{Strong detection algorithm for even p}
For strong detection we propose the following thresholding polynomial
\begin{definition}[thresholding polynomial for even p]\label{pDetectionPolynomial}
   On directed complete $p$-uniform hypergraph with $n$ vertices, we define $S_{\ell,v}$ as the set of all copies of $2$-regular hypergraphs generated in the following way
   \begin{itemize}
       \item We construct $\ell$ levels of vertices, with each level containing $pv/2$ vertices. 
       \item For $t\in [\ell]$,  we connect a perfect matching with $v$ hyperedges between level $t$ and level $t+1$. Each hyperedge directs from $p/2$ vertices in level $t$ to $p/2$ vertices in level $t+1$.
       \item Finally we similarly connect a perfect matching with $v$ hyperedges between level $\ell-1$ and level $0$. Each hyperedge directs from $p/2$ vertices in level $0$ to $p/2$ vertices in level $\ell-1$
   \end{itemize}

    The thresholding polynomial is given by $P(Y)=\sum_{\alpha\in S_{\ell,v}} \chi_\alpha(Y)$ where $\chi_\alpha(Y)$ is the Fourier basis associated with hypergraph $\alpha$: $\chi_\alpha(Y)=\sum_{(i_1,i_2,\ldots,i_p)\in \alpha} Y_{i_1,i_2,\ldots,i_p}$.  
\end{definition}
\begin{lemma}
    Suppose we have $\gamma=c_p v^{(p-2)/2}\lambda^2 n^{p/2}=1+\Omega(1)$ with $c_p$ small enough constant related to $p$, then taking $\ell=O(\log_{\gamma} n)$,
    the projection of likelihood ratio $\hat{\mu}(Y)$ with respect to $S_{\ell,v}$ is $\omega(1)$ when $n/v^2=\omega(\text{poly}(\ell))$.
\end{lemma}
\begin{proof}
    Given fixed vertices in level $t$ and level $t+1$, we have $((pv/2)!)^2/v!$ choices for the hyperedges between level $t$ and level $t+1$. Therefore we have $$S_{\ell,v}=(1-o(1))\left(\frac{(pv/2)!)}{v!}\right)^\ell n^{\ell pv/2}$$. Therefore by Strling's approximation this implies $S_{\ell,v}\geq v^{(p-2)v/2} \lambda^{2v} n^{\ell pv/2}\exp(-O(v))$. Therefore we have
    \begin{equation*}
        \sum_{\alpha\in S_\ell}\mu_{\alpha}^2\geq \lambda^{2\ell v} n^{\ell pv/2} v^{(p-2)v/2}\exp(-O(v))=\omega(1) 
    \end{equation*}
    when we have $\ell,v$ as described. 
\end{proof}

\begin{lemma}\label{Higher-Order-Concentration}
    Suppose we have $\gamma=c_p v^{(p-2)/2}\lambda^2 n^{p/2}=1+\Omega(1)$, and $\ell=O(\log_{\gamma} n)$ with $c_p$ small enough constant related to $p$. If $n/v^2=\omega(\text{poly}(\Gamma\ell))$ then we have the following concentration property:
    \begin{equation*}
        \E P^2(Y)=(1+o(1))(\E P(Y))^2
    \end{equation*}
    
\end{lemma}
\begin{proof}
    We consider $\alpha,\beta\in S_{\ell,v}$. We first choose $\alpha\cap \beta$ and shared vertices as subgraph of $\alpha$ and then select the remaining hypergraph of $\alpha,\beta$. As before, we have $\E[\chi_\alpha(Y)\chi_\beta(Y)]\leq \lambda^{4\ell v-2k} \Gamma^{2r-pk}$, where $r$ is the number of shared vertices,$k$ is the number of shared hyperedges, and $\Gamma=\E [x_i^4]=n^{o(1)}$. Considering shared vertices and hyperedges in $\alpha$, if there are $r_t$ shared vertices in level $t$ and $k_t$ shared hyperedges between level $t$ and level $t+1$, then there are
\begin{align*}
     N_{\alpha\cap\beta}&\leq\prod_{t=0}^{\ell}{n\choose r_{t}} {r_{t}\choose \frac{pk_{t}}{2}} {r_{t}\choose \frac{pk_{t-1}}{2}}  \frac{(((pk_{t}/2)!)^2
     }{k_t!}\\
     &\leq n^{r} v^{(p-2)k/2} \exp (O(r))
 \end{align*}
    such subgraphs. On the other hand, the number of choices for the remaining hypergraph of $\alpha$ is bounded  by
    \begin{align*}
        N_{\alpha-\beta} &= \prod_{t=0}^{\ell-1} {n\choose \frac{pv}{2}-r_{t}} \frac{((p(v-k_t)/2)!)^2}{(v-k_t)!}
        \\
        &= \lvert S_{\ell,v}\rvert
        n^{-r} v^{r-pk+k}\exp(O(r))
    \end{align*}
    Then we consider choices for $\beta$. Denote the number of degree-$1$ vertices in $\alpha\cap \beta$ as $s_1$ and the number of shared vertices not contained in $\alpha\cap\beta$ as $s_0$. Let $s=s_0+s_1$, then there are at most $\ell^s$ ways of putting $\alpha\cap\beta$ in $\beta$. For the same reasoning, number of ways for choosing the remaining hypergraph of $\beta$ is also bounded by $\lvert S_{\ell,v}\rvert
        n^{-r} v^{r-pk+k}\exp(O(r))$. Therefore, with respect to fixed number of vertices and hyperedges $r_t,k_t$ in each level of $\alpha$, the total number of such hypergraph pairs $S_{r,k,\ell,v}$ is bounded by
        \begin{equation*}
            \lvert S_{\ell,v}\rvert^2 n^{-r} v^{2r-2pk+2k}\ell^s v^{(p-2)k/2}\exp(O(r))
        \end{equation*}
        Therefore the corresponding contribution is given by
        \begin{equation*}
            \sum_{\alpha,\beta\in S_{r,k,\ell,v}} \frac{\E \chi_\alpha(Y)\chi_\beta(Y)}{(\E P(Y))^2}\leq n^{-r} v^{2r-2pk+2k}\ell^s v^{(p-2)k/2}\lambda^{-2k}\Gamma^{2r-pk}\exp(O(r))
        \end{equation*}
        Since we have $2r-s\geq pk$ by degree constraints, this is bounded by
        \begin{equation*}
            c^{pk/2}n^{-pk/2} v^{-(p-2)k/2}\lambda^{-2k} \ell \left(\frac{n}{cv^2\Gamma^2}\right)^{-r+pk/2}
        \end{equation*}
        where $c$ is constant.  Summing up for $\alpha,\beta$ with respect to different $r_t,k_t$ and combining the fact that if $r=pk/2$ then $k\geq 2\ell$, we have
        \begin{align*}
            &\sum_{r_t,k_t} c^{pk/2}n^{-pk/2} v^{-(p-2)k/2}\lambda^{-2k} \ell \left(\frac{n}{cv^2\Gamma^2}\right)^{-r+pk/2}\\
            =& \sum_{k_t\geq 0} \sum_{\substack{r_t\geq \delta(t) \text{max}(k_t,k_{t-1})\\ 2r\geq pk+1}} c^{pk/2}n^{-pk/2} v^{-(p-2)k/2}\lambda^{-2k} \ell \left(\frac{n}{cv^2\Gamma^2}\right)^{-r+pk/2} + \\
            & \sum_{k_t\geq 1} \left(c^{p/2} n^{-p/2}\lambda^{-2}v^{-(p-2)/2}\right)^{\sum_{k_t}} 
        \end{align*}
        When $\gamma=c_p n^{-p/2}\lambda^{-2}v^{-1/2}>1$ and $\ell=C\log_\gamma n$ with constant $C$ large enough, the second term is bounded by $n^{-\Omega(1)}$. For the first term we note that given $k_t$ for $t\in [\ell]$, we have $\ell^{r-pk/2}$ choices for $r_t$. We denote $k_\Delta=\lvert k_{t+1}-k_t\rvert$. Then we have $k_\Delta=O(r-3k/2)$.  Then given $k_\Delta$ we have at most $k_\Delta$ different values for $k_t$. As a result fixing these $k_\Delta$ different values, there are $\ell^{k_\Delta}$ choices for $k_t$ for $t\in [\ell]$. Therefore the first term is bounded by
        \begin{equation*}
       \sum_{k_\Delta\geq 0} \left[\left(\frac{n}{cv^2\text{poly}(\Gamma\ell)}\right)^{-\text{max}(1/2,k_\Delta)}\prod_{t=1}^{k_\Delta} \left(\sum_{k_t\geq 0}(c^{p/2} n^{-p/2} \lambda^{-2} v^{-(p-2)/2} )^{k_t}\right)\right]=o(1)
        \end{equation*}
         In all, we have
 $\sum_{\alpha,\beta\in S_{\ell,v}} \E[\chi_\alpha(Y)\chi_\beta(Y)]\leq (1+o(1)) (\E P(Y))^2$
\end{proof}

Next we show that the running time can be improved using color-coding method. We describe the evaluation algorithm \ref{algoCircleTensorHigherOrder}.
\begin{algorithm}
\KwData{Given $Y\in (\mathbb{R}^{n})^{\otimes p}$ s.t $Y=\lambda x^{\otimes p}+W$}
\KwResult{$P(Y)\in\mathbb{R}$ which is the sum of multilinear monomials corresponding to hypergraphs  in $S_{\ell,v}$(up to accuracy $1+n^{-\Omega(1)}$)}
 \For{$i\gets1$ \KwTo $C$}{
     Sample coloring $c_i:[n]\mapsto [\ell]$ uniformly at random\;
   Construct a matrix $M,N\in \mathbb{R}^{(2^{p\ell v}-1)n^{pv/2}\times (2^{p\ell v}-1)n^{pv/2}}$\;
    
    Record $p_{c_i}= \frac{(p\ell v)^{p\ell_v}}{(p\ell v)!}\text{Tr}(M^{\ell-1}N)$\;}
    Return  $\frac{1}{C}\sum_{i=1}^C p_{c_i}$
 \caption{Algorithm for evaluating the thresholding polynomial }\label{algoCircleTensorHigherOrder}
 \end{algorithm}
We describe the matrices $M,N$ used in the algorithm. The rows and columns of $M$ are indexed by $(V_1,S)$ and $(V_2,T)$ where $V_1\in [n]^{pv/2}$ and $V_2\in [n]^{pv/2}$ correspond to  sets of labels of vertices while $S,T\subsetneq [\ell v]$ correspond to subsets of colors. We have $M_{(V_1,S),(V_2,T)}= 0$ if $S\cup \{c(v):v\in V_1\}\neq T$ or $\{c(v):v\in V_1\}$ and $S$ are not disjoint. Otherwise $M_{(V_1,S),(V_2,T)}$ is given by $\sum_{\gamma \in S_{V_1,V_2}}\chi_{\gamma}(Y)$ where $S_{V_1,V_2}$ is the set of perfect matching induced by $V_1$ and $V_2$(each hyperedge in $S_{V_1,V_2}$ direct from $p/2$ vertices from $V_1$ to $p/2$ 
    vertices from $V_2$).
    
    For matrix $N$, the indexing is the same as $M$. The entry $N_{(V_1,S),(V_2,T)}= 0$ if $S\cup \{c(v):v\in V_1\}\neq T$ or $\{c(v):v\in V_1\}$ and $S$ are not disjoint. Otherwise $N_{(V_1,S),(V_2,T)}$ is given by $\sum_{\gamma \in S_{V_2,V_1}}\chi_{\gamma}(Y)$ where $S_{V_2,V_1}$ is the set of perfect matching induced by $V_2$ and $V_1$(each hyperedge in hypergraph  $\alpha\in S_{V_2,V_1}$ directs from $p/2$ 
    vertices from $V_2$ to $p/2$ vertices from $V_1$).

\begin{lemma}[Evaluation of thresholding polynomial]\label{pevaDetection}
       There exists a $n^{O(v)}$-time algorithm that given a coloring $c$:$[n]\to [p\ell v]$(where $p\ell v$ is the number of vertices in hypergraph $\alpha\in S_{\ell,v}$) and a tensor $Y\in (\mathbb{R}^{n})^{\otimes p}$ evaluates degree $2\ell v$ polynomial below in polynomial time
       \begin{equation}
       p_c(Y)=\mathop{\sum}_{\alpha\in S_{\ell,v}} \chi_\alpha(Y)F_{c,\alpha}
       \end{equation}
       \begin{equation}
    F_{c,\alpha}=\frac{(p\ell v)^{p\ell v}}{(p\ell v)!} \cdot \mathbf{1}_{c(\alpha)=[p\ell v]}
\end{equation}
    when thresholding polynomial $P(Y)$ defined in \ref{pDetectionPolynomial} satisfies $(\E P(Y))^2=(1-o(1))\E P^2(Y)$, we can take $\exp(O(\ell v))$ random colorings and give an accurate estimation of the thresholding polynomial by averaging $p_c(Y)$. 
    \end{lemma}
    \begin{proof}

     The observation is that $p_c(Y)$ is just given by $(p\ell v)^{p\ell_v}/(p\ell v)!$ times the trace of $(M)^{\ell-1}N$. This can be done in time $n^{3pv/2}\exp(O(\ell v))$. 
    
    Next we prove that averaging random coloring for $p_c(Y)$ will give accurate estimation for $P(Y)$ in the detection algorithm. First we note that $\E_c p_c(Y)=P(Y)$. Next, for single coloring we have
    \begin{equation*}
        \E p^2_c(Y)= \sum_{\alpha,\beta\in  S_{\ell,v}}\E [F_{c,\alpha} F_{c,\beta} \chi_\alpha(Y)\chi_\beta(Y)] \leq \exp(O(\ell v)) \E P^2(Y)\leq \exp(O(\ell v)) (\E p_c(Y))^2
    \end{equation*}
    where we use the result that $\E P^2(Y)=(1+o(1)) (\E P(Y))^2$. Therefore, by averaging $L=\exp(O(\ell v))$ random colorings, the variance can be reduced such that $\frac{\sum_{t=1}^L p_{c_t}(Y)}{L}=(1\pm o(1))P(Y)$ w.h.p. 
     \end{proof}
     
\textbf{Remark}: When $\lambda=\Omega(n^{-p/4})$, we can take $c_p \lambda n^{p/4}v>1$ with $c_p$ being small enough constant dependent on order $p$. This leads to $O(n^{2pv})$ time algorithm with constant $v$. When $\lambda=\omega(n^{-p/4})$, we can simply take $P(Y)=\sum_{i_1<i_2<\ldots<i_p}Y_{i_1,i_2,\ldots,i_p} Y_{i_p,i_{p-1},\ldots,i_1}$. Obviously such polynomial can be evaluated in linear time.

\begin{proof}[Proof of Theorem \ref{thm:detection-tensor-Higher-Order}]
    Combining the concentration property proved in lemma \ref{Higher-Order-Concentration} and  the running time proved in lemma \ref{pevaDetection}, we get the claim.
\end{proof}
\subsection{Weak recovery algorithm for even p}
For weak recovery we want to propose estimator $P(Y)\in \mathbb{R}^{n\times n}$ such that
\begin{equation*}
    \frac{\left(\E\left\lprod P(Y),xx^\top\right\rprod\right)^2
    }{(\E P(Y)\rVert_F^2\lVert \E \lVert xx^\top\rVert_F^2)}
\end{equation*}

For even $p$, it can always be decomposed into two odd numbers $p_1$ and $p_2$ s.t $p=p_1+p_2$. Let $p_1,p_2$ be such a pair of odd numbers that minimizes  $|p_2-p_1|$. Then we define the estimation vector for weak recovery as following:  
\begin{definition}[Estimator for even-$p$ weak recovery]\label{pEstimatorPolynomial}
    On $p$-uniform directed complete hypergraph on $n$ vertices, we define the following set of subgraphs

    Given tensor $Y\in (\mathbb{R}^{n})^{\otimes p}$, we have estimator $P(Y)\in \mathbb{R}^{n\times n}$ where each entry $P_{ij}(Y)$ is a degree $(2\ell-1) v$ polynomial given by $\sum_{\alpha\in S_{\ell,v,i,j}} \chi_\alpha(Y)$ where $S_{\ell,v,i,j}$ is the set of hypergraph generated in the following way:
    \begin{itemize}
        \item we construct $2\ell$ levels of vertices. Level $0$ contains vertex $i$ and $(p_1v-1)/2$ vertices in addition. For $0<t<\ell$, level $2t$ contains $p_1 v$ vertices while for $0<t<\ell-1$ level $2t+1$ contains $p_2 v$ vertices. Level $2\ell-1$ contains vertex $j$ and $(p_2v-1)/2$ vertices in addition. All vertices are distinct. 
    \item We construct a perfect matching between level $t,t+1$ for $t\in[1,2\ell-2]$,each hyperedge directs from $p_1$ vertices in even level to $p_2$ vertices in odd level.
    \item Level $0$ and $1$ are connected as bipartite hypergraph s.t each vertex in level $0$ excluding $i$ has degree $2$ while vertex $i$ and vertices in level $1$ has degree $1$. Level $2\ell-2$ and level $2\ell-1$ are connected as bipartite hypergraph s.t vertices in level $2\ell-1$ excluding $j$ has degree $2$ while vertices in level $2\ell-2$ and vertex $j$ has degree $1$
    \end{itemize}
\end{definition}
\begin{lemma}
    Taking $\gamma=c_p n^{p/2} v^{(p-2)/2} \lambda^2=1+\Omega(1)$(where $c_p$ is small enough constant related to $p$) and $\ell=O(\log_\gamma n)$ in the estimator above, then if $n=\omega(v^2\text{poly}(\ell\Gamma))$, we have 
    \begin{equation*}
        \frac{\E\lprod P(Y),xx^\top\rprod}{\left(\E \lVert P(Y)\rVert_F^2\E \lVert x\rVert^4\right)^{1/2}}=\Omega(1)
    \end{equation*}
\end{lemma}
\begin{proof}
    We need to show the estimator $P(Y)\in\mathbb{R}^n$ above achieves constant correlation with the hidden vector $x$. Equivalently we want to show that for each $i,j\in [n]$, we have
    \begin{equation*}
        \left(\sum_{\alpha\in S_{\ell,v,i,j}} \E [\chi_\alpha(Y)x_ix_j]\right)^2=\Omega\left(\sum_{\alpha,\beta\in S_{\ell,v,i,j}}\E[\chi_\alpha(Y)\chi_\beta(Y)]\right)
    \end{equation*}
    Since we have $\left(\sum_{\alpha,\beta\in S_{\ell,v,i,j}}\E[\chi_\alpha(Y)x_ix_j]\right)=\lambda^{(2\ell-1)v}\lvert S_{\ell,v,i,j}\rvert$, we only need to bound the size of $S_{\ell,v,i,j}$. Applying combinatorial arguments to the generating process of $S_{\ell,v,i,j}$, we have
    \begin{equation*}
        \lvert S_{\ell,v,i,j}\rvert=(1-o(1)) \left({n\choose p_1v}{n\choose p_2v}\right)^{\ell-1}{n\choose (p_1v-1)/2}{n\choose (p_2v-1)/2} \frac{((p_1v)!(p_2v)!)^{2\ell-1}}{v^{2\ell-1}2^{(pv-2)/2}}
    \end{equation*}
    On the other hand,  first choose $\alpha\cap\beta$ and shared vertices(excluding $i$ and $j$) as subgraph of hypergraph $\alpha\in S_{\ell,v,i,j}$. For the shared vertices and hyperedges consisting in $\alpha$, if there are $r_t$ vertices in level $t\in [2\ell-1]$ and $k_t$ hyperedges between level $r_t$ and level $r_{t+1}$, then the number of such intersection is bounded by $N_{\alpha\cap\beta}$:
    \begin{align*}
        & {n\choose r_0}{n\choose r_{2\ell-1}} {2r_0+1\choose p_1k_0} {2r_{2\ell-1}+1\choose p_2 k_{2\ell-2}} \prod_{t=1}^{\ell-1} {n\choose r_{2t-1}}{n\choose r_{2t}}{r_{2t-1}\choose p_1k_{2t-1}}{r_{2t-1}\choose p_1k_{2t-1}}{r_{2t}\choose p_2k_{2t}} \\
        & {r_{2t}\choose p_2k_{2t-1}}\prod_{t=0}^{2\ell-2} \frac{(p_1k_t)!(p_2k_t)!}{k_t!} 
    \end{align*}
    This is upper bounded by $$\prod_{t=0}^{2\ell-2} \frac{(p_1k_t)!(p_2k_t)!}{k_t!}\prod_{t=0}^{2\ell-1} {n\choose r_t}\exp(O(r))$$. Next we choose the remaining hypergraph $\alpha\setminus \beta$ and $\beta\setminus \alpha$ respectively. For $\alpha\setminus\beta$, we have
    \begin{align*}
    N_{\alpha\setminus\beta}&={n\choose \frac{p_1v-1}{2}-r_0} {n\choose \frac{p_1v-1}{2}-r_{2\ell-1}}\prod_{t=1}^{\ell-1} {n\choose p_2v-r_{2t-1}}{n\choose p_1v-r_{2t}}\\ &\prod_{t=0}^{2\ell-2} \frac{(p_1(v-k_t))!(p_2(v-k_t))!}{k_t!}\\
    &\leq \lvert S_{\ell,v,i,j}\rvert n^{-r} v^r v^{-(p-1)k} \exp(O(r))
    \end{align*}
    Suppose there are $s_1$ degree $1$ vertices in $\alpha\cap\beta$ and $s_0$ vertices shared between $\alpha,\beta $ but not contained in $\alpha\cap \beta$, denoting $s=s_0+s_1$, then there are $\ell^s$ ways of placing $\alpha\cap \beta$ and shared vertices in hypergraph $\beta$ and the count of remaining hypergraph is also bounded by $\lvert S_{\ell,v,i,j}\rvert n^{-r} v^{r}v^{-(p-1)k}\exp(O(r))$. Moreover we have
    \begin{align*}
       \E[\chi_\alpha(Y)\chi_\beta(Y)]&=(1+n^{-\Omega(1)}) \lambda^{2(2\ell-1)v-2k}\E\left[\prod_{j\in\alpha\Delta\beta} x_j^{\text{deg}(j,\alpha\Delta\beta)}\right]\\
       &\leq \lambda^{-2k} \Gamma^{O(2r-pk)}\E[\chi_\alpha(Y)]\E[\chi_\beta(Y)]
    \end{align*}
    where $\text{deg}(j,\alpha,\Delta\beta)$ represents the degree of vertex $j$ in hypergraph $\alpha\Delta\beta$. Therefore the contribution to $\frac{\E[P_{ij}^2(Y)]}{(\E[P_{ij}(Y)x_ix_j])^2}$ with respect to specific $r_t,k_t$ is bounded by
    \begin{equation*}
        (n^{-r} v^r v^{(1-p)k})^2 n^rv^{(p-2)k/2}\ell^s \Gamma^{O(2r-pk)}\lambda^{-2k}\exp(O(r))
    \end{equation*}
    Because we have $2r-s-2\geq pk$ by degree constraints, we study the terms in cases of $2r-s>pk$ and  $2r-s<=pk$. For  the case $2r-s>pk$, the sum of contribution is $o(1)$ by the same reasoning as in the proof of detection. For the case $2r-s=pk$,  $\alpha\cap\beta$ consists of $2$ hyperpaths   respectively starting from $i$ and $j$. Further each shared vertex is contained in $\alpha\cap\beta$.  For such case the contribution is bounded by
    \begin{align*}
     \sum_{\ell_1,\ell_2}
     c^{p(\ell_1+\ell_2)}
     n^{-p(\ell_1+\ell_2)/2}v^{-(p-2)(\ell_1+\ell_2)/2}\lambda^{-2(\ell_1+\ell_2)}
     & \leq \left(\sum_{\ell_1}c^{\ell_1}n^{-p\ell_1/2}v^{-(p-2)\ell_1/2}\lambda^{-2\ell_1}\right)^2\\
     & \leq \left(\frac{1}{1-v^{-(p-2)/2}n^{-p/2}\lambda^{-2}}\right)^2
    \end{align*}
    For the case $2r-s<pk$, $\alpha\cap\beta$ contains more than $2\ell-1$ hyperedges. For $\ell=O(\log_\gamma n)$ with hidden constant large enough, the contribution is also $o(1)$. Therefore in all we have
    \begin{equation*}
        \frac{\left(\E \lprod P(Y), xx^\top\rprod\right)^2}{\E \lVert P(Y)\rVert_F^2 \E\lVert xx^\top\rVert_F^2}=\frac{1}{1-v^{-(p-2)/2}n^{-p/2}\lambda^{-2}}+o(1)
    \end{equation*}
    Therefore when $n=\omega(cv^2 \Gamma^2\text{polylog}(n))$ and $\gamma=1+\Omega(1)$, we have weak recovery algorithm by taking random eigenvector in the top $(1-\gamma^{-1})^{-O(1)}$ span of matrix $P(Y)$(as shown in  \cite{8104074}). When $\gamma=\omega(1)$, taking leading eigenvector of $P(Y)$ gives strong recovery guarantee.
\end{proof}
    
    Next we evaluate the  polynomial estimator for weak recovery using color-coding method, as shown in algorithm \ref{algoEstimationTensorHigherOrder}. We denote $\ell_v =\frac{p(2\ell-1) v+1}{2}$
    
\begin{algorithm}
\KwData{Given $Y\in (\mathbb{R}^{n})^{\otimes p}$  s.t $Y=\lambda x^{\otimes p}+W$}
\KwResult{$P(Y)\in\mathbb{R}^{n\times n}$, with $P_{ij}(Y)=\sum_{\alpha\in S_{\ell,v,i,j}}\chi_\alpha(Y)$(up to accuracy $1+n^{-\Omega(1)}$)}
 $C\gets \exp(100\ell v)$\;
 \For{$i\gets1$ \KwTo $C$}{
     Sample coloring $c_i:[n]\mapsto [\ell_v]$ uniformly at random\;
   Construct matrices $M,N\in \mathbb{R}^{(2^{\ell_v}-1)n^{p_1v}\times (2^{\ell_v}-1)n^{p_2v}}$\;
    Construct  matrices: $A\in\mathbb{R}^{n^{(p_1v+1)/2}\times (2^{\ell_v}-1)n^{p_2v}}$ ,$B\in\mathbb{R}^{(2^{\ell_v}-1)n^{p_1v}\times n^{(p_2v+1)/2}}$\;
    Construct matrix $L^{(1)}\in \mathbb{R}^{n\times n^{(p_1v+1)/2}}, L^{(2)}\in \mathbb{R}^{ n^{(p_2v+1)/2}\times n}$ \;
    Record matrix $p_{c_i}=L^{(1)} A(NM)^{\ell-2}L^{(2)}$\;}
    Return  $\frac{1}{C}\sum_{i=1}^C p_{c_i}$
 \caption{Algorithm for evaluating estimation matrix }\label{algoEstimationTensorHigherOrder}
 \end{algorithm}
 
 Next we describe how to construct matrices $M,N,A,B$. The rows and columns of $M$ are indexed by $(V_1,S)$ and $(V_2,T)$ where $V_1\in [n]^{p_1v}$ and $V_2\in [n]^{p_2v}$ are set of vertices while $S,T\subsetneq [\ell v]$ are subset of colors. We have $M_{(V_1,S),(V_2,T)}= 0$ if $S\cup \{c(v):v\in V_1\}\neq T$ or $\{c(v):v\in V_1\}$ and $S$ are not disjoint. Otherwise $M_{(V_1,S),(V_2,T)}$ is given by $\sum_{\gamma \in S_{V_1,V_2}}\chi_{\gamma}(Y)$ where $S_{V_1,V_2}$ is the set of perfect matching induced by $V_1$ and $V_2$(each hyperedge in $S_{V_1,V_2}$ direct from $p_1$ 
    vertices from $V_1$ to $p_2$ vertices from $V_2$).
    
    For matrix $N$, the indexing are the same as $M$. We have $N_{(V_1,S),(V_2,T)}= 0$ if $S\cup \{c(v):v\in V_1\}\neq T$ or $\{c(v):v\in V_1\}$ and $S$ are not disjoint. Otherwise $N_{(V_1,S),(V_2,T)}$ is given by $\sum_{\gamma \in S_{V_2,V_1}}\chi_{\gamma}(Y)$ where $S_{V_2,V_1}$ is the set of perfect matching induced by $V_2$ and $V_1$(each hyperedge in hypergraph  $\alpha\in  S_{V_2,V_1}$ directs from $p_2$ 
    vertices from $V_2$ to $p_1$ vertices from $V_1$).

We consider a subset of hypergraphs contained in $S_{\ell,v,i,j}$ with $i,j\in [n]$, denoted by $\mathcal{H}_{i,j,V_1,V_2}$. The set of vertices in  level $0$ of these hypergraphs is fixed to be $\{i\}\cup V_1$, where $V_1\subseteq [n]$. The set of vertices in the level $1$ of these hypergraphs is fixed to be $V_2\subseteq [n]$.   We denote $S_{i,V_1,V_2}$ as the following set of spanning subgraphs: a hypergraph $\alpha\in S_{i,V_1,V_2}$ if and only if there exists a hypergraph $\beta\in \mathcal{H}_{i,j,V_1,V_2}$ such that the hyperedge set of $\alpha$ is the same as the set of hyperedges between level $0$ and $1$ of $\beta$. 

In the same way, we consider a subset of hypergraphs contained in $S_{\ell,v,i,j}$ with $i,j\in [n]$, denoted by $\mathcal{L}_{i,j,V_1,V_2}$, with vertices in levels $2\ell-2$ and $2\ell-1$ fixed. We denote $S_{V_1,V_2,j}$ as the following set of spanning subgraphs: a hypergraph $\alpha\in S_{V_1,V_2,j}$ if and only if there exists a hypergraph $\beta\in \mathcal{L}_{i,j,V_1,V_2}$ such that the hyperedge set of $\alpha$ is the same as the set of hyperedges between level $2\ell-2$ and $2\ell-1$ of $\beta$. 

   By these definitions, the entry $A_{(i,V_1),(V_2,T)}$ is given by $\sum_{\alpha\in S_{i,V_1,V_2}}\chi_\alpha(Y)$ if $T=\{c(v):v\in V_1\cup v\},v\not\in V_1$ and $0$ otherwise. The entry $B_{(V_1,S),V_2}$ is given by   $\sum_{\alpha\in S_{V_1,V_2}}\chi_\alpha(Y)$ if $S\cup \{c(v):v\in V_1\cup V_2\}=[\ell_v], S\cap \{c(v):v\in V_1\cup V_2\}=\emptyset$ and zero otherwise.
  
   Finally, we construct the deterministic matrices $L^{(1)},L^{(2)}$. The columns of matrix $L^{(1)}$ are indexed by $(j,V_1)$, where $j\in [n]$ and $V_1$ is a size-$n^{(p_1v+1)/2}$ subset of $[n]$. The rows are indexed by $i\in [n]$. The entry $L^{(1)}_{i,(j,V_1)}=1$ if $j=i$ and $0$ otherwise. The transpose of $L^{(2)}$ is indexed in the same way.  The entry $L^{(2)}_{(j,V_1),i}=1$ if $i=j$ and $0$ otherwise.
   
    \begin{lemma}[Evaluation of polynomial estimator]\label{pevaRec}
    Denote the number of vertices in any hypergraph contained in $S_{\ell,v,i}$ as $\ell_v$, then $\ell_v =\frac{p(2\ell-1) v+1}{2}$.
    Given sampled  colorings $c_i$:$[n]\to [\ell_v]$ and a tensor $Y\in (\mathbb{R}^{n})^{\otimes p}$,
the algorithm \ref{algoEstimationTensorHigherOrder} return a matrix $p(Y,c_1,\ldots,c_C)\in\mathbb{R}^{n\times n}$ in time   $n^{O(pv)}\exp(O(\ell_v))$. When $\frac{\E \lprod P(Y), xx^\top\rprod}{n \left(\E \lVert P(Y)\rVert_F^2 \right)^{1/2}}=\delta=\Omega(1)$, we have \[\frac{\E \lprod p(Y,c_1,\ldots,c_C), xx^\top\rprod}{n \left(\E \lVert p(Y,c_1,\ldots,c_C)\rVert_F^2 \right)^{1/2}}\geq (1-o(1))\delta\geq \Omega(1).\]
       
    \end{lemma}
    \textbf{Remark}:  When $\lambda=\omega(n^{-p/4})$, using power method for extracting leading eigenvector, we have $n^{p+o(1)}$ time algorithm for evaluating the leading  eigenvector of the matrix returned by the algorithm.
    \begin{proof}
    The critical observation is that  $\mathop{\sum}_{\alpha\in S_{\ell,v,i,j}} \chi_\alpha(Y)F_{c,\alpha}$ can be obtained from the matrix $A(NM)^{\ell-2}NB$ in the algorithm by summing up all entries in $H$ indexed by row $(i,\cdot)$ and column $(j,\cdot)$. 
    Thus given random coloring $c$ the algorithm evaluates matrix $p_c(Y)$ satisfying the following:
     \begin{equation*}   p_{c,i,j}(Y)=\mathop{\sum}_{\alpha\in S_{\ell,v,i,j}} \chi_\alpha(Y)F_{c,\alpha}
       \end{equation*}
       \begin{equation*}
    F_{c,\alpha}=\frac{\ell_v^{\ell_v}}{\ell_v!} \cdot \mathbf{1}_{c(\alpha)=[\ell_v]}
    \end{equation*}
     Thus \[p_c(Y)=\mathop{\sum}_{\alpha\in S_{\ell,v,i,j}} \chi_\alpha(Y)F_{c,\alpha}=L^{(1)}A(NM)^{\ell-2}NBL^{(2)}\]
    
    
    By the same argument in the strong detection algorithm, we can obtain accurate estimation of $P(Y)$ by averaging $\exp(O(\ell v))$ random colorings when $\gamma=c_p\lambda^2 n^{p/2}v^{(p-2)/2}>1$ where $c_p$ is small enough constant related to $p$. Therefore taking a random vector in the span of leading $\delta^{-O(1)}$ eigenvectors of $D(Y)=\frac{1}{L}\sum_{t=1}^L p_{c_t}(Y)$ generates an estimator achieving weak recovery. Since $\ell=O(\log_\gamma n)$, the polynomial can be evaluated in time $O(n^{2pv})$ when we have $c_p\lambda^2 n^{p/2}v^{(p-2)/2}>1$ and $n/v^2=\omega(\Gamma\ell)$. This leads to polynomial time algorithm when $v$ is constant. 
    
    \end{proof}

\end{document}